\documentclass[reprint, amsfonts, amssymb, amsmath,  showkeys, pra, superscriptaddress, nofootinbib, twocolumn, longbibliography]{revtex4-2}
\usepackage[english]{babel}
\usepackage{amsthm}
\usepackage[toc,page]{appendix}
\usepackage[colorlinks=true,citecolor=blue,linkcolor=magenta]{hyperref}
\usepackage{bbm}
\usepackage{graphicx}
\usepackage{stackengine,xcolor}
\usepackage{mathrsfs}  
\usepackage{float}
\makeatletter
\let\newfloat\newfloat@ltx
\makeatother
\usepackage{algorithm}
\usepackage{algpseudocode}

\newcommand{\Tr}{{\rm Tr}}

\newcommand{\g}{\mathfrak{g}}

\newcommand{\so}{\mathfrak{so}}

\renewcommand{\vec}[1]{\boldsymbol{#1}} 
\newcommand{\mf}[1]{\mathfrak{#1}}

\newcommand{\thv}{\vec{\theta}}
\newcommand{\id}{\openone}
\newcommand{\ad}{^{\dagger}}
\newcommand{\ket}[1]{|#1\rangle}

\newcommand{\AC}{\mathcal{A}}

\newcommand{\GC}{\mathcal{G}}
\newcommand{\HC}{\mathcal{H}}

\newcommand{\LC}{\mathcal{L}}
\newcommand{\MC}{\mathcal{M}}
\newcommand{\NC}{\mathcal{N}}
\newcommand{\OC}{\mathcal{O}}
\newcommand{\PC}{\mathcal{P}}

\newcommand{\TC}{\mathcal{T}}

\newcommand{\SPBB}{\mathbb{SP}}
\newcommand{\SO}{\mathbb{SO}}
\newcommand{\U}{\mathbb{U}}
\def\tn{^{\otimes n}}
\newcommand{\mbb}[1]{\mathbb{#1}}
\def\w{\omega}
\newcommand{\arr}{\xrightarrow[]{}}
\def\lm{\lambda}
\newcommand{\ot}[1]{^{\otimes #1}}

\newtheorem{theorem}{Theorem}
\newtheorem{lemma}{Lemma}

\newtheorem{definition}{Definition}

\begin{document}

\title{Optimal Haar random fermionic linear optics circuits}

\author{Paolo Braccia}
\affiliation{Theoretical Division, Los Alamos National Laboratory, Los Alamos, New Mexico 87545, USA}

\author{N. L. Diaz}
\affiliation{Information Sciences, Los Alamos National Laboratory, Los Alamos, New Mexico 87545, USA}
\affiliation{Center for Nonlinear Studies, Los Alamos National Laboratory, Los Alamos, New Mexico 87545, USA}
\affiliation{Departamento de F\'isica-IFLP/CONICET, Universidad Nacional de La Plata, C.C. 67, La Plata 1900, Argentina}

\author{Martin Larocca}
\affiliation{Theoretical Division, Los Alamos National Laboratory, Los Alamos, New Mexico 87545, USA}

\author{M. Cerezo}
\thanks{cerezo@lanl.gov}
\affiliation{Information Sciences, Los Alamos National Laboratory, Los Alamos, New Mexico 87545, USA}

\author{Diego Garc\'ia-Mart\'in}
\affiliation{Information Sciences, Los Alamos National Laboratory, Los Alamos, New Mexico 87545, USA}

\begin{abstract}
Sampling unitary Fermionic Linear Optics (FLO), or matchgate circuits, has become a fundamental tool in quantum information. Such capability enables a large number of applications ranging from randomized benchmarking of continuous gate sets, to fermionic classical shadows. In this work, we introduce optimal algorithms to sample over the non-particle-preserving (active) and particle-preserving (passive) FLO Haar measures. In particular, we provide appropriate distributions for the gates of $n$-qubit parametrized circuits which produce random active and passive FLO. In contrast to previous approaches, which either incur  classical $\mathcal{O}(n^3)$ compilation costs or have suboptimal depths, our methods directly output circuits which \textit{simultaneously} achieve an optimal down-to-the-constant-factor $\Theta(n)$ depth and $\Theta(n^2)$  gate count; with only a $\Theta(n^2)$ classical overhead.  Finally, we also provide quantum circuits to sample Clifford FLO with an optimal $\Theta(n^2)$  gate count. 
\end{abstract}
\maketitle

\section{Introduction}

Unitary Fermionic Linear Optics (FLO)~\cite{valiant2001quantum,knill2001fermionic,terhal2002classical,divincenzo2005fermionic} are widely used in a large number of disciplines, such as quantum chemistry~\cite{kivlichan2018quantum,arute2020hartree,arrazola2022universal}, condensed-matter and many-body physics~\cite{verstraete2009quantum,kraus2011compressed,cervera2018exact,jiang2018quantum,dallaire2019low,sopena2022algebraic,ruiz2024bethe,ruiz2024efficient}, or quantum information and computation~\cite{bravyi2004lagrangian,jozsa2008matchgates,brod2011extending,brod2014computational,brod2016efficient,helsen2022matchgate,wan2022matchgate,diaz2023showcasing, diaz2023parallel,mele2024efficient,kokcu2022fixed,kokcu2022algebraic}. They are defined as the set of unitary transformations that preserve the fermionic Majorana operators (which are analogous to position and momentum for bosons) under conjugation~\cite{knill2001fermionic}, and they have been subject to numerous investigations~\cite{valiant2001quantum,knill2001fermionic,terhal2002classical,divincenzo2005fermionic,bravyi2004lagrangian,jozsa2008matchgates,brod2011extending,brod2014computational,brod2016efficient,guaita2024representation,oszmaniec2022fermion,helsen2022matchgate,wan2022matchgate,diaz2023showcasing, mele2024efficient}. For instance, it is well known that they  constitute a restricted model of computation that is efficiently classically simulable, for certain input states and measurements~\cite{valiant2001quantum,knill2001fermionic,terhal2002classical,divincenzo2005fermionic,somma2006efficient,jozsa2008matchgates,brod2011extending,brod2014computational,brod2016efficient,goh2023lie,miller2025simulation}. 
They have also been shown to be equivalent to matchgate circuits and fermionic Gaussian unitaries generated by free-fermionic Hamiltonians~\cite{knill2001fermionic,terhal2002classical,divincenzo2005fermionic}. Furthermore, modifying these circuit by enhanced connectivity (e.g., via SWAP gates)~\cite{jozsa2008matchgates,brod2014computational} or through the  addition of certain non-FLO gates~\cite{brod2011extending,oszmaniec2017universal}  makes them universal for quantum computation. Examples of applications of FLO include random circuit sampling~\cite{oszmaniec2022fermion} (for which an exponential quantum speed-up can be achieved), randomized benchmarking of continuous gate sets~\cite{helsen2022matchgate}, or fermionic classical shadows~\cite{zhao2021fermionic,wan2022matchgate,heyraud2024unified}. A common theme in many such applications is the need to randomly sample FLO from the associated Haar (or uniform) measure, or at least from an ensemble that reproduces their first $t$ moments, i.e., a FLO $t$-design.  This holds true for both non-particle-preserving (active) and particle-preserving (passive) FLO.

The usual strategy for sampling FLO makes use of the well-known fact that the adjoint action of active FLO is isomorphic to the special orthogonal group $\SO(2n)$ (or the unitary group $\U(n)$, for passive FLO), where $n$ is the number of fermionic modes, or equivalently, the number of qubits under the Jordan-Wigner transformation. Here, one samples a $2n\times 2n$ special orthogonal matrix, and then compiles it into a matchgate circuit. In particular, to obtain a Haar random orthogonal matrix, the entries of the matrix are independently drawn from a Gaussian distribution with zero mean and unit variance, and the columns orthonormalized using the Gram-Schmidt algorithm (alternatively, a QR decomposition can be performed with a fixed gauge)~\cite{mezzadri2006generate}. The classical computational cost of performing the Gram-Schmidt orthonormalization, or the QR decomposition, is $\OC(n^3)$~\footnote{The algorithm described produces a  random orthogonal matrix. In order to obtain a Haar random \textit{special} orthogonal matrix, one can check whether the determinant is $\pm 1$, and exchange two columns if it is $-1$. Such procedure can be realized in practice in $\OC(n^3)$ time.}. On top of that, the classical computational cost of compiling the sampled matrix into a quantum circuit, which is also $\OC(n^3)$, needs to be taken into account. This is discussed in detail in the Supplemental Information (SI).

Instead of following the previous strategy, it is more convenient to sample the circuits directly on quantum hardware. That is, to employ a quantum circuit architecture with   $\Theta(n^2)$ one- and two-qubit parametrized matchgates (acting on nearest neighbors on a one-dimensional qubit chain), such that carefully sampling the parameters of these gates delivers unitaries from the Haar measure, see Fig.~\ref{fig:1}. In this way, the classical computation overhead is just that of sampling the $\Theta(n^2)$ parameters that are required to fully characterize FLO, which is optimal. 

Our main contribution is to actually find such constructions for both active and passive FLO. Specifically, we provide quantum circuits and associated parameter probability distributions to directly sample Haar random FLO unitaries. Crucially, our circuits \emph{simultaneously} achieve  an optimal down-to-the-constant-factor $\Theta(n)$ depth, $\Theta(n^2)$  gate count and $\Theta(n^2)$ classical computation overhead. This contrasts with previous approaches, that either produce circuits with a triangular shape that result in suboptimal  circuit depths and idling qubits~\cite{helsen2022matchgate,heyraud2024unified}, or require the use of classical compilation techniques (with an associated $\OC(n^3)$ computational cost)~\cite{oszmaniec2022fermion}.

At the technical level, we start with parametrizations of the classical compact Lie groups $\SO(2n)$ and $\U(n)$, together with the corresponding invariant measures. These were originally found by Hurwitz in 1897~\cite{diaconis2017hurwitz}, and introduced to quantum information theory in Ref.~\cite{zyczkowski1994random}. Directly translating these constructions into the realm of quantum circuits (which was done for active FLO in~\cite{helsen2022matchgate,heyraud2024unified}) produces circuits with a triangular structure, and hence a suboptimal  depth. Here, we develop a set of rules that allow us to find Haar random FLO circuits with a brick-wall structure and, consequently, an optimal depth, while at the same time avoiding the cost of classical compilation.

Furthermore, we find a simple algorithm to uniformly sample the group of (active) Clifford FLO, that is, the subgroup of matchgate circuits that are also Clifford unitaries~\footnote{We recall that the adjoint action of this group is isomorphic to $B_{2n}\cap \SO(2n)$, where $B_{2n}$ denotes the hyperoctahedral group.}~\cite{gottesman1998heisenbergrepresentation,aaronson2004improved}. Clifford FLO are relevant because they were shown to form a $3$-design over the FLO Haar measure~\cite{wan2022matchgate,heyraud2024unified} (although not a $4$-design as proven here as a bonus result).
In practice, the previous implies that for many applications~\cite{oszmaniec2022fermion,wan2022matchgate}, it suffices to use random Clifford FLO circuits instead of Haar random ones. Our sampling algorithm avoids the cost of compiling random signed permutation matrices (which are matrices in $\SO(2n)$ corresponding to Clifford FLO transformations) into  quantum circuits, by directly sampling Clifford angles in a parametrized matchgate circuit architecture. Importantly, we obtain quantum circuits with a provably optimal average number of gates, although with a triangular shape. Our random Clifford FLO circuits can thus be used in practice, possibly after compressing them once they have been sampled~\cite{kokcu2022algebraic,camps2022algebraic}.

\begin{figure}[t]
    \centering
    \includegraphics[width=1\linewidth]{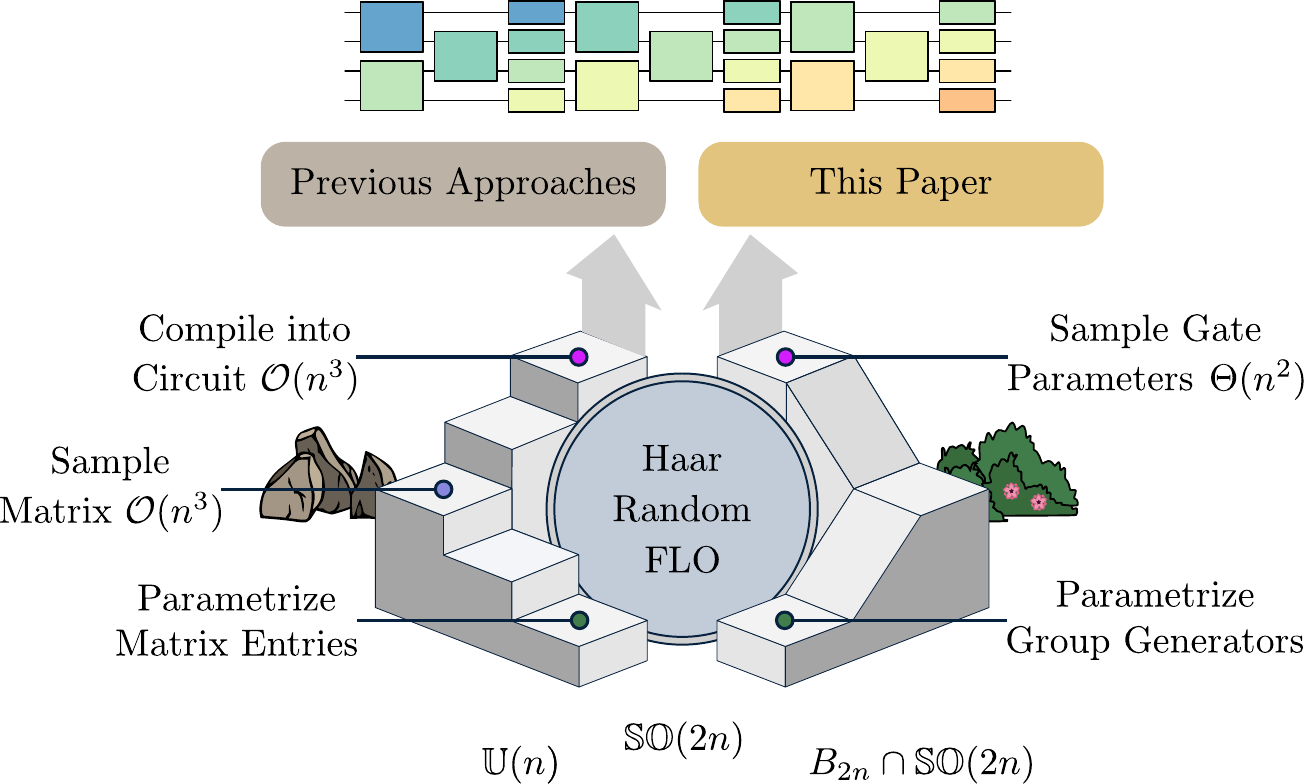}
    \caption{{\bf Schematic representation of our results}. Our goal is to find optimal circuits which implement Haar random active, passive, or Clifford FLO unitaries. Previous methods  are based on parametrizing the entries of matrices from the  $\SO(2n)$, $\U(n)$ and $B_{2n}\cap \SO(2n)$, respectively, to later sample $2n\times 2n$ matrices according to the Haar measure, and then compile these into qubit-based circuits. In this work we propose a more efficient method which instead parametrizes the group generators and only requires sampling gate parameters. }
    \label{fig:1}
\end{figure}

\section{Results}

\subsection{Haar measure on active FLO circuits}

We start by considering active FLO, which can be implemented on an $n$-qubit quantum computer as parametrized matchgate circuits acting on a one-dimensional array of qubits with open boundary conditions~\cite{jozsa2008matchgates,diaz2023showcasing}. In particular, these circuits are generated  as
\begin{equation} \label{eq:fermion-circuits}
    U(\thv)=\prod_{H_j\in \GC}e^{i\theta H_j}\,,
\end{equation}
where $\thv=(\theta_1,\theta_2,\dots)$ is a vector of real parameters and where the gate generators are taken from the  set $\GC=\{Z_q\}_{q=1}^n\cup\{X_q X_{q+1}\}_{q=1}^{n-1}$. Here,  $X_q$, $Z_q$ denote the single-qubit Pauli operators acting on the $q$-th qubit. Quantum circuits of the form~\eqref{eq:fermion-circuits} preserve fermionic parity, which in the qubit picture corresponds to the operator $Z^{\otimes n}$. When no further restriction is imposed on FLO, they are called active FLO (as opposed to passive FLO, which in addition preserve the fermionic particle-number $\frac{1}{2}\sum_{q=1}^n (\id-Z_q)$).

Our first main result is a construction that delivers optimal Haar random active FLO circuits, bypassing the need for classical compilation. In particular, we show that the following theorem holds (see Methods and SI~\ref{ap:theo-1} for a detailed proof).

\begin{figure*}[t]
    \centering
    \includegraphics[width=\linewidth]{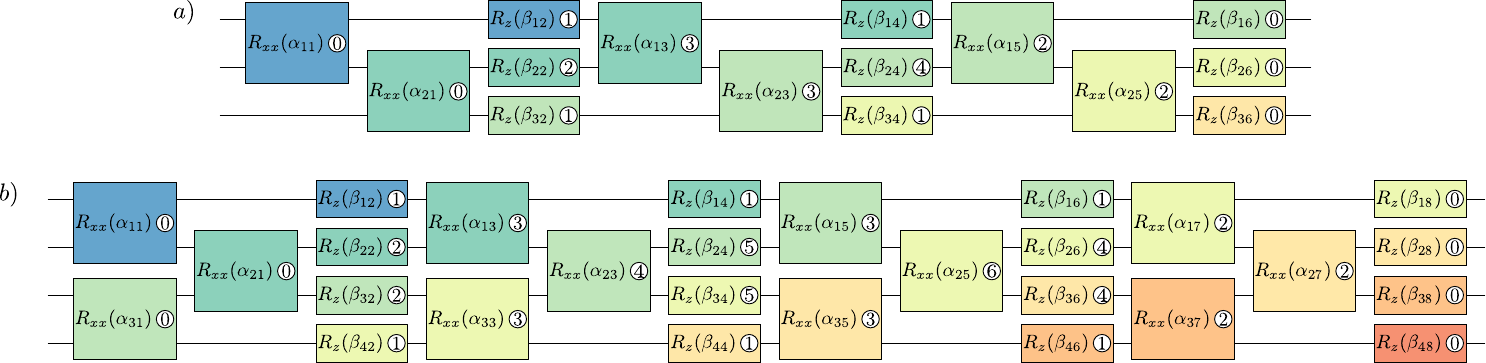}
    \caption{{\bf Optimal Haar random active FLO circuits}.  These quantum circuits produce Haar random active FLO on $n=3$ (a) and $n=4$ qubits (b), respectively, when the parameters are sampled according to Eq.~\eqref{eq:active_haar},~\eqref{eq:active_haar_layer}, and~\eqref{eq:f_function}.  
    Gates with the same color belong to the same ``anti-diagonal'', which can be used to understand the arrangement of the probability density functions within the circuits. Namely, how the powers of the sine functions that the gates' parameters are sampled from are distributed, as marked by the circled numbers next to their respective gate. The pattern is the following: we start at zero (i.e., uniform distribution) on one end of the anti-diagonal —at the left end if the anti-diagonal lies above the main (longest) anti-diagonal, and at the right end otherwise—, and then increment by one, alternating sides as we move toward the center.}
    \label{fig:flo-circuits}
\end{figure*}

\begin{theorem} \label{th:fermion-haar}
    Any fermionic linear optics circuit can be decomposed (up to a global minus sign) as 
    \begin{equation} \label{eq:spin-matrix-decomposition}
            U(\thv)= L_{2n} L_{2n-1}\cdots L_1\,,
    \end{equation}
    where 
    \begin{align}\label{eq:R-prod}
        &L_{2k-1} = \prod_{j=1}^{n-1} e^{i\alpha_{jk} X_j X_{j+1}}\,,\; L_{2k}= \prod_{j=1}^{n} e^{i \beta_{jk} Z_j}\,,
    \end{align}
    for $1<k<n$. The parameters take values  $\alpha_{j,1}\in [0,2\pi)$, $\beta_{j,n}\in [0,2\pi)$ and $\alpha_{jk}, \in [0,\pi]$, $\beta_{j,k}\in [0,\pi]$.
    
    Moreover, the normalized Haar measure for the adjoint representation with respect to the parametrization in Eqs.~\eqref{eq:spin-matrix-decomposition} and~\eqref{eq:R-prod} is given by
     \begin{equation}\label{eq:active_haar}
        d\mu(U)=\NC \prod_{k=1}^n \mu(L_{2k} L_{2k-1})\prod d\alpha\,d\beta\,,
    \end{equation}
    where $\NC=\left((2\pi)^{2n-1} \prod_{s=3}^{2n}\frac{\sqrt{\pi}^{\,s-2}}{\Gamma(s/2)}\right)^{-1}$,
    \small
    \begin{equation}\label{eq:active_haar_layer}
       \mu(L_{2k} L_{2k-1})= \prod_{j=1}^{n-1}\sin(\alpha_{jk})^{f_n(2j,2k-1)}\prod_{j=1}^{n}\sin(\beta_{jk})^{f_n(2j-1,2k)}\,,
    \end{equation}
    \normalsize
     and
    \begin{equation} \label{eq:f_function}
        f_n(u,v)=
    \begin{cases}
        \min(2v-2, 4n-2u-1) \quad &{\rm if} \quad u>v, \\
        \min(4n-2v, 2u-1) \quad &{\rm if} \quad u<v \\
    \end{cases}\,.
    \end{equation}
\end{theorem}

We now interpret the result in Theorem~\ref{th:fermion-haar}. This theorem gives us a general way of constructing arbitrary FLO circuits, up to an unobservable global sign. In fact, the most general circuit has $n(2n-1)$ single- and two-qubit gates, which coincides with the dimension of the Lie algebra $\so(2n)$. Thus, this decomposition is optimal in terms of the total number of local quantum gates.
The gates are arranged in $2n$ layers in a brick-wall fashion, half of them consisting of $n-1$ two-qubit $R_{xx}(\alpha_{jk})$  rotations, and the other half of $n$  single-qubit $R_z(\beta_{jk})$ ones, as per Eqs.~\eqref{eq:spin-matrix-decomposition} and~\eqref{eq:R-prod}. Here, the index $k$ is used to label the layers and $j$ the position of a gate within the layer.
As an example, we show explicitly the cases for $n=3$ and $n=4$ in Fig.~\ref{fig:flo-circuits}. Importantly, the construction in Theorem~\ref{th:fermion-haar} only involves one- and two-qubit gates acting on a one-dimensional array of qubits (i.e., it is exactly of the form in Eq.~\eqref{eq:fermion-circuits}), which render this architecture readily implementable on near-term quantum hardware~\cite{foxen2020demonstrating,arute2020hartree}. Furthermore, the depth of the circuits is linear in $n$, scaling as $3n$ (as opposed to $6n-7$ for circuits with a triangular structure~\cite{helsen2022matchgate,heyraud2024unified}). This is optimal as one cannot further compress the $n(2n-1)$ gates that are necessary to obtain Haar random FLO. We nevertheless note that \emph{after} one has sampled a random FLO circuit, there might be instance-dependent simplifications that could allow for further depth reduction.

Finally,  to sample circuits from the Haar measure, one has to independently sample each of the $n(2n-1)$ parameters according to Eqs.~\eqref{eq:active_haar}, ~\eqref{eq:active_haar_layer} and~\eqref{eq:f_function}. 
In Fig.~\ref{fig:pdfs} we show the Probability Density Functions (PDFs) for the parameters in a $100$-qubit circuit.
The first thing to notice is that the PDFs are either uniform or given by increasingly large powers of the sine function, up to $\sin(\theta)^{2n-2}$. The only angles that are sampled uniformly (ranging in $[0,2\pi)$ as opposed to $[0,\pi]$ for the rest) are those in the first and last layers of the circuit, i.e., those of the form $\alpha_{j,1}$ and $\beta_{j,n}$. Then, to understand how all the PDFs are located within the circuit, one can look at the ``anti-diagonals'' depicted in Fig.~\ref{fig:flo-circuits}. The powers increase along each anti-diagonal, alternating between the head and the tail of the anti-diagonal (see Fig.~\ref{fig:flo-circuits}). Moreover, whether the left-most end of an anti-diagonal is considered its head or its tail depends on whether or not it is located above the main anti-diagonal. This behavior is precisely captured by the $f_n$ function in Eq.~\eqref{eq:f_function}. 
In Fig.~\ref{fig:microwave}, we show how the powers of the sine functions appearing in the PDFs are distributed within a $50$-qubit circuit, giving rise to a distinctive pattern where the powers grow larger as one moves towards the center of the circuit.  

At this point, it is interesting to notice that as $n$ increases, the majority of the PDFs get concentrated around $\pi/2$, implying that the sampled gates become closer and closer to the Clifford gates $iZ$ or $iX X$, as shown in Fig.~\ref{fig:pdfs}. Hence, the larger the number of qubits, the more the core of the circuit will resemble a Clifford one. 
In the SI~\ref{ap:sampling}, we explain how to sample from highly peaked distributions with a constant computational cost in the power $f_n$, as care must actually be taken to prevent this cost from scaling with $f_n$.

 \begin{figure}[t]
    \centering
    \includegraphics[width=\linewidth]{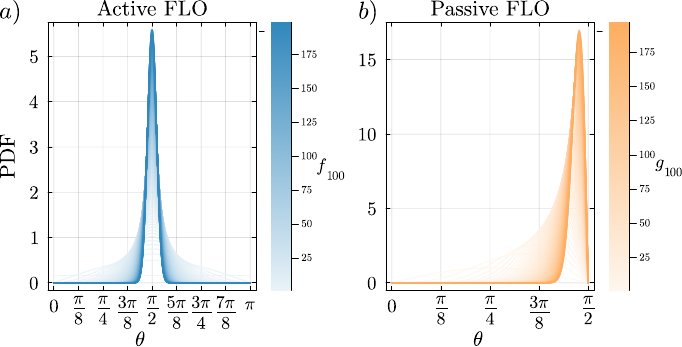}
    \caption{{\bf Parameter probability density functions}. The PDFs for the angles $\alpha_{j,k}$, $\beta_{j,k}$ and $\theta_{j,k}$ for Haar random active (a) and passive (b) FLO circuits, respectively, as a function of $f_{100}$ and $g_{100}$. As $f_{100}$ and $g_{100}$ increase, the angles concentrate around and towards $\pi/2$, which correspond to Clifford gates.}

    \label{fig:pdfs}
\end{figure}

\subsection{Haar measure on passive FLO circuits}

We now consider passive FLO. As explained before,  FLO preserve fermionic parity ($Z^{\otimes n}$ in Jordan-Wigner form), but not necessarily the particle number $\frac{1}{2}\sum_{q=1}^n (\id-Z_q)$. The subgroup of FLO that preserve both parity and particle number is known as passive FLO~\cite{oszmaniec2017universal}, and its adjoint action on the space of Majorana operators is isomorphic to the matrix Lie group $\SO(2n) \cap \mathbb{SP}(2n,\mathbb{R})$, which in turn is isomorphic to $\U(n) $. Here, $\mathbb{SP}(2n,\mathbb{R})$ denotes the group of $2n\times 2n$ symplectic matrices with real entries, and we recall that  a symplectic matrix $S$ is such that $S^T\Omega S=\Omega\equiv\id_n \otimes iY$.  Using the previous isomorphisms (see Methods and SI~\ref{ap:theo-2}), we can find the following set of matchgate generators for passive FLO,
\begin{equation}\label{eq:generators}
    \GC =  \{Z_q\}_{q=1}^n \cup \{X_qY_{q+1}-Y_q X_{q+1}\}_{q=1}^{n-1} \,.
\end{equation}
That is, any passive FLO can be written as in Eq.~\eqref{eq:fermion-circuits} using operators from $\GC$ (see, e.g., Fig.~\ref{fig:passive-flo-circuits}). 
In Fig.~\ref{fig:xy-yx} we show a decomposition of the $R_{xy-yx}^{j,j+1}(\theta)= e^{i\theta (X_jY_{j+1}-Y_{j}X_{j+1})/2}$ gate in terms of $R_z$ and $R_{xx}$ gates, such that both active and passive FLO  can be generated employing the exact same native gate set, which might be a desirable feature when it comes to actually implementing the circuits on quantum hardware.

\begin{figure}[t]
    \centering
    \includegraphics[width=\linewidth]{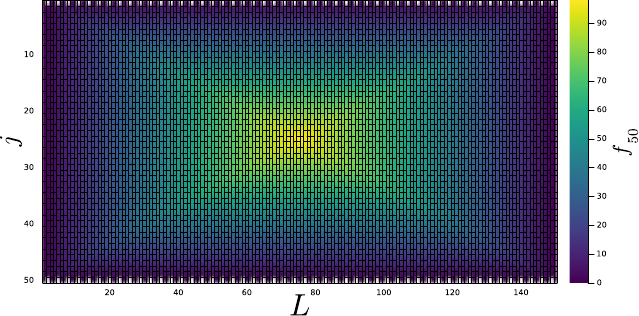}
    \caption{{\bf Structure of Haar random active FLO circuits.} Distribution of the powers of sine functions ($f_{50}$ in Eq.~\eqref{eq:f_function}) that have to be sampled in a Haar random active FLO circuit with $n = 50$ qubits, according to the architecture in Theorem~\ref{th:fermion-haar}. Each gate is shown as a rectangle, with larger blocks corresponding to $R_{xx}$ rotations and smaller ones to $R_z$ rotations. Colors encode the values of $f_{50}$; the horizontal axis denotes the circuit depth, and the vertical axis labels the qubits. 
    The white dents on the top and bottom are due to the absence of gates at those locations.}
    \label{fig:microwave}
\end{figure}

\begin{figure*}[t]
    \centering
    
    \includegraphics[width=\linewidth
    ]{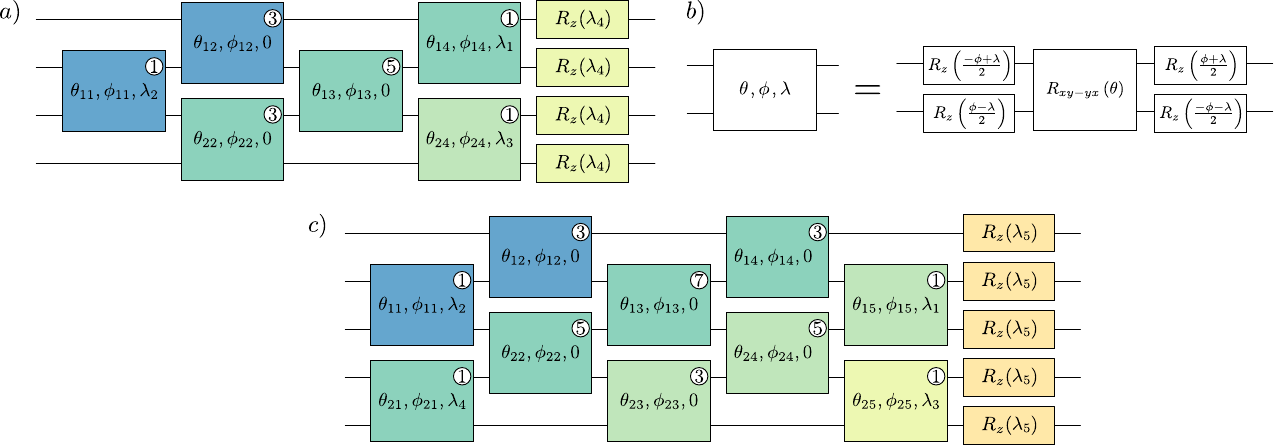}
    \caption{{\bf Optimal Haar random passive FLO circuits}. These quantum circuits produce Haar random passive FLO on $n=4$ (a) and $n=5$ qubits (c), respectively, when the parameters are sampled according to Eq.~\eqref{eq:passive-FLO-haar},~\eqref{eq:passive-FLO-haar-layer}, and~\eqref{eq:g_function}.
    Gates with the same color belong to the same ``anti-diagonal'', which can be used to understand the arrangement of the probability density functions within the circuits. Namely, how the odd powers of the sine functions in Eq.~\eqref{eq:passive-FLO-haar-layer} for the $\theta$ angles are distributed, as marked by the circled numbers next to their respective gate. The pattern, closely resembling that of Haar random active FLO circuits, is the following: we start at one end of the anti-diagonal —at the left end if the anti-diagonal lies above the main (longest) anti-diagonal, and at the right end if it lies below—, and then increment by two, alternating sides as we move toward the center. The main anti-diagonal follows the pattern of those above it when $n$ is odd, and of those below it otherwise. The last layer of  $R_z$ rotations does not participate in this pattern and the corresponding parameter is sampled uniformly, analogously to all the other phases $\phi, \lambda$ in each $U_{\tilde{R}_j}$. In addition, in (b) we depict a $U_{\tilde{R}_j}(\theta, \phi, \lambda)$ gate, as defined in Eq.~\eqref{eq:U-R-tilde}.}
    \label{fig:passive-flo-circuits}
\end{figure*}

As for the case of active FLO, we find optimal Haar random passive FLO circuits. Specifically, we prove the following theorem (see the SI~\ref{ap:theo-2} for a detailed proof).

\begin{theorem} \label{th:passive-fermion-haar}
Any passive fermionic linear optics circuit can be decomposed (up to a global minus sign) as 
\begin{equation} \label{eq:passive-flo-matrix-decomposition}
        U(\thv)= e^{i\frac{\lambda_n}{2}\sum_{q=1}^n Z_q} \tilde{L}_n \tilde{L}_{n-1} \cdots \tilde{L}_1\,,
\end{equation}
where 
\small
\begin{align}\label{eq:R-tilde-prod}
    \nonumber
    \tilde{L}_{2k} &= \prod_{j=1}^{\lfloor \frac{n}{2}\rfloor} U_{\tilde{R}_{2j-1}}(\theta_{j,2k},\phi_{j,2k},\lambda_{2j-1}\delta_{2k, n})\,,\\
    \nonumber
    \tilde{L}_{2k+1} &= \prod_{j=1}^{\lfloor \frac{n-1}{2}\rfloor} U_{\tilde{R}_{2j}}(\theta_{j,2k+1},\phi_{j,2k+1},\lambda_{2j-1}\delta_{2k+1, n}+\lambda_{2j}\delta_{2k+1,1})\,, \\
\end{align} 
\normalsize
with  $k\in 0,1,\dots,\left\lfloor\frac{n}{2}\right\rfloor$, $\tilde{L}_0\equiv\id$, and
    \begin{align} \label{eq:U-R-tilde}
       &U_{\tilde{R}_j}(\theta,\phi,\lambda)= \nonumber\\  &\;e^{i\frac{\phi+\lambda}{4} Z_j} e^{i\frac{-\phi-\lambda}{4} Z_{j+1}} e^{i\frac{\theta}{2} (X_jY_{j+1}-Y_j X_{j+1})} e^{i\frac{-\phi+\lambda}{4} Z_j} e^{i\frac{\phi-\lambda}{4} Z_{j+1}} \,\!.
    \end{align}
    The parameters  take values $\theta\in[0,\frac{\pi}{2}]$, $\phi\in[0,2\pi)$, and $\lambda\in[0,2\pi)$.

    In addition, the normalized Haar measure for the adjoint representation with respect to the parametrization in Eqs.~\eqref{eq:passive-flo-matrix-decomposition} and~\eqref{eq:R-tilde-prod} is given by
    \begin{equation} \label{eq:passive-FLO-haar}
        d\mu(U)= \NC \prod_{k=0}^{\lfloor \frac{n}{2}\rfloor} \mu(\tilde{L}_{2k}) \mu(\tilde{L}_{2k+1})\prod d\theta \,d\phi\, d\lambda\,, 
    \end{equation}
    where  $\NC=\left(\sqrt{2\pi}^{n(n+1)}\prod_{s=1}^{n-1}(2s)^{s-n}\right)^{-1}$,
    \begin{align}\label{eq:passive-FLO-haar-layer}
        \mu(\tilde{L}_{2k}) &= \prod_{j} \cos( \theta_{j,2k})\sin(\theta_{j,2k})^{g_n(2j-1,2k)}   \nonumber\\ 
        \mu(\tilde{L}_{2k+1}) &= \prod_{j} \cos( \theta_{j,2k+1})\sin(\theta_{j,2k+1})^{g_n(2j,2k+1)} \,,
    \end{align}
    with $\mu(\tilde{L}_{0})\equiv1$, and
    \begin{equation} \label{eq:g_function}
        g_n(u,v)=
    \begin{cases}
        \min(4v-3, 4n-4u-1) \quad &{\rm if} \quad u>v, \\
        \min(4n-4v+1, 4u-1) \quad &{\rm if} \quad u<v \\
    \end{cases}\,.
    \end{equation}
\end{theorem}

Theorem~\ref{th:passive-fermion-haar} provides a general architecture that can produce arbitrary passive FLO circuits in terms of one- and two-qubit gates acting on an open chain of qubits. The number of independent parameters is $n^2$, which equals the dimension of $\mathfrak{u}(n)$, and therefore is optimal. 
The gates in the circuits are arranged in a brick-wall fashion, as depicted in Fig.~\ref{fig:passive-flo-circuits}, which implies that the depth is also optimal.
There are $n$ layers, $\tilde{L}_{1}$ through $\tilde{L}_n$, with  an additional final layer of $R_z(\lambda_n)$ rotations on every qubit with the same rotation angle $\lambda_n$. 
Each layer $\tilde{L}$ consists of $\lfloor \frac{n}{2}\rfloor$ or $\lfloor \frac{n-1}{2}\rfloor$  two-qubit unitaries $U_{\tilde{R}_j}$, depending on whether the first gate of the layer is  $U_{\tilde{R}_2}$ or $U_{\tilde{R}_1}$, respectively.
 In turn, each $U_{\tilde{R}_j}$ acts on qubits $j,\,j+1$ and is decomposed as in Eq.~\eqref{eq:U-R-tilde} and Fig.~\ref{fig:xy-yx}, i.e., each $U_{\tilde{R}_j}$ consists of four $R_z$ and one $R_{xy-yx}(\theta)$ rotations.

We stress that contiguous $R_z$ rotations can (and should) be compiled into a single gate by simply adding the angles, and that we keep them separate in Fig.~\ref{fig:passive-flo-circuits} for the sake of explanation. Even more importantly, the $R_z$ gates can be implemented ``virtually'' in some platforms like superconducting qubits~\cite{mckay2017efficient}, which means that they are essentially errorless, and that the depth only corresponds to that arising from the two-qubit gates.

In order to obtain the Haar measure over passive FLO, one must sample the parameters in the circuit following a specific recipe. Namely, the parameters must be independently sampled according to Eqs.~\eqref{eq:passive-FLO-haar}, ~\eqref{eq:passive-FLO-haar-layer} and~\eqref{eq:g_function}. As such, the parameters $\phi$ and $\lambda$ which go into the $R_z$ rotations are all uniformly sampled in the interval $[0,2\pi)$. 
Moreover, the PDFs for the angles $\theta$ of the $R_{xy-yx}$ gates are given by increasingly larger odd powers of the sine function, modulated by a cosine. The pattern of how the PDFs are arranged within the circuit is completely analogous to that of active FLO (see Fig.~\ref{fig:microwave}), except for the fact that all the sine powers are odd. Hence, one must look at the ``anti-diagonals'' in the circuit to find out where these powers are located, as graphically shown in Fig.~\ref{fig:passive-flo-circuits}.
Again, as $n$ increases, the majority of the angles concentrate towards $\pi/2$, i.e., Clifford gates (as per Fig.~\ref{fig:pdfs}).

\subsection{Haar measure on Clifford FLO circuits}

The group of Clifford circuits consists of all the unitaries that normalize the Pauli group~\cite{gottesman1998heisenbergrepresentation}. That is, a Clifford unitary is such that its adjoint action on a Pauli string returns a Pauli string, up to a $\pm1,\,\pm i$ phase.  Importantly, this group plays a central role in  quantum error correction~\cite{bennett1996mixed,calderbank1997quantum}, classical shadows~\cite{huang2020predicting,west2024real}, or the classical simulation of quantum circuits~\cite{gottesman1998heisenbergrepresentation,aaronson2004improved}.

Here, we are interested in the intersection of the Clifford group and FLO, or in other words, in FLO circuits that are also Clifford.  As explained in the SI, these circuits are a representation of the matrices $O\in\SO(2n)$ that are signed permutations, because a Majorana operator $c_l$ (which is a Pauli string) must be mapped
to a Majorana operator $\pm c_m$, see Eq.~\eqref{eq:action-majorana}. The corresponding finite subgroup of  $\mathbb{SO}(2n)$  of signed permutation matrices of $2n$ items (with unit determinant) is a subgroup of the hyperoctahedral group $B_{2n}$, and contains $2^{2n-1} (2n)!$ elements.

The relevance of Clifford FLO stems from the fact that they form a $3$-design over FLO~\cite{wan2022matchgate,heyraud2024unified}. This entails important consequences, since in many applications--like fermionic classical shadows~\cite{wan2022matchgate} or fermionic random sampling~\cite{oszmaniec2022fermion}--it suffices to sample  Clifford FLO instead of sampling from the Haar measure over FLO. In addition, in the SI~\ref{ap:4-design}  we prove that the following result holds. 

\begin{theorem} \label{th:clifford}
    The Clifford FLO group does not form a FLO $4$-design.
\end{theorem}

We recall that this property parallels the one exhibited by Clifford circuits, which form a $3$-design over the unitary group, but not a $4$-design~\cite{webb2016clifford,zhu2017multiqubit,kueng2015qubit,zhu2016clifford}. Similarly, other subgroups of the Clifford group are also $3$-designs (but not $4$-designs) over subgroups of the unitary group~\cite{hashagen2018real,mitsuhashi2023clifford}.

\begin{figure}[t]
    \centering
    \includegraphics[width=.8\linewidth]{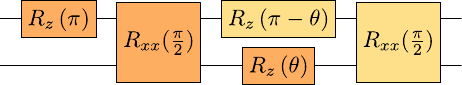}
    \caption{{\bf Decomposition of the $R_{xy-yx}(\theta)$ passive FLO gate}. We show how to implement an $R_{xy-yx}(\theta)$ rotation in terms of $R_z$ and $R_{xx}$ gates (up to an unobservable global sign). This decomposition was obtained following Eqs.~\eqref{eq:so-matrix-decomposition} and~\eqref{eq:R-O-tilde-prod} as indicated by the colors in the ladder-like decomposition.}
    \label{fig:xy-yx}
\end{figure}

Given the practical importance of random Clifford FLO, we now provide a simple algorithm to uniformly sample Clifford FLO circuits directly on quantum hardware.  The algorithm uses the triangular-shape circuit decomposition of Eqs.~\eqref{eq-ap:spin-matrix-decomposition} and~\eqref{eq-ap:R-prod} (see Supp. Fig.~\ref{fig-ap:so-circuits} for clarity), and is as follows:

\begin{algorithm}
\caption{Random Clifford FLO sampling}\label{alg:clifford}
\begin{algorithmic}[1]
\State \textbf{Input}: \texttt{Initial angles $\theta_{j,k}=0$ for all $j,k$}
 \For{$k$ in $2n,\dots,2$}
  \State \texttt{Uniformly sample a random integer $l$ in $[k]$}
  \For{$j\geq l$}
  \State \texttt{ $\theta_{j,k}\leftarrow\pi/2$}
  \EndFor
  \State \texttt{Set $\theta_{l,k}\leftarrow\theta_{l,k}$ or $\theta_{l,k}\leftarrow\theta_{l,k}+\pi$ at random with equal probability}
  \If{l=k}:
  \State \texttt{Set $\theta_{l-1,k}\leftarrow\theta_{l-1,k}$ or $\theta_{l-1,k}\leftarrow\theta_{l-1,k}+\pi$ at random with equal probability}
  \EndIf
  \EndFor
 \State \Return \texttt{angles}
\end{algorithmic}
\end{algorithm}

In simple terms, Algorithm~\ref{alg:clifford} goes over the layers $C_{2n-1}$ through $C_1$ (which have a ``ladder-like'' structure), and for each layer it picks a random integer $l$ from $1$ to $k$, where $k$ is the number of gates in the layer plus one. Then, all the angles $\theta_{j,k}$ corresponding to gates  acting on the $j$-th position within the layer, for $j\geq l$,  are set to $\pi/2$. Additionally, the angle $\theta_{l,k}$ for the gate in the $l$-th position (or $\theta_{l-1,k}$ if $l=k$)  is either left unchanged or added an extra $+\pi$ with equal probability.

We provide a proof for the correctness of Algorithm~\ref{alg:clifford} in the SI~\ref{ap:clifford}. Intuitively, it can be understood from the fact that random Clifford FLO sampling must be isomorphic to sampling random signed permutations matrices of $2n$ items with unit determinant. The latter can be done by sampling $2n-1$ integers from $[2n]$ without replacement (this fixes a permutation), and then assigning $2n-1$ signs at random. The first part of this procedure is equivalent to setting angles to $\pi/2$, while the second part corresponds to choosing between adding or not an extra $+\pi$ to certain angles.

The benefit of sampling random Clifford FLO circuits instead of Haar random FLO is twofold. First, it reduces the average number of two-qubit gates by a constant factor. Indeed, we analytically prove this quantity to be $n^2/2$, which is optimal, as this is exactly half the average number of inversions needed to express a random permutation. In practice, Algorithm~\ref{alg:clifford} can significantly reduce the number of gates in the circuit as compared to previous approaches that do not rely on classical compilation~\cite{heyraud2024unified}. Such reduction is shown in Fig.~\ref{fig:clifford}. 
Second, since (by definition) random Clifford FLO circuits can be fully implemented by only using Clifford gates, they are well suited for stabilizer-code-based fault-tolerant devices~\cite{gottesman1997stabilizer,fowler2012surface}.

We remark that Algorithm~\ref{alg:clifford} does not use the brick-wall architecture of Theorem~\ref{th:fermion-haar}, and thus it produces circuits with a suboptimal depth in general. The main difficulty in transforming  Algorithm~\ref{alg:clifford} into a sampling procedure that delivers circuits with an optimal depth is that for Clifford FLO circuits there exists a non-zero probability of sampling rotation angles that are zero. Hence, it is not at all clear how to obtain a general procedure that brings the circuit from a triangular shape with correlated angles and ``holes'' (corresponding to the gates whose rotation angles are zero), into a compact form. Yet, we stress again that our algorithm produces circuits with an optimal number of two-qubit gates on average, and that once a Clifford FLO circuit has been sampled using Algorithm~\ref{alg:clifford}, one can use standard transpilation techniques to reduce its depth~\cite{kokcu2022algebraic,camps2022algebraic} (at the cost of a classical computation overhead).

Alternatively, one could combine the procedure outlined in Ref.~\cite{heyraud2024unified} with the brick-wall architecture from Theorem~\ref{th:fermion-haar} to find that random Clifford FLO can be sampled using the following PDFs for the Clifford angles $\{0,\pi/2,\pi,3\pi/2\}$ in each gate,
\begin{equation}
    p_0 = p_\pi = \frac{1}{2(f_n+2)} \,,\quad p_{\frac{\pi}{2}}= p_{\frac{3\pi}{2}}=\frac{f_n+1}{2(f_n+2)}\,.
\end{equation}
Here, $f_n$ is the power of the sine function appearing in the Haar measure associated to each angle (see Eq.~\eqref{eq:f_function}), that is, $f_n(2j,2k-1)$ for the $\alpha_{jk}$ angles and $f_n(2j-1,2k)$ for the $\beta_{jk}$ ones. While this strategy produces circuits with a compressed depth, the number of gates is not optimal, as shown in Fig.~\ref{fig:clifford}. It thus may be more convenient to use Algorithm~\ref{alg:clifford} instead, and compress the circuits afterwards.

\begin{figure}[t]
    \centering
    \includegraphics[width=1\linewidth]{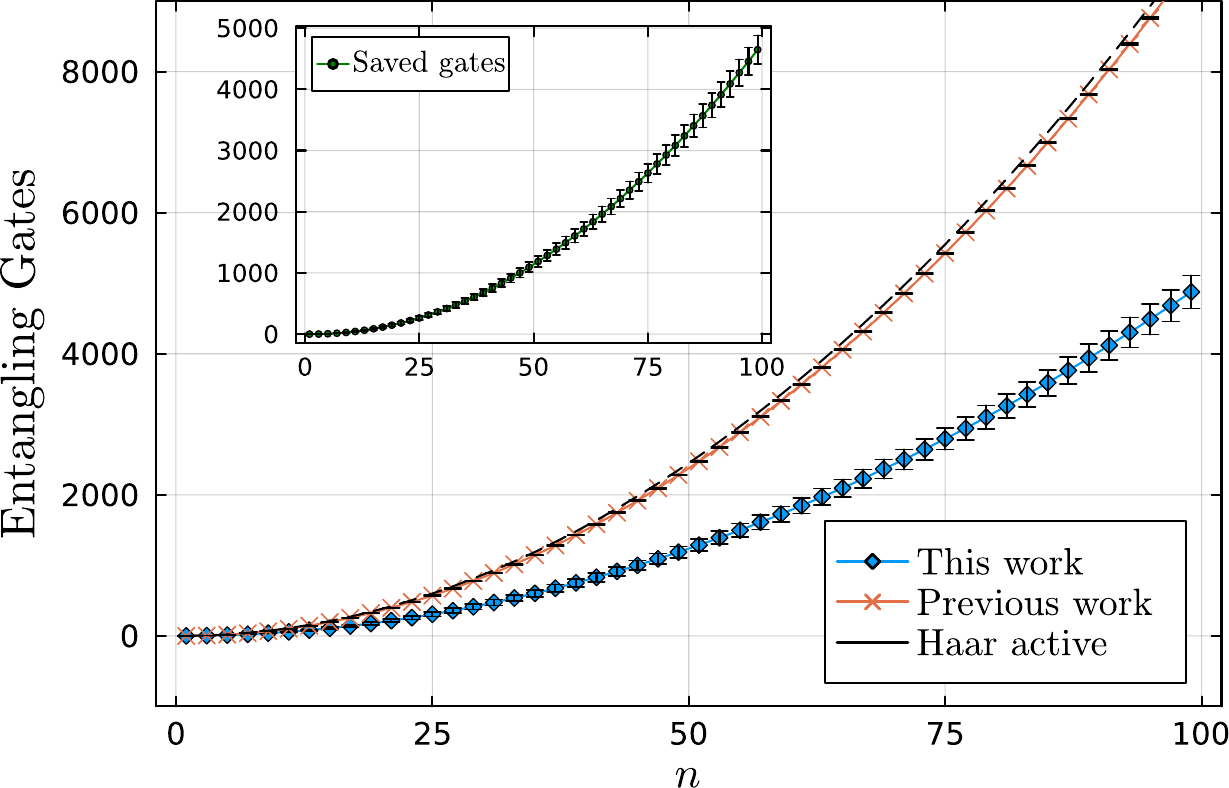}
    \caption{{\bf Number of two-qubit gates of random Clifford FLO circuits}. We compare the average number of two-qubit gates obtained by sampling $10^4$ random Clifford FLO circuits  for each circuit width $n$, according to Algorithm~\ref{alg:clifford} (blue diamonds) or the one proposed in Ref.~\cite{heyraud2024unified} (orange crosses). The error bars show the standard deviations, and the inset the difference in the total number of two-qubit gates between the two sampling protocols. We furthermore plot the number of two-qubit gates of Haar random active FLO circuits for comparison.}
    \label{fig:clifford}
\end{figure}

\begin{figure}[t]
    \centering
    \includegraphics[width=\linewidth]{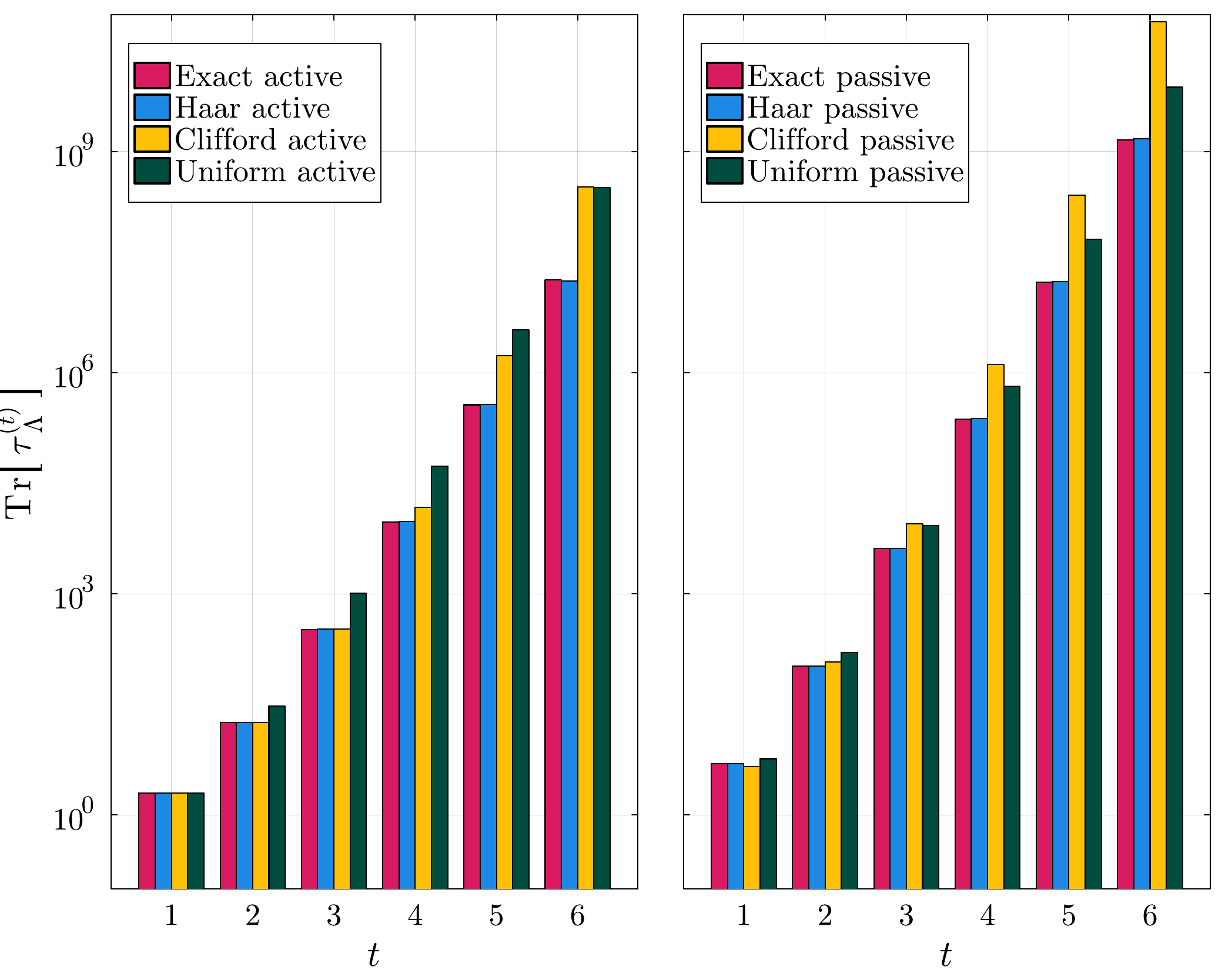}
    \caption{{\bf Deviations from the Haar measure}. We compute the expectation value $\Tr\left[\tau_\Lambda^{(t)}\right]$ for different circuit architectures and parameter sampling methods. We take $n=4$ qubits and sample $10^6$  unitaries in each case. In particular, we study the frame potential quantity for both active (left panel) and passive (right panel) FLO circuits. For each value of $t$ we compare: the exact value (magenta bars); the average value obtained by sampling the Haar random architectures from Theorems~\ref{th:fermion-haar} and~\ref{th:passive-fermion-haar} (blue bars); the average valued obtained by sampling the same circuits' architectures but with a uniform sampling of the parameters in $[0,2\pi)$ (yellow bars); and random Clifford FLO circuits sampled according to Algorithms~\ref{alg:clifford} and~\ref{alg-ap:clifford-passive} (dark green bars). We provide details for the computation of the exact values of $\Tr\left[\tau_\Lambda^{(t)}\right]$ for Haar random FLO circuits in the SI~\ref{ap:Howe}.}
    \label{fig:brick-haar}
\end{figure}

\subsection{Numerical experiments}

In this section, we present the results of numerical simulations that illustrate and verify our theoretical results. In particular, we used the architectures presented in Theorems~\ref{th:fermion-haar} and~\ref{th:passive-fermion-haar} to sample active and passive FLO circuits according to their respective Haar measures, as well as random Clifford FLO following Algorithm~\ref{alg:clifford}.

To begin, we will show that by randomly sampling unitaries according to our proposed schemes, we can recover the correct statistical properties of the Haar measure. For this purpose, we will construct the frame potentials~\cite{mele2023introduction} arising from the sampled unitaries. For convenience, we recall that given a set of  unitaries $\Lambda$ and a probability distribution $p(U)$ over the set, its $t$-th moment superoperator is defined as 
\begin{equation}
    \tau_\Lambda^{(t)}=\sum_{U\in \Lambda} p(U)\,  U^{\otimes t} \otimes (U^*)^{\otimes t} \,.
\end{equation}
If the set is continuous, then one simple replaces the summation by an integral. Moreover, when $\Lambda$ forms a group, and the probability distribution over the unitaries leads to the associated Haar measure, then the moment superoperator is a projector onto the group's $t$-th fold commutant~\cite{ragone2022representation,mele2023introduction}. Hence, in this case the eigenvalues of $\tau_\Lambda^{(t)}$ are either zero or one, and the quantity $\Tr\left[\tau_\Lambda^{(t)}\right]=\mathbb{E}_U\left[\left|\Tr\left[U\right]\right|^{2t}\right]$, which is known as the frame potential, is equal to the $t$-th-fold commutant dimension.

In Fig.~\ref{fig:brick-haar} we compute the frame potential $\Tr\left[\tau_\Lambda^{(t)}\right]$ for different values of $t$ when $\Lambda$ is a set of $10^6$ $4$-qubit unitaries sampled according to our proposed algorithms over active and passive FLO. There, we also plot the dimensions of the associated $t$-th-fold commutants, which we compute exactly using a general computational method (i.e., not specific to these groups and representations) that we present in the SI~\ref{ap:Howe}, and which may be of independent interest. As we can see from the figure, despite the statistical uncertainty arising from the fact that $\Lambda$ contains a finite number of samples, the dimensions obtained closely match. To further showcase that the correctness of our results is a consequence of correctly sampling the unitaries from the Haar measure, we also showcase in Fig.~\ref{fig:brick-haar} the value of $\Tr\left[\tau_\Lambda^{(t)}\right]$ obtained when sampling FLOs in an alternative, more naive, fashion. Namely, we use the same  circuit structure as in Theorems~\ref{th:fermion-haar} and \ref{th:passive-fermion-haar} (i.e. the same gate placements), 
but with uniform parameter distributions in $[0,2\pi)$ for all angles. As seen in Fig.~\ref{fig:brick-haar}, the ensuing values of $\Tr\left[\tau_\Lambda^{(t)}\right]$ deviate significantly from the Haar value for $t\geq2$ when we consider these alternative sampling methods that do not lead to the  Haar measure.

Finally, we employ Algorithm~\ref{alg:clifford} to produce random Clifford FLO and compare their frame potentials against the others. Interestingly, besides verifying that Clifford FLO are a FLO $3$-design, we find that $4$-qubit random  Clifford FLO circuits are much closer to Haar random circuits, for $t=4,5$, than those relying on uniformly  sampling the parameters, although we have proven that Clifford FLO are not an exact $4$-design (see SI~\ref{ap:4-design}). In addition, we use an analogous algorithm to Algorithm~\ref{alg:clifford}, that we introduce in the SI~\ref{ap:clifford}, to uniformly sample the group of passive Clifford FLO. Our numerical results indicate that passive Clifford FLO do not appear to form $t$-designs over passive FLO, which may be related to the fact that the particle number is not a Pauli symmetry~\cite{mitsuhashi2023clifford}.

\section{Discussion}

In this work we have presented \emph{optimal} quantum circuits for sampling Haar random active and passive FLO directly on quantum hardware. Unlike previous techniques which rely on sampling $2n\times 2n$ matrices and then compiling them into $n$-qubit circuits (or else use suboptimal triangular-shaped architectures), we only require sampling gate parameters and inputting them into a fixed circuit architecture, thus leading to more efficient schemes. As such, our proposed sampling methods may be included in any quantum information protocol which requires random FLO circuits.

Going forward, one may wonder if such exact parametrization of the Haar measure can be achieved for other quantum circuits that sample over Lie groups (or their intersection with the Clifford group). However, given that most circuits whose gates are generated by Pauli operators are either FLO (studied here) or lead to groups of exponentially large dimension~\cite{wiersema2023classification,kokcu2024classification,aguilar2024full}, the ensuing circuits would be exponentially deep. This does not preclude the  possibility of finding efficient methods to exactly sample from the Haar measure for circuits whose gate generators are not Paulis, but linear combinations thereof. Such is the case, for instance,  for permutation-invariant circuits~\cite{schatzki2022theoretical}, which belong to a group of dimension $\Theta(n^3)$ and for which one could potentially efficiently parametrize the Haar measure. Here, one can expect that non-local gates might be needed, as it is well known that permutation invariant gates generated by local Paulis fail to be universal~\cite{kazi2023universality,marvian2022restrictions}.  We leave, however, such question open for future work.

\section{Methods}

While we leave the detailed rigorous proofs of all of our results to the SI, we here provide an overview of the most relevant proof techniques.
In particular, in order to construct our optimal Haar random FLO circuits, we employ the following schematic strategy: 
\begin{enumerate}
    \item We start with known parametrizations of the standard representations of the classical compact Lie groups  $\SO(2n)$ and $\U(n)$, found by Hurwitz, together with the corresponding invariant measures. Here, we note that while there exist several parametrizations for these Lie groups, not all of them are well suited for implementation as quantum circuits. For instance, the construction of~\cite{spengler2012composite}  requires the use of multi-qubit gates instead of only using single- and two-qubit gates, like our scheme does. 

    \item We identify isomorphisms between the Lie algebras $\mathfrak{so}(2n)$ or $\mathfrak{u}(n)$, and the dynamical Lie algebra of the respective FLO circuits, active or passive.

    \item We use these isomorphisms to translate Hurwitz's decompositions into suboptimal  quantum circuits with a triangular shape.

    \item We derive a set of rules that allow us to track how the Haar measure transforms under certain rearrangements of the circuits' gates known as \textit{turnovers}~\cite{kokcu2022algebraic, camps2022algebraic}. This allows us to bring the previous triangular circuits into their optimal shape.
\end{enumerate}

We will now delve into more details for the case of active FLO (the reasoning for passive FLO is completely analogous), leaving the comprehensive proofs for the SI.

\subsection{Decomposition of $\SO(2n)$ with Givens rotations}

The standard representation of the special orthogonal group $\mathbb{SO}(d)$ is provided by the orthogonal matrices of size $d\times d$ and determinant equal to one, acting irreducibly on $\mathbb{R}^{d}$. These matrices satisfy  $OO^T=O^TO=\id_d$ for every $O\in\SO(d)$ (where $\id_d$ is the $d\times d$ identity matrix).

Any matrix $O\in\SO(d)$ can be decomposed as~\cite{diaconis2017hurwitz}
\begin{equation} \label{eq:so-matrix-decomposition}
    O=O_1 O_2\cdots O_{d-1}\,, 
\end{equation}
for
\begin{equation}
    O_j=R_j(\theta_{j,j+1})\cdots R_1(\theta_{1,j+1})\,,
\end{equation}
with Euler angles $\theta_{1,j+1}\in [0,2\pi)$, and $\theta_{j',j+1}\in [0,\pi]$ for $j'>1$. Here, the $R_j$ are Givens rotations defined by
\begin{equation} \label{eq:R-O-tilde-prod}
    R_j(\theta)=e^{\theta L_{jj+1}}= \!\begin{pmatrix}
        \id_{j-1} \\ & \cos\theta & \sin\theta \\ & -\sin\theta & \cos\theta \\ & & & \id_{d-j-1}
    \end{pmatrix}\,,
\end{equation}
where $\{L_{jk}\}_{1\leq j<k\leq d}$, and
\begin{equation}
    \left(L_{jk}\right)_{m,l} = \delta_{jm}\delta_{kl}-\delta_{jl}\delta_{km}\,,
\end{equation}
is an orthogonal basis for the vector space of anti-symmetric matrices of size $d\times d$, and hence for the special orthogonal algebra $\mathfrak{so}(d)$. 
The idea behind this decomposition is that the inverse of each $O_j$ (i.e., its transpose $O_j^T$) can always be used to bring the $j$-th row of $O$ to the corresponding canonical basis vector $\vec{e}_j^T$, where $(\vec{e}_j)_k=\delta_{jk}$. 

Moreover, the normalized Haar measure $d\mu(O)$ on the standard representation of $\SO(d)$ with respect to the parametrization in Eqs.~\eqref{eq:so-matrix-decomposition} and~\eqref{eq:R-O-tilde-prod} is given by~\cite{diaconis2017hurwitz}
    \begin{equation} \label{eq:haar-so}
        d\mu(O) = \NC \prod_{1\leq j <k\leq d}\sin(\theta_{j,k})^{j-1} d\theta_{j,k}\,,
    \end{equation}
    where $\NC=\left((2\pi)^{d-1} \prod_{k=3}^{d}\frac{\sqrt{\pi}^{k-2}}{\Gamma(k/2)}\right)^{-1}$.

\subsection{Mapping the Lie algebras of $\SO(2n)$ and $\mathbb{SPIN}(2n)$}

The group of active FLO unitaries, or mathcgate circuits, is well known to be isomorphic to the group $\mathbb{SPIN}(2n)$~\cite{guaita2024representation}. In order to study FLO circuits, it is most convenient to introduce the Majorana operators, which under the Jordan-Wigner transformation take the form 
\begin{equation}
\begin{split}
    c_1&=X\id\dots \id,\; c_3= ZX\id\dots \id, \;\;\dots, \; c_{2n-1}=Z\dots Z X\,, \nonumber\\
        c_2&=Y\id\dots \id,\; c_4= ZY\id\dots \id,\;\;\dots, \; \;\; c_{2n}\;\;\;=Z\dots Z Y\,.
\end{split}
\end{equation}
These operators act on the Hilbert space $(\mathbb{C}^{2})^{\otimes n}$ of $n$ qubits, are Hermitian (as they correspond to Pauli strings), and satisfy the anti-commutation relations $    \{c_j,c_k\}=2\delta_{jk}$, for $j,k=1,\dots ,2n$. 
It can be shown (see e.g.,~\cite{diaz2023showcasing}) that the dynamical Lie algebra $\g$ of active FLO circuits, i.e., the Lie closure --or the real vector space spanned by the nested commutators-- of the generators $i\GC$ in Eq.~\eqref{eq:fermion-circuits}~\cite{zeier2011symmetry},   is
\begin{equation}\label{eq:dla}
    \g={\rm span}_{\mathbb R}\{c_j c_k\}_{1\leq j<k \leq 2n}\,.
\end{equation} 
The Lie algebra in Eq.~\eqref{eq:dla} is the real vector space spanned by the product of two distinct Majoranas, and it is isomorphic to the orthogonal Lie algebra $\mathfrak{so}(2n)$. 
As we show in the SI, this isomorphism can be realized by the following linear map
\begin{equation} \label{eq:isomorphism}
    \varphi(L_{jk})= \frac{c_j c_k}{2}\,,
\end{equation} 
between the generators of $\so(2n)$ and the products of two distinct Majorana operators. As such,  any FLO circuit will be a unitary in the Lie group $e^{\g}$.  

\begin{figure}[t]
    \centering

     \includegraphics[width=\linewidth]{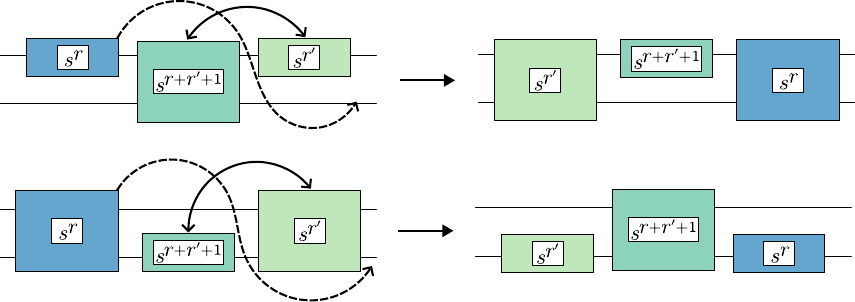}
    \caption{{\bf Bringing the FLO Haar measure to optimal depth}. We graphically show the set of rules that allow us to find optimal Haar random active FLO circuits. In particular,  these rules specify how the Haar measure transforms under certain ``turnovers'' of the gates, according to Eqs.~\eqref{eq:turnover_before} and~\eqref{eq:turnover_after}.}
    \label{fig:turnover-Haar}
\end{figure}

\subsection{Translating Hurwitz's decomposition into quantum circuits}

While the Lie algebras of active FLO circuits and special orthogonal matrices of even dimension are isomorphic, the group of unitaries of the form~\eqref{eq:fermion-circuits} implemented by these circuits and $\SO(2n)$ are not. This is a consequence of the fact that $\SO(2n)$ is not a simply-connected Lie group~\cite{hall2013lie}. Instead, FLO unitaries are a representation of the $\mathbb{SPIN}(2n)$ group, which is the double cover of $\mathbb{SO}(2n)$~\cite{guaita2024representation}. This can be understood from the factor $1/2$ in Eq.~\eqref{eq:isomorphism}: consider a matrix $O=e^{\theta L_{jk}}\in\SO(2n)$, and let us use Eq.~\eqref{eq:isomorphism} to obtain $U=e^{\theta c_j c_k/2}$. It is clear that $O'= e^{(\theta+2\pi) L_{jk}}=O$ but $U'=e^{(\theta+2\pi) c_j c_k/2}=-U$. Hence, both $U$ and $-U$ are mapped to the same $O$, which implies that the groups are not isomorphic. 
However, the adjoint representation ${\rm Ad}_U(\cdot) = U(\cdot)U^\dagger$ of these unitaries acting on the vector space of Majorana operators is isomorphic to that of $\SO(2n)$ on $\mathbb{R}^{2n}$, since in the adjoint representation both $U$ and $-U$ have the exact same action. 
This isomorphism is realized as~\cite{jozsa2008matchgates}
\begin{equation} \label{eq:action-majorana}
    U c_{l\,} U^{\dagger} = \sum_
{m=1}^{2n} (O)_{lm}c_{m}\,,
\end{equation}
where $(O)_{lm}$ are the entries of an special orthogonal matrix belonging to $\SO(2n)$.

This fact will allow us to use the isomorphism $\varphi$ in Eq.~\eqref{eq:isomorphism} to straightforwardly map the decomposition~\eqref{eq:so-matrix-decomposition}, and associated Haar measure~\eqref{eq:haar-so}, of the matrices in the standard representation of $\SO(2n)$,  to (the adjoint action of) a quantum circuit, by replacing the Givens rotations $R_j(\theta)$ with
\begin{equation} 
    U_{R_j}(\theta)=e^{\theta c_j c_{j+1}/2}\,.
\end{equation}
One can readily verify that for odd $j=1,3,\dots$, we obtain $c_j c_{j+1}=iZ_j$, while for even $j=2,4,\dots$, we find $c_j c_{j+1}=iX_jX_{j+1}$, recovering the usual matchgate circuits generators $\GC=\{Z_q\}_{q=1}^n \cup \{X_qX_{q+1}\}_{q=1}^{n-1}$.

As explained in the main text, this direct mapping produces quantum circuits with a triangular shape, and hence a suboptimal depth.

\subsection{Compressing the circuits to optimal depth}

Finally, we resort to the following ``turnover'' property~\cite{kokcu2022algebraic,camps2022algebraic},
\begin{equation}
    U_{R_j}(\alpha)U_{R_{j+1}}(\beta)U_{R_j}(\gamma) = U_{R_{j+1}}(a)U_{R_j}(b)U_{R_{j+1}}(c) \,,
\end{equation}
for conveniently chosen $a(\alpha,\beta,\gamma)$, $b(\alpha,\beta,\gamma)$ and $c(\alpha,\beta,\gamma)$,
to shuffle gates around and bring the previously obtained triangular circuits into circuits with an optimally compressed brick-wall structure.
Specifically, we derive a set of rules specifying how the Haar measure transforms under certain turnovers. Namely, we prove that
the unnormalized distribution
    \begin{equation} \label{eq:turnover_before}
       \chi_{r,r'}(\alpha,\beta,\gamma) = \sin^r (\alpha) \sin^{r+r'+1}(\beta) \sin^{r'}(\gamma) \,,
    \end{equation}
    transforms into
    \begin{equation} \label{eq:turnover_after}
        f_{r,r'}(a,b,c) = \sin^{r'} (a)\sin^{r+r'+1} (b) \sin^{r}(c) \,.
    \end{equation}
    That is,  $\chi_{r,r'}(\alpha,\beta,\gamma)\,d\alpha\, d\beta\, d\gamma =f_{r,r'}(a, b, c)\, da\, db\, dc$, 
as graphically depicted in Fig.~\ref{fig:turnover-Haar}. In the SI, we rigorously show that this simple set of rules is sufficient to obtain optimal Haar random active FLO circuits from those obtained via Hurwitz's decomposition.

\section*{Acknowledgments}

We thank Michael Ragone, Antonio Anna Mele and Andrea Palermo for useful conversations. PB, MC and DGM were supported by the Laboratory Directed Research and Development (LDRD) program of Los Alamos National Laboratory (LANL) under project numbers 20230049DR and 20230527ECR. NLD was also initially supported by CONICET  Argentina. NLD was also supported by the Center for Nonlinear Studies at LANL. MC and DGM also acknowledge support by LANL's ASC Beyond Moore’s Law project.

\bibliography{quantum}

\begin{thebibliography}{78}%
\makeatletter
\providecommand \@ifxundefined [1]{%
 \@ifx{#1\undefined}
}%
\providecommand \@ifnum [1]{%
 \ifnum #1\expandafter \@firstoftwo
 \else \expandafter \@secondoftwo
 \fi
}%
\providecommand \@ifx [1]{%
 \ifx #1\expandafter \@firstoftwo
 \else \expandafter \@secondoftwo
 \fi
}%
\providecommand \natexlab [1]{#1}%
\providecommand \enquote  [1]{``#1''}%
\providecommand \bibnamefont  [1]{#1}%
\providecommand \bibfnamefont [1]{#1}%
\providecommand \citenamefont [1]{#1}%
\providecommand \href@noop [0]{\@secondoftwo}%
\providecommand \href [0]{\begingroup \@sanitize@url \@href}%
\providecommand \@href[1]{\@@startlink{#1}\@@href}%
\providecommand \@@href[1]{\endgroup#1\@@endlink}%
\providecommand \@sanitize@url [0]{\catcode `\\12\catcode `\$12\catcode
  `\&12\catcode `\#12\catcode `\^12\catcode `\_12\catcode `\%12\relax}%
\providecommand \@@startlink[1]{}%
\providecommand \@@endlink[0]{}%
\providecommand \url  [0]{\begingroup\@sanitize@url \@url }%
\providecommand \@url [1]{\endgroup\@href {#1}{\urlprefix }}%
\providecommand \urlprefix  [0]{URL }%
\providecommand \Eprint [0]{\href }%
\providecommand \doibase [0]{https://doi.org/}%
\providecommand \selectlanguage [0]{\@gobble}%
\providecommand \bibinfo  [0]{\@secondoftwo}%
\providecommand \bibfield  [0]{\@secondoftwo}%
\providecommand \translation [1]{[#1]}%
\providecommand \BibitemOpen [0]{}%
\providecommand \bibitemStop [0]{}%
\providecommand \bibitemNoStop [0]{.\EOS\space}%
\providecommand \EOS [0]{\spacefactor3000\relax}%
\providecommand \BibitemShut  [1]{\csname bibitem#1\endcsname}%
\let\auto@bib@innerbib\@empty
\bibitem [{\citenamefont {Valiant}(2001)}]{valiant2001quantum}%
  \BibitemOpen
  \bibfield  {author} {\bibinfo {author} {\bibfnamefont {L.~G.}\ \bibnamefont
  {Valiant}},\ }\bibfield  {title} {\bibinfo {title} {Quantum computers that
  can be simulated classically in polynomial time},\ }in\ \href
  {https://doi.org/10.1145/380752.380785} {\emph {\bibinfo {booktitle}
  {Proceedings of the thirty-third annual ACM symposium on Theory of
  computing}}}\ (\bibinfo {year} {2001})\ pp.\ \bibinfo {pages}
  {114--123}\BibitemShut {NoStop}%
\bibitem [{\citenamefont {Knill}(2001)}]{knill2001fermionic}%
  \BibitemOpen
  \bibfield  {author} {\bibinfo {author} {\bibfnamefont {E.}~\bibnamefont
  {Knill}},\ }\bibfield  {title} {\bibinfo {title} {Fermionic linear optics and
  matchgates},\ }\href {https://arxiv.org/abs/quant-ph/0108033} {\bibfield
  {journal} {\bibinfo  {journal} {arXiv preprint arXiv:quant-ph/0108033}\ }
  (\bibinfo {year} {2001})}\BibitemShut {NoStop}%
\bibitem [{\citenamefont {Terhal}\ and\ \citenamefont
  {DiVincenzo}(2002)}]{terhal2002classical}%
  \BibitemOpen
  \bibfield  {author} {\bibinfo {author} {\bibfnamefont {B.~M.}\ \bibnamefont
  {Terhal}}\ and\ \bibinfo {author} {\bibfnamefont {D.~P.}\ \bibnamefont
  {DiVincenzo}},\ }\bibfield  {title} {\bibinfo {title} {Classical simulation
  of noninteracting-fermion quantum circuits},\ }\href
  {https://doi.org/10.1103/PhysRevA.65.032325} {\bibfield  {journal} {\bibinfo
  {journal} {Physical Review A}\ }\textbf {\bibinfo {volume} {65}},\ \bibinfo
  {pages} {032325} (\bibinfo {year} {2002})}\BibitemShut {NoStop}%
\bibitem [{\citenamefont {DiVincenzo}\ and\ \citenamefont
  {Terhal}(2005)}]{divincenzo2005fermionic}%
  \BibitemOpen
  \bibfield  {author} {\bibinfo {author} {\bibfnamefont {D.~P.}\ \bibnamefont
  {DiVincenzo}}\ and\ \bibinfo {author} {\bibfnamefont {B.~M.}\ \bibnamefont
  {Terhal}},\ }\bibfield  {title} {\bibinfo {title} {Fermionic linear optics
  revisited},\ }\href {https://doi.org/10.1007/s10701-005-8657-0} {\bibfield
  {journal} {\bibinfo  {journal} {Foundations of Physics}\ }\textbf {\bibinfo
  {volume} {35}},\ \bibinfo {pages} {1967} (\bibinfo {year}
  {2005})}\BibitemShut {NoStop}%
\bibitem [{\citenamefont {Kivlichan}\ \emph {et~al.}(2018)\citenamefont
  {Kivlichan}, \citenamefont {McClean}, \citenamefont {Wiebe}, \citenamefont
  {Gidney}, \citenamefont {Aspuru-Guzik}, \citenamefont {Chan},\ and\
  \citenamefont {Babbush}}]{kivlichan2018quantum}%
  \BibitemOpen
  \bibfield  {author} {\bibinfo {author} {\bibfnamefont {I.~D.}\ \bibnamefont
  {Kivlichan}}, \bibinfo {author} {\bibfnamefont {J.}~\bibnamefont {McClean}},
  \bibinfo {author} {\bibfnamefont {N.}~\bibnamefont {Wiebe}}, \bibinfo
  {author} {\bibfnamefont {C.}~\bibnamefont {Gidney}}, \bibinfo {author}
  {\bibfnamefont {A.}~\bibnamefont {Aspuru-Guzik}}, \bibinfo {author}
  {\bibfnamefont {G.~K.-L.}\ \bibnamefont {Chan}},\ and\ \bibinfo {author}
  {\bibfnamefont {R.}~\bibnamefont {Babbush}},\ }\bibfield  {title} {\bibinfo
  {title} {Quantum simulation of electronic structure with linear depth and
  connectivity},\ }\href {https://doi.org/10.1103/PhysRevLett.120.110501}
  {\bibfield  {journal} {\bibinfo  {journal} {Physical review letters}\
  }\textbf {\bibinfo {volume} {120}},\ \bibinfo {pages} {110501} (\bibinfo
  {year} {2018})}\BibitemShut {NoStop}%
\bibitem [{\citenamefont {Arute}\ \emph {et~al.}(2020)\citenamefont {Arute},
  \citenamefont {Arya}, \citenamefont {Babbush}, \citenamefont {Bacon},
  \citenamefont {Bardin}, \citenamefont {Barends}, \citenamefont {Boixo},
  \citenamefont {Broughton}, \citenamefont {Buckley}, \citenamefont {Buell}
  \emph {et~al.}}]{arute2020hartree}%
  \BibitemOpen
  \bibfield  {author} {\bibinfo {author} {\bibfnamefont {F.}~\bibnamefont
  {Arute}}, \bibinfo {author} {\bibfnamefont {K.}~\bibnamefont {Arya}},
  \bibinfo {author} {\bibfnamefont {R.}~\bibnamefont {Babbush}}, \bibinfo
  {author} {\bibfnamefont {D.}~\bibnamefont {Bacon}}, \bibinfo {author}
  {\bibfnamefont {J.~C.}\ \bibnamefont {Bardin}}, \bibinfo {author}
  {\bibfnamefont {R.}~\bibnamefont {Barends}}, \bibinfo {author} {\bibfnamefont
  {S.}~\bibnamefont {Boixo}}, \bibinfo {author} {\bibfnamefont
  {M.}~\bibnamefont {Broughton}}, \bibinfo {author} {\bibfnamefont {B.~B.}\
  \bibnamefont {Buckley}}, \bibinfo {author} {\bibfnamefont {D.~A.}\
  \bibnamefont {Buell}}, \emph {et~al.},\ }\bibfield  {title} {\bibinfo {title}
  {Hartree-fock on a superconducting qubit quantum computer},\ }\href
  {https://doi.org/10.1126/science.abb9811} {\bibfield  {journal} {\bibinfo
  {journal} {Science}\ }\textbf {\bibinfo {volume} {369}},\ \bibinfo {pages}
  {1084} (\bibinfo {year} {2020})}\BibitemShut {NoStop}%
\bibitem [{\citenamefont {Arrazola}\ \emph {et~al.}(2022)\citenamefont
  {Arrazola}, \citenamefont {Di~Matteo}, \citenamefont {Quesada}, \citenamefont
  {Jahangiri}, \citenamefont {Delgado},\ and\ \citenamefont
  {Killoran}}]{arrazola2022universal}%
  \BibitemOpen
  \bibfield  {author} {\bibinfo {author} {\bibfnamefont {J.~M.}\ \bibnamefont
  {Arrazola}}, \bibinfo {author} {\bibfnamefont {O.}~\bibnamefont {Di~Matteo}},
  \bibinfo {author} {\bibfnamefont {N.}~\bibnamefont {Quesada}}, \bibinfo
  {author} {\bibfnamefont {S.}~\bibnamefont {Jahangiri}}, \bibinfo {author}
  {\bibfnamefont {A.}~\bibnamefont {Delgado}},\ and\ \bibinfo {author}
  {\bibfnamefont {N.}~\bibnamefont {Killoran}},\ }\bibfield  {title} {\bibinfo
  {title} {Universal quantum circuits for quantum chemistry},\ }\href
  {https://doi.org/10.22331/q-2022-06-20-742} {\bibfield  {journal} {\bibinfo
  {journal} {Quantum}\ }\textbf {\bibinfo {volume} {6}},\ \bibinfo {pages}
  {742} (\bibinfo {year} {2022})}\BibitemShut {NoStop}%
\bibitem [{\citenamefont {Verstraete}\ \emph {et~al.}(2009)\citenamefont
  {Verstraete}, \citenamefont {Cirac},\ and\ \citenamefont
  {Latorre}}]{verstraete2009quantum}%
  \BibitemOpen
  \bibfield  {author} {\bibinfo {author} {\bibfnamefont {F.}~\bibnamefont
  {Verstraete}}, \bibinfo {author} {\bibfnamefont {J.~I.}\ \bibnamefont
  {Cirac}},\ and\ \bibinfo {author} {\bibfnamefont {J.~I.}\ \bibnamefont
  {Latorre}},\ }\bibfield  {title} {\bibinfo {title} {Quantum circuits for
  strongly correlated quantum systems},\ }\href
  {https://doi.org/10.1103/PhysRevA.79.032316} {\bibfield  {journal} {\bibinfo
  {journal} {Physical Review A}\ }\textbf {\bibinfo {volume} {79}},\ \bibinfo
  {pages} {032316} (\bibinfo {year} {2009})}\BibitemShut {NoStop}%
\bibitem [{\citenamefont {Kraus}(2011)}]{kraus2011compressed}%
  \BibitemOpen
  \bibfield  {author} {\bibinfo {author} {\bibfnamefont {B.}~\bibnamefont
  {Kraus}},\ }\bibfield  {title} {\bibinfo {title} {Compressed quantum
  simulation of the ising model},\ }\href
  {https://doi.org/10.1103/PhysRevLett.107.250503} {\bibfield  {journal}
  {\bibinfo  {journal} {Physical review letters}\ }\textbf {\bibinfo {volume}
  {107}},\ \bibinfo {pages} {250503} (\bibinfo {year} {2011})}\BibitemShut
  {NoStop}%
\bibitem [{\citenamefont {Cervera-Lierta}(2018)}]{cervera2018exact}%
  \BibitemOpen
  \bibfield  {author} {\bibinfo {author} {\bibfnamefont {A.}~\bibnamefont
  {Cervera-Lierta}},\ }\bibfield  {title} {\bibinfo {title} {Exact ising model
  simulation on a quantum computer},\ }\href
  {https://doi.org/10.22331/q-2018-12-21-114} {\bibfield  {journal} {\bibinfo
  {journal} {Quantum}\ }\textbf {\bibinfo {volume} {2}},\ \bibinfo {pages}
  {114} (\bibinfo {year} {2018})}\BibitemShut {NoStop}%
\bibitem [{\citenamefont {Jiang}\ \emph {et~al.}(2018)\citenamefont {Jiang},
  \citenamefont {Sung}, \citenamefont {Kechedzhi}, \citenamefont
  {Smelyanskiy},\ and\ \citenamefont {Boixo}}]{jiang2018quantum}%
  \BibitemOpen
  \bibfield  {author} {\bibinfo {author} {\bibfnamefont {Z.}~\bibnamefont
  {Jiang}}, \bibinfo {author} {\bibfnamefont {K.~J.}\ \bibnamefont {Sung}},
  \bibinfo {author} {\bibfnamefont {K.}~\bibnamefont {Kechedzhi}}, \bibinfo
  {author} {\bibfnamefont {V.~N.}\ \bibnamefont {Smelyanskiy}},\ and\ \bibinfo
  {author} {\bibfnamefont {S.}~\bibnamefont {Boixo}},\ }\bibfield  {title}
  {\bibinfo {title} {Quantum algorithms to simulate many-body physics of
  correlated fermions},\ }\href
  {https://doi.org/10.1103/PhysRevApplied.9.044036} {\bibfield  {journal}
  {\bibinfo  {journal} {Physical Review Applied}\ }\textbf {\bibinfo {volume}
  {9}},\ \bibinfo {pages} {044036} (\bibinfo {year} {2018})}\BibitemShut
  {NoStop}%
\bibitem [{\citenamefont {Dallaire-Demers}\ \emph {et~al.}(2019)\citenamefont
  {Dallaire-Demers}, \citenamefont {Romero}, \citenamefont {Veis},
  \citenamefont {Sim},\ and\ \citenamefont {Aspuru-Guzik}}]{dallaire2019low}%
  \BibitemOpen
  \bibfield  {author} {\bibinfo {author} {\bibfnamefont {P.-L.}\ \bibnamefont
  {Dallaire-Demers}}, \bibinfo {author} {\bibfnamefont {J.}~\bibnamefont
  {Romero}}, \bibinfo {author} {\bibfnamefont {L.}~\bibnamefont {Veis}},
  \bibinfo {author} {\bibfnamefont {S.}~\bibnamefont {Sim}},\ and\ \bibinfo
  {author} {\bibfnamefont {A.}~\bibnamefont {Aspuru-Guzik}},\ }\bibfield
  {title} {\bibinfo {title} {Low-depth circuit ansatz for preparing correlated
  fermionic states on a quantum computer},\ }\href
  {https://doi.org/10.1088/2058-9565/ab3951/} {\bibfield  {journal} {\bibinfo
  {journal} {Quantum Science and Technology}\ }\textbf {\bibinfo {volume}
  {4}},\ \bibinfo {pages} {045005} (\bibinfo {year} {2019})}\BibitemShut
  {NoStop}%
\bibitem [{\citenamefont {Sopena}\ \emph {et~al.}(2022)\citenamefont {Sopena},
  \citenamefont {Gordon}, \citenamefont {Garc{\'{i}}a-Mart{\'{i}}n},
  \citenamefont {Sierra},\ and\ \citenamefont
  {L{\'{o}}pez}}]{sopena2022algebraic}%
  \BibitemOpen
  \bibfield  {author} {\bibinfo {author} {\bibfnamefont {A.}~\bibnamefont
  {Sopena}}, \bibinfo {author} {\bibfnamefont {M.~H.}\ \bibnamefont {Gordon}},
  \bibinfo {author} {\bibfnamefont {D.}~\bibnamefont
  {Garc{\'{i}}a-Mart{\'{i}}n}}, \bibinfo {author} {\bibfnamefont
  {G.}~\bibnamefont {Sierra}},\ and\ \bibinfo {author} {\bibfnamefont
  {E.}~\bibnamefont {L{\'{o}}pez}},\ }\bibfield  {title} {\bibinfo {title}
  {Algebraic {B}ethe {C}ircuits},\ }\href
  {https://doi.org/10.22331/q-2022-09-08-796} {\bibfield  {journal} {\bibinfo
  {journal} {{Quantum}}\ }\textbf {\bibinfo {volume} {6}},\ \bibinfo {pages}
  {796} (\bibinfo {year} {2022})}\BibitemShut {NoStop}%
\bibitem [{\citenamefont {Ruiz}\ \emph
  {et~al.}(2024{\natexlab{a}})\citenamefont {Ruiz}, \citenamefont {Sopena},
  \citenamefont {Gordon}, \citenamefont {Sierra},\ and\ \citenamefont
  {L{\'o}pez}}]{ruiz2024bethe}%
  \BibitemOpen
  \bibfield  {author} {\bibinfo {author} {\bibfnamefont {R.}~\bibnamefont
  {Ruiz}}, \bibinfo {author} {\bibfnamefont {A.}~\bibnamefont {Sopena}},
  \bibinfo {author} {\bibfnamefont {M.~H.}\ \bibnamefont {Gordon}}, \bibinfo
  {author} {\bibfnamefont {G.}~\bibnamefont {Sierra}},\ and\ \bibinfo {author}
  {\bibfnamefont {E.}~\bibnamefont {L{\'o}pez}},\ }\bibfield  {title} {\bibinfo
  {title} {The bethe ansatz as a quantum circuit},\ }\href
  {https://doi.org/10.22331/q-2024-05-23-1356} {\bibfield  {journal} {\bibinfo
  {journal} {Quantum}\ }\textbf {\bibinfo {volume} {8}},\ \bibinfo {pages}
  {1356} (\bibinfo {year} {2024}{\natexlab{a}})}\BibitemShut {NoStop}%
\bibitem [{\citenamefont {Ruiz}\ \emph
  {et~al.}(2024{\natexlab{b}})\citenamefont {Ruiz}, \citenamefont {Sopena},
  \citenamefont {Pozsgay},\ and\ \citenamefont
  {L{\'o}pez}}]{ruiz2024efficient}%
  \BibitemOpen
  \bibfield  {author} {\bibinfo {author} {\bibfnamefont {R.}~\bibnamefont
  {Ruiz}}, \bibinfo {author} {\bibfnamefont {A.}~\bibnamefont {Sopena}},
  \bibinfo {author} {\bibfnamefont {B.}~\bibnamefont {Pozsgay}},\ and\ \bibinfo
  {author} {\bibfnamefont {E.}~\bibnamefont {L{\'o}pez}},\ }\bibfield  {title}
  {\bibinfo {title} {Efficient eigenstate preparation in an integrable model
  with hilbert space fragmentation},\ }\href {https://arxiv.org/abs/2411.15132}
  {\bibfield  {journal} {\bibinfo  {journal} {arXiv preprint arXiv:2411.15132}\
  } (\bibinfo {year} {2024}{\natexlab{b}})}\BibitemShut {NoStop}%
\bibitem [{\citenamefont {Bravyi}(2005)}]{bravyi2004lagrangian}%
  \BibitemOpen
  \bibfield  {author} {\bibinfo {author} {\bibfnamefont {S.}~\bibnamefont
  {Bravyi}},\ }\bibfield  {title} {\bibinfo {title} {Lagrangian representation
  for fermionic linear optics},\ }\href
  {https://doi.org/10.5555/2011637.2011640} {\bibfield  {journal} {\bibinfo
  {journal} {Quantum Info. Comput.}\ }\textbf {\bibinfo {volume} {5}},\
  \bibinfo {pages} {216–238} (\bibinfo {year} {2005})}\BibitemShut {NoStop}%
\bibitem [{\citenamefont {Jozsa}\ and\ \citenamefont
  {Miyake}(2008)}]{jozsa2008matchgates}%
  \BibitemOpen
  \bibfield  {author} {\bibinfo {author} {\bibfnamefont {R.}~\bibnamefont
  {Jozsa}}\ and\ \bibinfo {author} {\bibfnamefont {A.}~\bibnamefont {Miyake}},\
  }\bibfield  {title} {\bibinfo {title} {Matchgates and classical simulation of
  quantum circuits},\ }\href {https://doi.org/10.1098/rspa.2008.0189}
  {\bibfield  {journal} {\bibinfo  {journal} {Proceedings of the Royal Society
  A: Mathematical, Physical and Engineering Sciences}\ }\textbf {\bibinfo
  {volume} {464}},\ \bibinfo {pages} {3089} (\bibinfo {year}
  {2008})}\BibitemShut {NoStop}%
\bibitem [{\citenamefont {Brod}\ and\ \citenamefont
  {Galvao}(2011)}]{brod2011extending}%
  \BibitemOpen
  \bibfield  {author} {\bibinfo {author} {\bibfnamefont {D.~J.}\ \bibnamefont
  {Brod}}\ and\ \bibinfo {author} {\bibfnamefont {E.~F.}\ \bibnamefont
  {Galvao}},\ }\bibfield  {title} {\bibinfo {title} {Extending matchgates into
  universal quantum computation},\ }\href
  {https://doi.org/10.1103/PhysRevA.84.022310} {\bibfield  {journal} {\bibinfo
  {journal} {Physical Review A—Atomic, Molecular, and Optical Physics}\
  }\textbf {\bibinfo {volume} {84}},\ \bibinfo {pages} {022310} (\bibinfo
  {year} {2011})}\BibitemShut {NoStop}%
\bibitem [{\citenamefont {Brod}\ and\ \citenamefont
  {Childs}(2014)}]{brod2014computational}%
  \BibitemOpen
  \bibfield  {author} {\bibinfo {author} {\bibfnamefont {D.~J.}\ \bibnamefont
  {Brod}}\ and\ \bibinfo {author} {\bibfnamefont {A.~M.}\ \bibnamefont
  {Childs}},\ }\bibfield  {title} {\bibinfo {title} {The computational power of
  matchgates and the xy interaction on arbitrary graphs},\ }\href
  {https://doi.org/10.26421/QIC14.11-12-1} {\bibfield  {journal} {\bibinfo
  {journal} {Quantum Information and Computation}\ }\textbf {\bibinfo {volume}
  {14}},\ \bibinfo {pages} {901} (\bibinfo {year} {2014})}\BibitemShut
  {NoStop}%
\bibitem [{\citenamefont {Brod}(2016)}]{brod2016efficient}%
  \BibitemOpen
  \bibfield  {author} {\bibinfo {author} {\bibfnamefont {D.~J.}\ \bibnamefont
  {Brod}},\ }\bibfield  {title} {\bibinfo {title} {Efficient classical
  simulation of matchgate circuits with generalized inputs and measurements},\
  }\href {https://doi.org/10.1103/PhysRevA.93.062332} {\bibfield  {journal}
  {\bibinfo  {journal} {Physical Review A}\ }\textbf {\bibinfo {volume} {93}},\
  \bibinfo {pages} {062332} (\bibinfo {year} {2016})}\BibitemShut {NoStop}%
\bibitem [{\citenamefont {Helsen}\ \emph {et~al.}(2022)\citenamefont {Helsen},
  \citenamefont {Nezami}, \citenamefont {Reagor},\ and\ \citenamefont
  {Walter}}]{helsen2022matchgate}%
  \BibitemOpen
  \bibfield  {author} {\bibinfo {author} {\bibfnamefont {J.}~\bibnamefont
  {Helsen}}, \bibinfo {author} {\bibfnamefont {S.}~\bibnamefont {Nezami}},
  \bibinfo {author} {\bibfnamefont {M.}~\bibnamefont {Reagor}},\ and\ \bibinfo
  {author} {\bibfnamefont {M.}~\bibnamefont {Walter}},\ }\bibfield  {title}
  {\bibinfo {title} {Matchgate benchmarking: Scalable benchmarking of a
  continuous family of many-qubit gates},\ }\href
  {https://doi.org/10.22331/q-2022-02-21-657} {\bibfield  {journal} {\bibinfo
  {journal} {Quantum}\ }\textbf {\bibinfo {volume} {6}},\ \bibinfo {pages}
  {657} (\bibinfo {year} {2022})}\BibitemShut {NoStop}%
\bibitem [{\citenamefont {Wan}\ \emph {et~al.}(2023)\citenamefont {Wan},
  \citenamefont {Huggins}, \citenamefont {Lee},\ and\ \citenamefont
  {Babbush}}]{wan2022matchgate}%
  \BibitemOpen
  \bibfield  {author} {\bibinfo {author} {\bibfnamefont {K.}~\bibnamefont
  {Wan}}, \bibinfo {author} {\bibfnamefont {W.~J.}\ \bibnamefont {Huggins}},
  \bibinfo {author} {\bibfnamefont {J.}~\bibnamefont {Lee}},\ and\ \bibinfo
  {author} {\bibfnamefont {R.}~\bibnamefont {Babbush}},\ }\bibfield  {title}
  {\bibinfo {title} {Matchgate shadows for fermionic quantum simulation},\
  }\href {https://doi.org/10.1007/s00220-023-04844-0} {\bibfield  {journal}
  {\bibinfo  {journal} {Communications in Mathematical Physics}\ }\textbf
  {\bibinfo {volume} {404}},\ \bibinfo {pages} {629} (\bibinfo {year}
  {2023})}\BibitemShut {NoStop}%
\bibitem [{\citenamefont {Diaz}\ \emph
  {et~al.}(2023{\natexlab{a}})\citenamefont {Diaz}, \citenamefont
  {Garc{\'\i}a-Mart{\'\i}n}, \citenamefont {Kazi}, \citenamefont {Larocca},\
  and\ \citenamefont {Cerezo}}]{diaz2023showcasing}%
  \BibitemOpen
  \bibfield  {author} {\bibinfo {author} {\bibfnamefont {N.~L.}\ \bibnamefont
  {Diaz}}, \bibinfo {author} {\bibfnamefont {D.}~\bibnamefont
  {Garc{\'\i}a-Mart{\'\i}n}}, \bibinfo {author} {\bibfnamefont
  {S.}~\bibnamefont {Kazi}}, \bibinfo {author} {\bibfnamefont {M.}~\bibnamefont
  {Larocca}},\ and\ \bibinfo {author} {\bibfnamefont {M.}~\bibnamefont
  {Cerezo}},\ }\bibfield  {title} {\bibinfo {title} {Showcasing a barren
  plateau theory beyond the dynamical lie algebra},\ }\href
  {https://arxiv.org/abs/2310.11505} {\bibfield  {journal} {\bibinfo  {journal}
  {arXiv preprint arXiv:2310.11505}\ } (\bibinfo {year}
  {2023}{\natexlab{a}})}\BibitemShut {NoStop}%
\bibitem [{\citenamefont {Diaz}\ \emph
  {et~al.}(2023{\natexlab{b}})\citenamefont {Diaz}, \citenamefont {Braccia},
  \citenamefont {Larocca}, \citenamefont {Matera}, \citenamefont {Rossignoli},\
  and\ \citenamefont {Cerezo}}]{diaz2023parallel}%
  \BibitemOpen
  \bibfield  {author} {\bibinfo {author} {\bibfnamefont {N.~L.}\ \bibnamefont
  {Diaz}}, \bibinfo {author} {\bibfnamefont {P.}~\bibnamefont {Braccia}},
  \bibinfo {author} {\bibfnamefont {M.}~\bibnamefont {Larocca}}, \bibinfo
  {author} {\bibfnamefont {J.~M.}\ \bibnamefont {Matera}}, \bibinfo {author}
  {\bibfnamefont {R.}~\bibnamefont {Rossignoli}},\ and\ \bibinfo {author}
  {\bibfnamefont {M.}~\bibnamefont {Cerezo}},\ }\bibfield  {title} {\bibinfo
  {title} {Parallel-in-time quantum simulation via page and wootters quantum
  time},\ }\href {https://arxiv.org/abs/2308.12944} {\bibfield  {journal}
  {\bibinfo  {journal} {arXiv preprint arXiv:2308.12944}\ } (\bibinfo {year}
  {2023}{\natexlab{b}})}\BibitemShut {NoStop}%
\bibitem [{\citenamefont {Mele}\ and\ \citenamefont
  {Herasymenko}(2025)}]{mele2024efficient}%
  \BibitemOpen
  \bibfield  {author} {\bibinfo {author} {\bibfnamefont {A.~A.}\ \bibnamefont
  {Mele}}\ and\ \bibinfo {author} {\bibfnamefont {Y.}~\bibnamefont
  {Herasymenko}},\ }\bibfield  {title} {\bibinfo {title} {Efficient learning of
  quantum states prepared with few fermionic non-gaussian gates},\ }\href
  {https://doi.org/10.1103/PRXQuantum.6.010319} {\bibfield  {journal} {\bibinfo
   {journal} {PRX Quantum}\ }\textbf {\bibinfo {volume} {6}},\ \bibinfo {pages}
  {010319} (\bibinfo {year} {2025})}\BibitemShut {NoStop}%
\bibitem [{\citenamefont {K{\"o}kc{\"u}}\ \emph
  {et~al.}(2022{\natexlab{a}})\citenamefont {K{\"o}kc{\"u}}, \citenamefont
  {Steckmann}, \citenamefont {Wang}, \citenamefont {Freericks}, \citenamefont
  {Dumitrescu},\ and\ \citenamefont {Kemper}}]{kokcu2022fixed}%
  \BibitemOpen
  \bibfield  {author} {\bibinfo {author} {\bibfnamefont {E.}~\bibnamefont
  {K{\"o}kc{\"u}}}, \bibinfo {author} {\bibfnamefont {T.}~\bibnamefont
  {Steckmann}}, \bibinfo {author} {\bibfnamefont {Y.}~\bibnamefont {Wang}},
  \bibinfo {author} {\bibfnamefont {J.}~\bibnamefont {Freericks}}, \bibinfo
  {author} {\bibfnamefont {E.~F.}\ \bibnamefont {Dumitrescu}},\ and\ \bibinfo
  {author} {\bibfnamefont {A.~F.}\ \bibnamefont {Kemper}},\ }\bibfield  {title}
  {\bibinfo {title} {Fixed depth hamiltonian simulation via cartan
  decomposition},\ }\href {https://doi.org/10.1103/PhysRevLett.129.070501}
  {\bibfield  {journal} {\bibinfo  {journal} {Physical Review Letters}\
  }\textbf {\bibinfo {volume} {129}},\ \bibinfo {pages} {070501} (\bibinfo
  {year} {2022}{\natexlab{a}})}\BibitemShut {NoStop}%
\bibitem [{\citenamefont {K{\"o}kc{\"u}}\ \emph
  {et~al.}(2022{\natexlab{b}})\citenamefont {K{\"o}kc{\"u}}, \citenamefont
  {Camps}, \citenamefont {Oftelie}, \citenamefont {Freericks}, \citenamefont
  {de~Jong}, \citenamefont {Van~Beeumen},\ and\ \citenamefont
  {Kemper}}]{kokcu2022algebraic}%
  \BibitemOpen
  \bibfield  {author} {\bibinfo {author} {\bibfnamefont {E.}~\bibnamefont
  {K{\"o}kc{\"u}}}, \bibinfo {author} {\bibfnamefont {D.}~\bibnamefont
  {Camps}}, \bibinfo {author} {\bibfnamefont {L.~B.}\ \bibnamefont {Oftelie}},
  \bibinfo {author} {\bibfnamefont {J.~K.}\ \bibnamefont {Freericks}}, \bibinfo
  {author} {\bibfnamefont {W.~A.}\ \bibnamefont {de~Jong}}, \bibinfo {author}
  {\bibfnamefont {R.}~\bibnamefont {Van~Beeumen}},\ and\ \bibinfo {author}
  {\bibfnamefont {A.~F.}\ \bibnamefont {Kemper}},\ }\bibfield  {title}
  {\bibinfo {title} {Algebraic compression of quantum circuits for hamiltonian
  evolution},\ }\href {https://doi.org/10.1103/PhysRevA.105.032420} {\bibfield
  {journal} {\bibinfo  {journal} {Physical Review A}\ }\textbf {\bibinfo
  {volume} {105}},\ \bibinfo {pages} {032420} (\bibinfo {year}
  {2022}{\natexlab{b}})}\BibitemShut {NoStop}%
\bibitem [{\citenamefont {Guaita}\ \emph {et~al.}(2024)\citenamefont {Guaita},
  \citenamefont {Hackl},\ and\ \citenamefont
  {Quella}}]{guaita2024representation}%
  \BibitemOpen
  \bibfield  {author} {\bibinfo {author} {\bibfnamefont {T.}~\bibnamefont
  {Guaita}}, \bibinfo {author} {\bibfnamefont {L.}~\bibnamefont {Hackl}},\ and\
  \bibinfo {author} {\bibfnamefont {T.}~\bibnamefont {Quella}},\ }\bibfield
  {title} {\bibinfo {title} {Representation theory of gaussian unitary
  transformations for bosonic and fermionic systems},\ }\href
  {https://arxiv.org/abs/2409.11628} {\bibfield  {journal} {\bibinfo  {journal}
  {arXiv preprint arXiv:2409.11628}\ } (\bibinfo {year} {2024})}\BibitemShut
  {NoStop}%
\bibitem [{\citenamefont {Oszmaniec}\ \emph {et~al.}(2022)\citenamefont
  {Oszmaniec}, \citenamefont {Dangniam}, \citenamefont {Morales},\ and\
  \citenamefont {Zimbor{\'a}s}}]{oszmaniec2022fermion}%
  \BibitemOpen
  \bibfield  {author} {\bibinfo {author} {\bibfnamefont {M.}~\bibnamefont
  {Oszmaniec}}, \bibinfo {author} {\bibfnamefont {N.}~\bibnamefont {Dangniam}},
  \bibinfo {author} {\bibfnamefont {M.~E.}\ \bibnamefont {Morales}},\ and\
  \bibinfo {author} {\bibfnamefont {Z.}~\bibnamefont {Zimbor{\'a}s}},\
  }\bibfield  {title} {\bibinfo {title} {Fermion sampling: a robust quantum
  computational advantage scheme using fermionic linear optics and magic input
  states},\ }\href {https://doi.org/10.1103/PRXQuantum.3.020328} {\bibfield
  {journal} {\bibinfo  {journal} {PRX Quantum}\ }\textbf {\bibinfo {volume}
  {3}},\ \bibinfo {pages} {020328} (\bibinfo {year} {2022})}\BibitemShut
  {NoStop}%
\bibitem [{\citenamefont {Somma}\ \emph {et~al.}(2006)\citenamefont {Somma},
  \citenamefont {Barnum}, \citenamefont {Ortiz},\ and\ \citenamefont
  {Knill}}]{somma2006efficient}%
  \BibitemOpen
  \bibfield  {author} {\bibinfo {author} {\bibfnamefont {R.}~\bibnamefont
  {Somma}}, \bibinfo {author} {\bibfnamefont {H.}~\bibnamefont {Barnum}},
  \bibinfo {author} {\bibfnamefont {G.}~\bibnamefont {Ortiz}},\ and\ \bibinfo
  {author} {\bibfnamefont {E.}~\bibnamefont {Knill}},\ }\bibfield  {title}
  {\bibinfo {title} {Efficient solvability of {H}amiltonians and limits on the
  power of some quantum computational models},\ }\href
  {https://doi.org/https://doi.org/10.1103/PhysRevLett.97.190501} {\bibfield
  {journal} {\bibinfo  {journal} {Physical Review Letters}\ }\textbf {\bibinfo
  {volume} {97}},\ \bibinfo {pages} {190501} (\bibinfo {year}
  {2006})}\BibitemShut {NoStop}%
\bibitem [{\citenamefont {Goh}\ \emph {et~al.}(2023)\citenamefont {Goh},
  \citenamefont {Larocca}, \citenamefont {Cincio}, \citenamefont {Cerezo},\
  and\ \citenamefont {Sauvage}}]{goh2023lie}%
  \BibitemOpen
  \bibfield  {author} {\bibinfo {author} {\bibfnamefont {M.~L.}\ \bibnamefont
  {Goh}}, \bibinfo {author} {\bibfnamefont {M.}~\bibnamefont {Larocca}},
  \bibinfo {author} {\bibfnamefont {L.}~\bibnamefont {Cincio}}, \bibinfo
  {author} {\bibfnamefont {M.}~\bibnamefont {Cerezo}},\ and\ \bibinfo {author}
  {\bibfnamefont {F.}~\bibnamefont {Sauvage}},\ }\bibfield  {title} {\bibinfo
  {title} {Lie-algebraic classical simulations for quantum computing},\ }\href
  {https://arxiv.org/abs/2308.01432} {\bibfield  {journal} {\bibinfo  {journal}
  {arXiv preprint arXiv:2308.01432}\ } (\bibinfo {year} {2023})}\BibitemShut
  {NoStop}%
\bibitem [{\citenamefont {Miller}\ \emph {et~al.}(2025)\citenamefont {Miller},
  \citenamefont {Holmes}, \citenamefont {Salehi}, \citenamefont {Chakraborty},
  \citenamefont {Nyk{\"a}nen}, \citenamefont {Zimbor{\'a}s}, \citenamefont
  {Glos},\ and\ \citenamefont {Garc{\'\i}a-P{\'e}rez}}]{miller2025simulation}%
  \BibitemOpen
  \bibfield  {author} {\bibinfo {author} {\bibfnamefont {A.}~\bibnamefont
  {Miller}}, \bibinfo {author} {\bibfnamefont {Z.}~\bibnamefont {Holmes}},
  \bibinfo {author} {\bibfnamefont {{\"O}.}~\bibnamefont {Salehi}}, \bibinfo
  {author} {\bibfnamefont {R.}~\bibnamefont {Chakraborty}}, \bibinfo {author}
  {\bibfnamefont {A.}~\bibnamefont {Nyk{\"a}nen}}, \bibinfo {author}
  {\bibfnamefont {Z.}~\bibnamefont {Zimbor{\'a}s}}, \bibinfo {author}
  {\bibfnamefont {A.}~\bibnamefont {Glos}},\ and\ \bibinfo {author}
  {\bibfnamefont {G.}~\bibnamefont {Garc{\'\i}a-P{\'e}rez}},\ }\bibfield
  {title} {\bibinfo {title} {Simulation of fermionic circuits using majorana
  propagation},\ }\href {https://doi.org/10.48550/arXiv.2503.18939} {\bibfield
  {journal} {\bibinfo  {journal} {arXiv preprint arXiv:2503.18939}\ } (\bibinfo
  {year} {2025})}\BibitemShut {NoStop}%
\bibitem [{\citenamefont {Oszmaniec}\ and\ \citenamefont
  {Zimbor{\'a}s}(2017)}]{oszmaniec2017universal}%
  \BibitemOpen
  \bibfield  {author} {\bibinfo {author} {\bibfnamefont {M.}~\bibnamefont
  {Oszmaniec}}\ and\ \bibinfo {author} {\bibfnamefont {Z.}~\bibnamefont
  {Zimbor{\'a}s}},\ }\bibfield  {title} {\bibinfo {title} {Universal extensions
  of restricted classes of quantum operations},\ }\href
  {https://doi.org/10.1103/PhysRevLett.119.220502} {\bibfield  {journal}
  {\bibinfo  {journal} {Physical review letters}\ }\textbf {\bibinfo {volume}
  {119}},\ \bibinfo {pages} {220502} (\bibinfo {year} {2017})}\BibitemShut
  {NoStop}%
\bibitem [{\citenamefont {Zhao}\ \emph {et~al.}(2021)\citenamefont {Zhao},
  \citenamefont {Rubin},\ and\ \citenamefont {Miyake}}]{zhao2021fermionic}%
  \BibitemOpen
  \bibfield  {author} {\bibinfo {author} {\bibfnamefont {A.}~\bibnamefont
  {Zhao}}, \bibinfo {author} {\bibfnamefont {N.~C.}\ \bibnamefont {Rubin}},\
  and\ \bibinfo {author} {\bibfnamefont {A.}~\bibnamefont {Miyake}},\
  }\bibfield  {title} {\bibinfo {title} {Fermionic partial tomography via
  classical shadows},\ }\href {https://doi.org/10.1103/PhysRevLett.127.110504}
  {\bibfield  {journal} {\bibinfo  {journal} {Physical Review Letters}\
  }\textbf {\bibinfo {volume} {127}},\ \bibinfo {pages} {110504} (\bibinfo
  {year} {2021})}\BibitemShut {NoStop}%
\bibitem [{\citenamefont {Heyraud}\ \emph {et~al.}(2024)\citenamefont
  {Heyraud}, \citenamefont {Chomet},\ and\ \citenamefont
  {Tilly}}]{heyraud2024unified}%
  \BibitemOpen
  \bibfield  {author} {\bibinfo {author} {\bibfnamefont {V.}~\bibnamefont
  {Heyraud}}, \bibinfo {author} {\bibfnamefont {H.}~\bibnamefont {Chomet}},\
  and\ \bibinfo {author} {\bibfnamefont {J.}~\bibnamefont {Tilly}},\ }\bibfield
   {title} {\bibinfo {title} {Unified framework for matchgate classical
  shadows},\ }\bibfield  {journal} {\bibinfo  {journal} {arXiv preprint
  arXiv:2409.03836}\ }\href {https://doi.org/10.48550/arXiv.2409.03836}
  {10.48550/arXiv.2409.03836} (\bibinfo {year} {2024})\BibitemShut {NoStop}%
\bibitem [{\citenamefont {Mezzadri}(2007)}]{mezzadri2006generate}%
  \BibitemOpen
  \bibfield  {author} {\bibinfo {author} {\bibfnamefont {F.}~\bibnamefont
  {Mezzadri}},\ }\bibfield  {title} {{\selectlanguage {English}\bibinfo {title}
  {How to generate random matrices from the classical compact groups}},\ }\href
  {https://research-information.bris.ac.uk/en/publications/how-to-generate-random-matrices-from-the-classical-compact-groups}
  {\bibfield  {journal} {\bibinfo  {journal} {Notices of the American
  Mathematical Society}\ }\textbf {\bibinfo {volume} {54}},\ \bibinfo {pages}
  {592 } (\bibinfo {year} {2007})}\BibitemShut {NoStop}%
\bibitem [{\citenamefont {Diaconis}\ and\ \citenamefont
  {Forrester}(2017)}]{diaconis2017hurwitz}%
  \BibitemOpen
  \bibfield  {author} {\bibinfo {author} {\bibfnamefont {P.}~\bibnamefont
  {Diaconis}}\ and\ \bibinfo {author} {\bibfnamefont {P.~J.}\ \bibnamefont
  {Forrester}},\ }\bibfield  {title} {\bibinfo {title} {Hurwitz and the origins
  of random matrix theory in mathematics},\ }\href
  {https://doi.org/10.1142/S2010326317300017} {\bibfield  {journal} {\bibinfo
  {journal} {Random Matrices: Theory and Applications}\ }\textbf {\bibinfo
  {volume} {6}},\ \bibinfo {pages} {1730001} (\bibinfo {year}
  {2017})}\BibitemShut {NoStop}%
\bibitem [{\citenamefont {Zyczkowski}\ and\ \citenamefont
  {Kus}(1994)}]{zyczkowski1994random}%
  \BibitemOpen
  \bibfield  {author} {\bibinfo {author} {\bibfnamefont {K.}~\bibnamefont
  {Zyczkowski}}\ and\ \bibinfo {author} {\bibfnamefont {M.}~\bibnamefont
  {Kus}},\ }\bibfield  {title} {\bibinfo {title} {Random unitary matrices},\
  }\href {https://doi.org/10.1088/0305-4470/27/12/028} {\bibfield  {journal}
  {\bibinfo  {journal} {Journal of Physics A: Mathematical and General}\
  }\textbf {\bibinfo {volume} {27}},\ \bibinfo {pages} {4235} (\bibinfo {year}
  {1994})}\BibitemShut {NoStop}%
\bibitem [{\citenamefont
  {Gottesman}(1998)}]{gottesman1998heisenbergrepresentation}%
  \BibitemOpen
  \bibfield  {author} {\bibinfo {author} {\bibfnamefont {D.}~\bibnamefont
  {Gottesman}},\ }\bibfield  {title} {\bibinfo {title} {The heisenberg
  representation of quantum computers},\ }\href
  {https://arxiv.org/abs/quant-ph/9807006} {\bibfield  {journal} {\bibinfo
  {journal} {arXiv preprint quant-ph/9807006}\ } (\bibinfo {year}
  {1998})}\BibitemShut {NoStop}%
\bibitem [{\citenamefont {Aaronson}\ and\ \citenamefont
  {Gottesman}(2004)}]{aaronson2004improved}%
  \BibitemOpen
  \bibfield  {author} {\bibinfo {author} {\bibfnamefont {S.}~\bibnamefont
  {Aaronson}}\ and\ \bibinfo {author} {\bibfnamefont {D.}~\bibnamefont
  {Gottesman}},\ }\bibfield  {title} {\bibinfo {title} {Improved simulation of
  stabilizer circuits},\ }\href {https://doi.org/10.1103/PhysRevA.70.052328}
  {\bibfield  {journal} {\bibinfo  {journal} {Physical Review A}\ }\textbf
  {\bibinfo {volume} {70}},\ \bibinfo {pages} {052328} (\bibinfo {year}
  {2004})}\BibitemShut {NoStop}%
\bibitem [{\citenamefont {Camps}\ \emph {et~al.}(2022)\citenamefont {Camps},
  \citenamefont {K{\"o}kc{\"u}}, \citenamefont {Bassman~Oftelie}, \citenamefont
  {De~Jong}, \citenamefont {Kemper},\ and\ \citenamefont
  {Van~Beeumen}}]{camps2022algebraic}%
  \BibitemOpen
  \bibfield  {author} {\bibinfo {author} {\bibfnamefont {D.}~\bibnamefont
  {Camps}}, \bibinfo {author} {\bibfnamefont {E.}~\bibnamefont
  {K{\"o}kc{\"u}}}, \bibinfo {author} {\bibfnamefont {L.}~\bibnamefont
  {Bassman~Oftelie}}, \bibinfo {author} {\bibfnamefont {W.~A.}\ \bibnamefont
  {De~Jong}}, \bibinfo {author} {\bibfnamefont {A.~F.}\ \bibnamefont
  {Kemper}},\ and\ \bibinfo {author} {\bibfnamefont {R.}~\bibnamefont
  {Van~Beeumen}},\ }\bibfield  {title} {\bibinfo {title} {An algebraic quantum
  circuit compression algorithm for hamiltonian simulation},\ }\href
  {https://doi.org/10.1137/21M1439298} {\bibfield  {journal} {\bibinfo
  {journal} {SIAM Journal on Matrix Analysis and Applications}\ }\textbf
  {\bibinfo {volume} {43}},\ \bibinfo {pages} {1084} (\bibinfo {year}
  {2022})}\BibitemShut {NoStop}%
\bibitem [{\citenamefont {Foxen}\ \emph {et~al.}(2020)\citenamefont {Foxen},
  \citenamefont {Neill}, \citenamefont {Dunsworth}, \citenamefont {Roushan},
  \citenamefont {Chiaro}, \citenamefont {Megrant}, \citenamefont {Kelly},
  \citenamefont {Chen}, \citenamefont {Satzinger}, \citenamefont {Barends}
  \emph {et~al.}}]{foxen2020demonstrating}%
  \BibitemOpen
  \bibfield  {author} {\bibinfo {author} {\bibfnamefont {B.}~\bibnamefont
  {Foxen}}, \bibinfo {author} {\bibfnamefont {C.}~\bibnamefont {Neill}},
  \bibinfo {author} {\bibfnamefont {A.}~\bibnamefont {Dunsworth}}, \bibinfo
  {author} {\bibfnamefont {P.}~\bibnamefont {Roushan}}, \bibinfo {author}
  {\bibfnamefont {B.}~\bibnamefont {Chiaro}}, \bibinfo {author} {\bibfnamefont
  {A.}~\bibnamefont {Megrant}}, \bibinfo {author} {\bibfnamefont
  {J.}~\bibnamefont {Kelly}}, \bibinfo {author} {\bibfnamefont
  {Z.}~\bibnamefont {Chen}}, \bibinfo {author} {\bibfnamefont {K.}~\bibnamefont
  {Satzinger}}, \bibinfo {author} {\bibfnamefont {R.}~\bibnamefont {Barends}},
  \emph {et~al.},\ }\bibfield  {title} {\bibinfo {title} {Demonstrating a
  continuous set of two-qubit gates for near-term quantum algorithms},\ }\href
  {https://doi.org/10.1103/PhysRevLett.125.120504} {\bibfield  {journal}
  {\bibinfo  {journal} {Physical Review Letters}\ }\textbf {\bibinfo {volume}
  {125}},\ \bibinfo {pages} {120504} (\bibinfo {year} {2020})}\BibitemShut
  {NoStop}%
\bibitem [{\citenamefont {McKay}\ \emph {et~al.}(2017)\citenamefont {McKay},
  \citenamefont {Wood}, \citenamefont {Sheldon}, \citenamefont {Chow},\ and\
  \citenamefont {Gambetta}}]{mckay2017efficient}%
  \BibitemOpen
  \bibfield  {author} {\bibinfo {author} {\bibfnamefont {D.~C.}\ \bibnamefont
  {McKay}}, \bibinfo {author} {\bibfnamefont {C.~J.}\ \bibnamefont {Wood}},
  \bibinfo {author} {\bibfnamefont {S.}~\bibnamefont {Sheldon}}, \bibinfo
  {author} {\bibfnamefont {J.~M.}\ \bibnamefont {Chow}},\ and\ \bibinfo
  {author} {\bibfnamefont {J.~M.}\ \bibnamefont {Gambetta}},\ }\bibfield
  {title} {\bibinfo {title} {Efficient $z$ gates for quantum computing},\
  }\href {https://doi.org/10.1103/PhysRevA.96.022330} {\bibfield  {journal}
  {\bibinfo  {journal} {Phys. Rev. A}\ }\textbf {\bibinfo {volume} {96}},\
  \bibinfo {pages} {022330} (\bibinfo {year} {2017})}\BibitemShut {NoStop}%
\bibitem [{\citenamefont {Bennett}\ \emph {et~al.}(1996)\citenamefont
  {Bennett}, \citenamefont {DiVincenzo}, \citenamefont {Smolin},\ and\
  \citenamefont {Wootters}}]{bennett1996mixed}%
  \BibitemOpen
  \bibfield  {author} {\bibinfo {author} {\bibfnamefont {C.~H.}\ \bibnamefont
  {Bennett}}, \bibinfo {author} {\bibfnamefont {D.~P.}\ \bibnamefont
  {DiVincenzo}}, \bibinfo {author} {\bibfnamefont {J.~A.}\ \bibnamefont
  {Smolin}},\ and\ \bibinfo {author} {\bibfnamefont {W.~K.}\ \bibnamefont
  {Wootters}},\ }\bibfield  {title} {\bibinfo {title} {Mixed-state entanglement
  and quantum error correction},\ }\href
  {https://doi.org/10.1103/PhysRevA.54.3824} {\bibfield  {journal} {\bibinfo
  {journal} {Physical Review A}\ }\textbf {\bibinfo {volume} {54}},\ \bibinfo
  {pages} {3824} (\bibinfo {year} {1996})}\BibitemShut {NoStop}%
\bibitem [{\citenamefont {Calderbank}\ \emph {et~al.}(1997)\citenamefont
  {Calderbank}, \citenamefont {Rains}, \citenamefont {Shor},\ and\
  \citenamefont {Sloane}}]{calderbank1997quantum}%
  \BibitemOpen
  \bibfield  {author} {\bibinfo {author} {\bibfnamefont {A.~R.}\ \bibnamefont
  {Calderbank}}, \bibinfo {author} {\bibfnamefont {E.~M.}\ \bibnamefont
  {Rains}}, \bibinfo {author} {\bibfnamefont {P.~W.}\ \bibnamefont {Shor}},\
  and\ \bibinfo {author} {\bibfnamefont {N.~J.}\ \bibnamefont {Sloane}},\
  }\bibfield  {title} {\bibinfo {title} {Quantum error correction and
  orthogonal geometry},\ }\href {https://doi.org/10.1103/PhysRevLett.78.405}
  {\bibfield  {journal} {\bibinfo  {journal} {Physical Review Letters}\
  }\textbf {\bibinfo {volume} {78}},\ \bibinfo {pages} {405} (\bibinfo {year}
  {1997})}\BibitemShut {NoStop}%
\bibitem [{\citenamefont {Huang}\ \emph {et~al.}(2020)\citenamefont {Huang},
  \citenamefont {Kueng},\ and\ \citenamefont {Preskill}}]{huang2020predicting}%
  \BibitemOpen
  \bibfield  {author} {\bibinfo {author} {\bibfnamefont {H.-Y.}\ \bibnamefont
  {Huang}}, \bibinfo {author} {\bibfnamefont {R.}~\bibnamefont {Kueng}},\ and\
  \bibinfo {author} {\bibfnamefont {J.}~\bibnamefont {Preskill}},\ }\bibfield
  {title} {\bibinfo {title} {Predicting many properties of a quantum system
  from very few measurements},\ }\href
  {https://doi.org/10.1038/s41567-020-0932-7} {\bibfield  {journal} {\bibinfo
  {journal} {Nature Physics}\ }\textbf {\bibinfo {volume} {16}},\ \bibinfo
  {pages} {1050} (\bibinfo {year} {2020})}\BibitemShut {NoStop}%
\bibitem [{\citenamefont {West}\ \emph {et~al.}(2024)\citenamefont {West},
  \citenamefont {Mele}, \citenamefont {Larocca},\ and\ \citenamefont
  {Cerezo}}]{west2024real}%
  \BibitemOpen
  \bibfield  {author} {\bibinfo {author} {\bibfnamefont {M.}~\bibnamefont
  {West}}, \bibinfo {author} {\bibfnamefont {A.~A.}\ \bibnamefont {Mele}},
  \bibinfo {author} {\bibfnamefont {M.}~\bibnamefont {Larocca}},\ and\ \bibinfo
  {author} {\bibfnamefont {M.}~\bibnamefont {Cerezo}},\ }\bibfield  {title}
  {\bibinfo {title} {Real classical shadows},\ }\href
  {https://arxiv.org/abs/2410.23481} {\bibfield  {journal} {\bibinfo  {journal}
  {arXiv preprint arXiv:2410.23481}\ } (\bibinfo {year} {2024})}\BibitemShut
  {NoStop}%
\bibitem [{\citenamefont {Webb}(2016)}]{webb2016clifford}%
  \BibitemOpen
  \bibfield  {author} {\bibinfo {author} {\bibfnamefont {Z.}~\bibnamefont
  {Webb}},\ }\bibfield  {title} {\bibinfo {title} {The clifford group forms a
  unitary 3-design},\ }\href {https://doi.org/10.26421/QIC16.15-16-8}
  {\bibfield  {journal} {\bibinfo  {journal} {Quantum Information and
  Computation}\ }\textbf {\bibinfo {volume} {16}},\ \bibinfo {pages} {1379}
  (\bibinfo {year} {2016})}\BibitemShut {NoStop}%
\bibitem [{\citenamefont {Zhu}(2017)}]{zhu2017multiqubit}%
  \BibitemOpen
  \bibfield  {author} {\bibinfo {author} {\bibfnamefont {H.}~\bibnamefont
  {Zhu}},\ }\bibfield  {title} {\bibinfo {title} {Multiqubit clifford groups
  are unitary 3-designs},\ }\href {https://doi.org/10.1103/PhysRevA.96.062336}
  {\bibfield  {journal} {\bibinfo  {journal} {Physical Review A}\ }\textbf
  {\bibinfo {volume} {96}},\ \bibinfo {pages} {062336} (\bibinfo {year}
  {2017})}\BibitemShut {NoStop}%
\bibitem [{\citenamefont {Kueng}\ and\ \citenamefont
  {Gross}(2015)}]{kueng2015qubit}%
  \BibitemOpen
  \bibfield  {author} {\bibinfo {author} {\bibfnamefont {R.}~\bibnamefont
  {Kueng}}\ and\ \bibinfo {author} {\bibfnamefont {D.}~\bibnamefont {Gross}},\
  }\bibfield  {title} {\bibinfo {title} {Qubit stabilizer states are complex
  projective 3-designs},\ }\href {https://arxiv.org/abs/1510.02767} {\bibfield
  {journal} {\bibinfo  {journal} {arXiv preprint arXiv:1510.02767}\ } (\bibinfo
  {year} {2015})}\BibitemShut {NoStop}%
\bibitem [{\citenamefont {Zhu}\ \emph {et~al.}(2016)\citenamefont {Zhu},
  \citenamefont {Kueng}, \citenamefont {Grassl},\ and\ \citenamefont
  {Gross}}]{zhu2016clifford}%
  \BibitemOpen
  \bibfield  {author} {\bibinfo {author} {\bibfnamefont {H.}~\bibnamefont
  {Zhu}}, \bibinfo {author} {\bibfnamefont {R.}~\bibnamefont {Kueng}}, \bibinfo
  {author} {\bibfnamefont {M.}~\bibnamefont {Grassl}},\ and\ \bibinfo {author}
  {\bibfnamefont {D.}~\bibnamefont {Gross}},\ }\bibfield  {title} {\bibinfo
  {title} {The clifford group fails gracefully to be a unitary 4-design},\
  }\href {https://arxiv.org/abs/1609.08172} {\bibfield  {journal} {\bibinfo
  {journal} {arXiv preprint arXiv:1609.08172}\ } (\bibinfo {year}
  {2016})}\BibitemShut {NoStop}%
\bibitem [{\citenamefont {Hashagen}\ \emph {et~al.}(2018)\citenamefont
  {Hashagen}, \citenamefont {Flammia}, \citenamefont {Gross},\ and\
  \citenamefont {Wallman}}]{hashagen2018real}%
  \BibitemOpen
  \bibfield  {author} {\bibinfo {author} {\bibfnamefont {A.}~\bibnamefont
  {Hashagen}}, \bibinfo {author} {\bibfnamefont {S.}~\bibnamefont {Flammia}},
  \bibinfo {author} {\bibfnamefont {D.}~\bibnamefont {Gross}},\ and\ \bibinfo
  {author} {\bibfnamefont {J.}~\bibnamefont {Wallman}},\ }\bibfield  {title}
  {\bibinfo {title} {Real randomized benchmarking},\ }\href
  {https://doi.org/10.22331/q-2018-08-22-85} {\bibfield  {journal} {\bibinfo
  {journal} {Quantum}\ }\textbf {\bibinfo {volume} {2}},\ \bibinfo {pages} {85}
  (\bibinfo {year} {2018})}\BibitemShut {NoStop}%
\bibitem [{\citenamefont {Mitsuhashi}\ and\ \citenamefont
  {Yoshioka}(2023)}]{mitsuhashi2023clifford}%
  \BibitemOpen
  \bibfield  {author} {\bibinfo {author} {\bibfnamefont {Y.}~\bibnamefont
  {Mitsuhashi}}\ and\ \bibinfo {author} {\bibfnamefont {N.}~\bibnamefont
  {Yoshioka}},\ }\bibfield  {title} {\bibinfo {title} {Clifford group and
  unitary designs under symmetry},\ }\href
  {https://doi.org/10.1103/PRXQuantum.4.040331} {\bibfield  {journal} {\bibinfo
   {journal} {PRX Quantum}\ }\textbf {\bibinfo {volume} {4}},\ \bibinfo {pages}
  {040331} (\bibinfo {year} {2023})}\BibitemShut {NoStop}%
\bibitem [{\citenamefont {Gottesman}(1997)}]{gottesman1997stabilizer}%
  \BibitemOpen
  \bibfield  {author} {\bibinfo {author} {\bibfnamefont {D.}~\bibnamefont
  {Gottesman}},\ }\href@noop {} {\emph {\bibinfo {title} {Stabilizer codes and
  quantum error correction}}}\ (\bibinfo  {publisher} {California Institute of
  Technology},\ \bibinfo {year} {1997})\BibitemShut {NoStop}%
\bibitem [{\citenamefont {Fowler}\ \emph {et~al.}(2012)\citenamefont {Fowler},
  \citenamefont {Mariantoni}, \citenamefont {Martinis},\ and\ \citenamefont
  {Cleland}}]{fowler2012surface}%
  \BibitemOpen
  \bibfield  {author} {\bibinfo {author} {\bibfnamefont {A.~G.}\ \bibnamefont
  {Fowler}}, \bibinfo {author} {\bibfnamefont {M.}~\bibnamefont {Mariantoni}},
  \bibinfo {author} {\bibfnamefont {J.~M.}\ \bibnamefont {Martinis}},\ and\
  \bibinfo {author} {\bibfnamefont {A.~N.}\ \bibnamefont {Cleland}},\
  }\bibfield  {title} {\bibinfo {title} {Surface codes: Towards practical
  large-scale quantum computation},\ }\href
  {https://doi.org/10.1103/PhysRevA.86.032324} {\bibfield  {journal} {\bibinfo
  {journal} {Physical Review A}\ }\textbf {\bibinfo {volume} {86}},\ \bibinfo
  {pages} {032324} (\bibinfo {year} {2012})}\BibitemShut {NoStop}%
\bibitem [{\citenamefont {Mele}(2024)}]{mele2023introduction}%
  \BibitemOpen
  \bibfield  {author} {\bibinfo {author} {\bibfnamefont {A.~A.}\ \bibnamefont
  {Mele}},\ }\bibfield  {title} {\bibinfo {title} {Introduction to haar measure
  tools in quantum information: A beginner's tutorial},\ }\href
  {https://doi.org/10.22331/q-2024-05-08-1340} {\bibfield  {journal} {\bibinfo
  {journal} {Quantum}\ }\textbf {\bibinfo {volume} {8}},\ \bibinfo {pages}
  {1340} (\bibinfo {year} {2024})}\BibitemShut {NoStop}%
\bibitem [{\citenamefont {Ragone}\ \emph {et~al.}(2022)\citenamefont {Ragone},
  \citenamefont {Nguyen}, \citenamefont {Schatzki}, \citenamefont {Braccia},
  \citenamefont {Larocca}, \citenamefont {Sauvage}, \citenamefont {Coles},\
  and\ \citenamefont {Cerezo}}]{ragone2022representation}%
  \BibitemOpen
  \bibfield  {author} {\bibinfo {author} {\bibfnamefont {M.}~\bibnamefont
  {Ragone}}, \bibinfo {author} {\bibfnamefont {Q.~T.}\ \bibnamefont {Nguyen}},
  \bibinfo {author} {\bibfnamefont {L.}~\bibnamefont {Schatzki}}, \bibinfo
  {author} {\bibfnamefont {P.}~\bibnamefont {Braccia}}, \bibinfo {author}
  {\bibfnamefont {M.}~\bibnamefont {Larocca}}, \bibinfo {author} {\bibfnamefont
  {F.}~\bibnamefont {Sauvage}}, \bibinfo {author} {\bibfnamefont {P.~J.}\
  \bibnamefont {Coles}},\ and\ \bibinfo {author} {\bibfnamefont
  {M.}~\bibnamefont {Cerezo}},\ }\bibfield  {title} {\bibinfo {title}
  {Representation theory for geometric quantum machine learning},\ }\href
  {https://arxiv.org/abs/2210.07980} {\bibfield  {journal} {\bibinfo  {journal}
  {arXiv preprint arXiv:2210.07980}\ } (\bibinfo {year} {2022})}\BibitemShut
  {NoStop}%
\bibitem [{\citenamefont {Wiersema}\ \emph {et~al.}(2024)\citenamefont
  {Wiersema}, \citenamefont {K{\"o}kc{\"u}}, \citenamefont {Kemper},\ and\
  \citenamefont {Bakalov}}]{wiersema2023classification}%
  \BibitemOpen
  \bibfield  {author} {\bibinfo {author} {\bibfnamefont {R.}~\bibnamefont
  {Wiersema}}, \bibinfo {author} {\bibfnamefont {E.}~\bibnamefont
  {K{\"o}kc{\"u}}}, \bibinfo {author} {\bibfnamefont {A.~F.}\ \bibnamefont
  {Kemper}},\ and\ \bibinfo {author} {\bibfnamefont {B.~N.}\ \bibnamefont
  {Bakalov}},\ }\bibfield  {title} {\bibinfo {title} {Classification of
  dynamical lie algebras of 2-local spin systems on linear, circular and fully
  connected topologies},\ }\href {https://doi.org/10.1038/s41534-024-00900-2}
  {\bibfield  {journal} {\bibinfo  {journal} {npj Quantum Information}\
  }\textbf {\bibinfo {volume} {10}},\ \bibinfo {pages} {110} (\bibinfo {year}
  {2024})}\BibitemShut {NoStop}%
\bibitem [{\citenamefont {K{\"o}kc{\"u}}\ \emph {et~al.}(2024)\citenamefont
  {K{\"o}kc{\"u}}, \citenamefont {Wiersema}, \citenamefont {Kemper},\ and\
  \citenamefont {Bakalov}}]{kokcu2024classification}%
  \BibitemOpen
  \bibfield  {author} {\bibinfo {author} {\bibfnamefont {E.}~\bibnamefont
  {K{\"o}kc{\"u}}}, \bibinfo {author} {\bibfnamefont {R.}~\bibnamefont
  {Wiersema}}, \bibinfo {author} {\bibfnamefont {A.~F.}\ \bibnamefont
  {Kemper}},\ and\ \bibinfo {author} {\bibfnamefont {B.~N.}\ \bibnamefont
  {Bakalov}},\ }\bibfield  {title} {\bibinfo {title} {Classification of
  dynamical lie algebras generated by spin interactions on undirected graphs},\
  }\bibfield  {journal} {\bibinfo  {journal} {arXiv preprint arXiv:2409.19797}\
  }\href {https://doi.org/10.48550/arXiv.2409.19797}
  {10.48550/arXiv.2409.19797} (\bibinfo {year} {2024})\BibitemShut {NoStop}%
\bibitem [{\citenamefont {Aguilar}\ \emph {et~al.}(2024)\citenamefont
  {Aguilar}, \citenamefont {Cichy}, \citenamefont {Eisert},\ and\ \citenamefont
  {Bittel}}]{aguilar2024full}%
  \BibitemOpen
  \bibfield  {author} {\bibinfo {author} {\bibfnamefont {G.}~\bibnamefont
  {Aguilar}}, \bibinfo {author} {\bibfnamefont {S.}~\bibnamefont {Cichy}},
  \bibinfo {author} {\bibfnamefont {J.}~\bibnamefont {Eisert}},\ and\ \bibinfo
  {author} {\bibfnamefont {L.}~\bibnamefont {Bittel}},\ }\bibfield  {title}
  {\bibinfo {title} {Full classification of pauli lie algebras},\ }\href
  {https://arxiv.org/abs/2408.00081} {\bibfield  {journal} {\bibinfo  {journal}
  {arXiv preprint arXiv:2408.00081}\ } (\bibinfo {year} {2024})}\BibitemShut
  {NoStop}%
\bibitem [{\citenamefont {Schatzki}\ \emph {et~al.}(2024)\citenamefont
  {Schatzki}, \citenamefont {Larocca}, \citenamefont {Nguyen}, \citenamefont
  {Sauvage},\ and\ \citenamefont {Cerezo}}]{schatzki2022theoretical}%
  \BibitemOpen
  \bibfield  {author} {\bibinfo {author} {\bibfnamefont {L.}~\bibnamefont
  {Schatzki}}, \bibinfo {author} {\bibfnamefont {M.}~\bibnamefont {Larocca}},
  \bibinfo {author} {\bibfnamefont {Q.~T.}\ \bibnamefont {Nguyen}}, \bibinfo
  {author} {\bibfnamefont {F.}~\bibnamefont {Sauvage}},\ and\ \bibinfo {author}
  {\bibfnamefont {M.}~\bibnamefont {Cerezo}},\ }\bibfield  {title} {\bibinfo
  {title} {Theoretical guarantees for permutation-equivariant quantum neural
  networks},\ }\href {https://doi.org/10.1038/s41534-024-00804-1} {\bibfield
  {journal} {\bibinfo  {journal} {npj Quantum Information}\ }\textbf {\bibinfo
  {volume} {10}},\ \bibinfo {pages} {12} (\bibinfo {year} {2024})}\BibitemShut
  {NoStop}%
\bibitem [{\citenamefont {Kazi}\ \emph {et~al.}(2024)\citenamefont {Kazi},
  \citenamefont {Larocca},\ and\ \citenamefont
  {Cerezo}}]{kazi2023universality}%
  \BibitemOpen
  \bibfield  {author} {\bibinfo {author} {\bibfnamefont {S.}~\bibnamefont
  {Kazi}}, \bibinfo {author} {\bibfnamefont {M.}~\bibnamefont {Larocca}},\ and\
  \bibinfo {author} {\bibfnamefont {M.}~\bibnamefont {Cerezo}},\ }\bibfield
  {title} {\bibinfo {title} {On the universality of $s_n$-equivariant $ k
  $-body gates},\ }\href {https://doi.org/10.1088/1367-2630/ad4819} {\bibfield
  {journal} {\bibinfo  {journal} {New Journal of Physics}\ }\textbf {\bibinfo
  {volume} {26}},\ \bibinfo {pages} {053030} (\bibinfo {year}
  {2024})}\BibitemShut {NoStop}%
\bibitem [{\citenamefont {Marvian}(2022)}]{marvian2022restrictions}%
  \BibitemOpen
  \bibfield  {author} {\bibinfo {author} {\bibfnamefont {I.}~\bibnamefont
  {Marvian}},\ }\bibfield  {title} {\bibinfo {title} {Restrictions on
  realizable unitary operations imposed by symmetry and locality},\ }\href
  {https://www.nature.com/articles/s41567-021-01464-0} {\bibfield  {journal}
  {\bibinfo  {journal} {Nature Physics}\ }\textbf {\bibinfo {volume} {18}},\
  \bibinfo {pages} {283} (\bibinfo {year} {2022})}\BibitemShut {NoStop}%
\bibitem [{\citenamefont {Spengler}\ \emph {et~al.}(2012)\citenamefont
  {Spengler}, \citenamefont {Huber},\ and\ \citenamefont
  {Hiesmayr}}]{spengler2012composite}%
  \BibitemOpen
  \bibfield  {author} {\bibinfo {author} {\bibfnamefont {C.}~\bibnamefont
  {Spengler}}, \bibinfo {author} {\bibfnamefont {M.}~\bibnamefont {Huber}},\
  and\ \bibinfo {author} {\bibfnamefont {B.~C.}\ \bibnamefont {Hiesmayr}},\
  }\bibfield  {title} {\bibinfo {title} {Composite parameterization and haar
  measure for all unitary and special unitary groups},\ }\bibfield  {journal}
  {\bibinfo  {journal} {Journal of Mathematical Physics}\ }\textbf {\bibinfo
  {volume} {53}},\ \href {https://doi.org/10.1063/1.3672064}
  {10.1063/1.3672064} (\bibinfo {year} {2012})\BibitemShut {NoStop}%
\bibitem [{\citenamefont {Zeier}\ and\ \citenamefont
  {Schulte-Herbr{\"u}ggen}(2011)}]{zeier2011symmetry}%
  \BibitemOpen
  \bibfield  {author} {\bibinfo {author} {\bibfnamefont {R.}~\bibnamefont
  {Zeier}}\ and\ \bibinfo {author} {\bibfnamefont {T.}~\bibnamefont
  {Schulte-Herbr{\"u}ggen}},\ }\bibfield  {title} {\bibinfo {title} {Symmetry
  principles in quantum systems theory},\ }\href
  {https://doi.org/https://doi.org/10.1063/1.3657939} {\bibfield  {journal}
  {\bibinfo  {journal} {Journal of mathematical physics}\ }\textbf {\bibinfo
  {volume} {52}},\ \bibinfo {pages} {113510} (\bibinfo {year}
  {2011})}\BibitemShut {NoStop}%
\bibitem [{\citenamefont {Hall}(2013)}]{hall2013lie}%
  \BibitemOpen
  \bibfield  {author} {\bibinfo {author} {\bibfnamefont {B.~C.}\ \bibnamefont
  {Hall}},\ }\href@noop {} {\emph {\bibinfo {title} {Lie groups, Lie algebras,
  and representations}}}\ (\bibinfo  {publisher} {Springer},\ \bibinfo {year}
  {2013})\BibitemShut {NoStop}%
\bibitem [{\citenamefont {Knight}(1995)}]{knight1995fast}%
  \BibitemOpen
  \bibfield  {author} {\bibinfo {author} {\bibfnamefont {P.~A.}\ \bibnamefont
  {Knight}},\ }\bibfield  {title} {\bibinfo {title} {Fast rectangular matrix
  multiplication and qr decomposition},\ }\href
  {https://doi.org/10.1016/0024-3795(93)00230-W} {\bibfield  {journal}
  {\bibinfo  {journal} {Linear algebra and its applications}\ }\textbf
  {\bibinfo {volume} {221}},\ \bibinfo {pages} {69} (\bibinfo {year}
  {1995})}\BibitemShut {NoStop}%
\bibitem [{\citenamefont {Strassen}(1969)}]{strassen1969gaussian}%
  \BibitemOpen
  \bibfield  {author} {\bibinfo {author} {\bibfnamefont {V.}~\bibnamefont
  {Strassen}},\ }\bibfield  {title} {\bibinfo {title} {Gaussian elimination is
  not optimal},\ }\href {https://doi.org/10.1007/BF02165411} {\bibfield
  {journal} {\bibinfo  {journal} {Numerische mathematik}\ }\textbf {\bibinfo
  {volume} {13}},\ \bibinfo {pages} {354} (\bibinfo {year} {1969})}\BibitemShut
  {NoStop}%
\bibitem [{\citenamefont {Le~Gall}(2014)}]{le2014powers}%
  \BibitemOpen
  \bibfield  {author} {\bibinfo {author} {\bibfnamefont {F.}~\bibnamefont
  {Le~Gall}},\ }\bibfield  {title} {\bibinfo {title} {Powers of tensors and
  fast matrix multiplication},\ }\href
  {https://doi.org/10.1145/2608628.2608664} {\bibfield  {journal} {\bibinfo
  {journal} {Proceedings of the 39th international symposium on symbolic and
  algebraic computation}\ ,\ \bibinfo {pages} {296}} (\bibinfo {year}
  {2014})}\BibitemShut {NoStop}%
\bibitem [{\citenamefont {Aho}(1974)}]{aho1974design}%
  \BibitemOpen
  \bibfield  {author} {\bibinfo {author} {\bibfnamefont {A.~V.}\ \bibnamefont
  {Aho}},\ }\href@noop {} {\emph {\bibinfo {title} {The Design and Analysis of
  Computer Algorithms}}}\ (\bibinfo  {publisher} {Reading/Addison-Wesley},\
  \bibinfo {address} {Reading},\ \bibinfo {year} {1974})\BibitemShut {NoStop}%
\bibitem [{\citenamefont {Ginibre}(1965)}]{ginibre1965statistical}%
  \BibitemOpen
  \bibfield  {author} {\bibinfo {author} {\bibfnamefont {J.}~\bibnamefont
  {Ginibre}},\ }\bibfield  {title} {\bibinfo {title} {Statistical ensembles of
  complex, quaternion, and real matrices},\ }\href
  {https://doi.org/10.1063/1.1704292} {\bibfield  {journal} {\bibinfo
  {journal} {Journal of Mathematical Physics}\ }\textbf {\bibinfo {volume}
  {6}},\ \bibinfo {pages} {440} (\bibinfo {year} {1965})}\BibitemShut {NoStop}%
\bibitem [{\citenamefont {Garc{\'\i}a-Mart{\'\i}n}\ \emph
  {et~al.}(2023)\citenamefont {Garc{\'\i}a-Mart{\'\i}n}, \citenamefont
  {Larocca},\ and\ \citenamefont {Cerezo}}]{garcia2023deep}%
  \BibitemOpen
  \bibfield  {author} {\bibinfo {author} {\bibfnamefont {D.}~\bibnamefont
  {Garc{\'\i}a-Mart{\'\i}n}}, \bibinfo {author} {\bibfnamefont
  {M.}~\bibnamefont {Larocca}},\ and\ \bibinfo {author} {\bibfnamefont
  {M.}~\bibnamefont {Cerezo}},\ }\bibfield  {title} {\bibinfo {title} {Deep
  quantum neural networks form gaussian processes},\ }\href
  {https://arxiv.org/abs/2305.09957} {\bibfield  {journal} {\bibinfo  {journal}
  {arXiv preprint arXiv:2305.09957}\ } (\bibinfo {year} {2023})}\BibitemShut
  {NoStop}%
\bibitem [{\citenamefont {Fulton}\ and\ \citenamefont
  {Harris}(1991)}]{fulton1991representation}%
  \BibitemOpen
  \bibfield  {author} {\bibinfo {author} {\bibfnamefont {W.}~\bibnamefont
  {Fulton}}\ and\ \bibinfo {author} {\bibfnamefont {J.}~\bibnamefont
  {Harris}},\ }\href@noop {} {\emph {\bibinfo {title} {Representation Theory: A
  First Course}}}\ (\bibinfo  {publisher} {Springer},\ \bibinfo {year}
  {1991})\BibitemShut {NoStop}%
\bibitem [{\citenamefont {Humphreys}(2012)}]{humphreys2012introduction}%
  \BibitemOpen
  \bibfield  {author} {\bibinfo {author} {\bibfnamefont {J.~E.}\ \bibnamefont
  {Humphreys}},\ }\href@noop {} {\emph {\bibinfo {title} {Introduction to Lie
  algebras and representation theory}}},\ Vol.~\bibinfo {volume} {9}\ (\bibinfo
   {publisher} {Springer Science \& Business Media},\ \bibinfo {year}
  {2012})\BibitemShut {NoStop}%
\bibitem [{\citenamefont {Racah}(2006)}]{racah2006group}%
  \BibitemOpen
  \bibfield  {author} {\bibinfo {author} {\bibfnamefont {G.}~\bibnamefont
  {Racah}},\ }\bibfield  {title} {\bibinfo {title} {Group theory and
  spectroscopy},\ }\href {https://doi.org/10.1007/BFb0045770} {\bibfield
  {journal} {\bibinfo  {journal} {Springer Tracts in Modern Physics, Volume
  37}\ ,\ \bibinfo {pages} {28}} (\bibinfo {year} {2006})}\BibitemShut
  {NoStop}%
\bibitem [{\citenamefont {{The Sage Developers}}(2021)}]{sage}%
  \BibitemOpen
  \bibfield  {author} {\bibinfo {author} {\bibnamefont {{The Sage
  Developers}}},\ }\href@noop {} {\bibinfo {title} {Sagemath, the sage
  mathematics software system}},\ \bibinfo {howpublished}
  {\url{https://www.sagemath.org}} (\bibinfo {year} {2021}),\ \bibinfo {note}
  {version 9.3, 2021}\BibitemShut {NoStop}%
\bibitem [{\citenamefont {Knapp}(2013)}]{knapp2013lie}%
  \BibitemOpen
  \bibfield  {author} {\bibinfo {author} {\bibfnamefont {A.~W.}\ \bibnamefont
  {Knapp}},\ }\href {https://doi.org/10.1007/978-1-4757-2453-0_1} {\emph
  {\bibinfo {title} {Lie Groups Beyond an Introduction}}},\ Vol.\ \bibinfo
  {volume} {140}\ (\bibinfo  {publisher} {Springer Science \& Business Media},\
  \bibinfo {year} {2013})\BibitemShut {NoStop}%
\bibitem [{\citenamefont {Nazarov}\ \emph {et~al.}(2024)\citenamefont
  {Nazarov}, \citenamefont {Postnova},\ and\ \citenamefont
  {Scrimshaw}}]{nazarov2024skew}%
  \BibitemOpen
  \bibfield  {author} {\bibinfo {author} {\bibfnamefont {A.}~\bibnamefont
  {Nazarov}}, \bibinfo {author} {\bibfnamefont {O.}~\bibnamefont {Postnova}},\
  and\ \bibinfo {author} {\bibfnamefont {T.}~\bibnamefont {Scrimshaw}},\
  }\bibfield  {title} {\bibinfo {title} {tatistical ensembles of complex,
  quaternion, and real matricess},\ }\href {https://doi.org/10.1112/jlms.12813}
  {\bibfield  {journal} {\bibinfo  {journal} {Journal of the London
  Mathematical Society}\ }\textbf {\bibinfo {volume} {109}},\ \bibinfo {pages}
  {e12813} (\bibinfo {year} {2024})}\BibitemShut {NoStop}%
\end{thebibliography}%

\renewcommand\appendixname{Supp. Info.}
\appendix
\onecolumngrid

\renewcommand\figurename{Supplemental Figure}

\setcounter{lemma}{0}
\setcounter{theorem}{0}
\setcounter{algorithm}{0}
\setcounter{figure}{0}

\newcounter{supfig}
\setcounter{supfig}{\value{figure}}

\section{Proof of Theorem 1}
\label{ap:theo-1}

In this section we provide the detailed proof of Theorem~\ref{th:fermion-haar}, which we here recall for convenience.
\begin{theorem} \label{th-ap:fermion-haar}
    Any fermionic linear optics circuit can be decomposed (up to a global minus sign) as 
    \begin{equation} \label{eq-ap:spin-matrix-decomposition-th}
            U(\thv)= L_{2n} L_{2n-1}\cdots L_1\,,
    \end{equation}
    where 
    \begin{align}\label{eq-ap:R-prod-th}
        &L_{2k-1} = \prod_{j=1}^{n-1} e^{i\alpha_{jk} X_j X_{j+1}}\,,\quad L_{2k}= \prod_{j=1}^{n} e^{i \beta_{jk} Z_j}\,,
    \end{align}
    for $1<k<n$. The parameters take values  $\alpha_{j,1}\in [0,2\pi)$, $\beta_{j,n}\in [0,2\pi)$ and $\alpha_{jk}, \in [0,\pi]$, $\beta_{j,k}\in [0,\pi]$.
    Moreover, the normalized Haar measure for the adjoint representation with respect to the parametrization in Eqs.~\eqref{eq-ap:spin-matrix-decomposition-th} and~\eqref{eq-ap:R-prod-th} is given by
     \begin{equation}\label{eq-ap:active_haar-th}
        d\mu(U)=\NC \prod_{k=1}^n \mu(L_{2k} L_{2k-1})\prod d\alpha\,d\beta\,,
    \end{equation}
    where $\NC=\left((2\pi)^{2n-1} \prod_{s=3}^{2n}\frac{\sqrt{\pi}^{\,s-2}}{\Gamma(s/2)}\right)^{-1}$,
    \begin{equation}\label{eq-ap:active_haar_layer}
       \mu(L_{2k} L_{2k-1})= \prod_{j=1}^{n-1}\sin(\alpha_{jk})^{f_n(2j,2k-1)}\prod_{j=1}^{n}\sin(\beta_{jk})^{f_n(2j-1,2k)}\,,
    \end{equation}
    and
    \begin{equation} \label{eq-ap:f_function}
        f_n(u,v)=
    \begin{cases}
        \min(2v-2, 4n-2u-1) \quad &{\rm if} \quad u>v, \\
        \min(4n-2v, 2u-1) \quad &{\rm if} \quad u<v \\
    \end{cases}\,.
    \end{equation}
\end{theorem}
The outline of the proof is the following. We first recall a decomposition of the elements in the standard representation of $\mathbb{SO}(d)$ in terms of Givens rotations, as reported in Ref.~\cite{diaconis2017hurwitz}, together with the associated Haar measure. Next, we explicitly write down an isomorphism between the Lie algebras of the groups $\mathbb{SO}(2n)$ and $\mathbb{SPIN}(2n)$, which will provide us the necessary intuition to map the previous decomposition from $\SO(2n)$ to fermionic linear optics (FLO) circuits. Lastly, making use of the circuit transformations described in Refs.~\cite{kokcu2022algebraic,camps2022algebraic}, we will prove how the Haar measure transforms under such rearrangements of the  gates, and obtain our optimal Haar random FLO circuits.

\subsection{Decomposition of $\mathbb{SO}(d)$ matrices}
\label{ap:decomposition}

We first show that any matrix $O\in\SO(d)$ can be decomposed as
\begin{equation}\label{eq-ap:so-matrix-decomposition}
    O=O_1 O_2\cdots O_{d-1}\,, 
\end{equation}
for
\begin{equation}\label{eq-ap:R-O-tilde-prod}
    O_j=R_j(\theta_{j,j+1})\cdots R_1(\theta_{1,j+1})\,,
\end{equation}
with Euler angles $\theta_{1,j+1}\in [0,2\pi)$, and $\theta_{j',j+1}\in [0,\pi]$ for $j'>1$. Here, the Givens rotations are defined by
\begin{equation}\label{eq-ap:givens}
    R_j(\theta)=e^{\theta L_{jj+1}}= \!\begin{pmatrix}
        \id_{j-1} \\ & \cos\theta & \sin\theta \\ & -\sin\theta & \cos\theta \\ & & & \id_{d-j-1}
    \end{pmatrix}\,,
\end{equation}
where $\{L_{jk}\}_{1\leq j<k\leq d}$ is an orthogonal basis of the special orthogonal algebra $\mathfrak{so}(d)$, with 
\begin{equation}\label{eq-ap:L_def}
    \left(L_{jk}\right)_{m,l} = \delta_{jm}\delta_{kl}-\delta_{jl}\delta_{km} \,.
\end{equation}
To do so, we need to prove that given an arbitrary $O\in\SO(d)$, we can always multiply it by the inverse of~\eqref{eq-ap:so-matrix-decomposition} and obtain the $d\times d$ identity matrix, $\id_d$. That is, that there always exist angles $\theta_{j,k}$ such that
\begin{equation}
    O O_{d-1}^TO_{d-2}^T\cdots O_1^T = \id_d\,,
\end{equation}
is satisfied. Let us start with the product $OO_{d-1}^T$, and let us first elucidate the action of the Givens rotation $R_1^T(\theta_{1,d})$ when multiplying $O$ from the left,
\begin{align}\nonumber
    O R_1^T(\theta_{1,d}) &= \begin{pmatrix}
        o_{11} & \cdots & o_{1 d} \\
        \vdots & \ddots & \vdots \\
        o_{d 1}  & \cdots & o_{d d} 
    \end{pmatrix} \begin{pmatrix}
         \cos\theta_{1,d} & -\sin\theta_{1,d} \\ \sin\theta_{1,d} & \;\;\;\cos\theta_{1,d} \\ & & & \id_{d-2}
    \end{pmatrix} 
    \\&= \begin{pmatrix}
        o_{11} \cos\theta_{1,d} + o_{12} \sin \theta_{1,d} & -o_{11} \sin \theta_{1,d} + o_{12}\cos\theta_{1,d} & o_{13} & \cdots & o_{1 d} \\ \vdots &
        \vdots & \vdots & \ddots & \vdots \\
        o_{d 1} \cos\theta_{1,d} + o_{d 2} \sin \theta_{1,d} & -o_{d 1} \sin \theta_{1,d} + o_{d 2}\cos\theta_{1,d} & o_{d 3} & \cdots & o_{d d} 
    \end{pmatrix}\,.
\end{align}
We can see that $R_1^T(\theta_{1,d})$ acts non-trivially on the first two columns of $O$, and as the identity on the rest. Analogously, $R_j^T(\theta_{j,k})$ only acts non-trivially on columns $j, j+1$. Now, we will use $R_1^T(\theta_{1,d})$ to set to zero the $(d,1)$ entry of the matrix $O R_1^T(\theta_{1,d})$. We do so by solving 
\begin{equation} \label{eq-ap:cos-sin}
     o_{d 1} \cos\theta_{1,d} + o_{d 2} \sin \theta_{1,d} = 0\,,
\end{equation}
whose solution is 
\begin{equation}\label{eq-ap:arctan}
    \theta_{1,d}= \arctan \left(-\frac{o_{d1}}{o_{d 2}}\right)\,,
\end{equation}
whenever $o_{d 2}\neq 0$ and $\theta_{1,d}=\pi/2$ otherwise. Importantly,  this equation can always be solved with a value $\theta_{1,d}\in[0,\pi]$. Next, we apply $R_2^T(\theta_{2,d})$, which is used to set to zero the $(d,2)$ entry of the matrix. Note that $R_2^T(\theta_{2,d})$ acts trivially on the first column, which guarantees that the $(d,1)$ entry remains equal to zero. We iterate this process through $R_{d-1}^T(\theta_{d-1,d})$, setting the first $d-1$ entries of the bottom row to zero,
\small
    \begin{equation} \label{eq-ap:R-action}
    \!\!\!\!\begin{pmatrix}
        o_{11} & \cdots & o_{1 d} \\
        \vdots & \ddots & \vdots \\
        o_{d 1}  & \cdots & o_{d d} 
    \end{pmatrix} \xrightarrow{R_1^T(\theta_{1,d})} \begin{pmatrix}
        o_{11}' & o_{12}' & \cdots & o_{1 d} \\
        \vdots & \vdots & \ddots & \vdots \\
       0  & o_{d 2}' &\cdots & o_{d d} 
    \end{pmatrix} \xrightarrow{R_2^T(\theta_{2,d})} \begin{pmatrix}
        o_{11}' & o_{12}'' & o_{13}' & \cdots & o_{1 d} \\
        \vdots & \vdots & \vdots & \ddots & \vdots \\
       0  & 0 & o_{d 3}' & \cdots & o_{d d} 
    \end{pmatrix} \rightarrow \cdots \rightarrow \begin{pmatrix}
        o_{11}' & o_{12}'' & \cdots & o_{1 d-1}' & 0 \\
        \vdots & \vdots & \ddots & \vdots &\vdots \\
        0  & 0 & \cdots & 0 & \pm 1
    \end{pmatrix}\,.
\end{equation}
\normalsize

Since the matrix $O_{d-1}$ is orthogonal, the last entry of the bottom row in Eq.~\eqref{eq-ap:R-action} must be equal to $\pm 1$. Furthermore, all entries in the last column but the last one must be zero, which also follows from the orthogonality of the matrices. We can then use the freedom to choose $\theta_{1,d}\in[0,2\pi)$ instead of $\theta_{1,d}\in[0,\pi]$ to set this entry to $+1$. Indeed, changing $\theta_{1,d}$ to $\theta_{1,d}+\pi$ still satisfies Eq.~\eqref{eq-ap:arctan}, but changes $o_{d2}'$ to $-o_{d2}'$; then, $\theta_{2,d}$ must change to $\pi-\theta_{2,d}$ to satisfy Eq.~\eqref{eq-ap:arctan}. In turn, the previous changes $o_{d3}'=-o_{d 2}' \sin \theta_{2,d} + o_{d 3}\cos \theta_{2,d}$ to $o_{d3}'=o_{d 2}' \sin \theta_{2,d} - o_{d 3}\cos \theta_{2,d}=-o_{d 3}'$. And so on. 

Successively applying $O_{d-2}^T$, $O_{d-3}^T,\dots$, we make zeros in the $d-2$, $d-3,\dots$ rows, and eventually obtain the identity matrix,
\begin{equation}
        \begin{pmatrix}
        o_{11} & \cdots & o_{1 d} \\
        \vdots & \ddots & \vdots \\
        o_{d 1}  & \cdots & o_{d d} 
    \end{pmatrix} \xrightarrow{O_{d-1}^T} \begin{pmatrix}
        o_{11}' & o_{12}'' & \cdots & o_{1 d-1}' & 0 \\
        \vdots & \vdots & \ddots & \vdots &\vdots \\
        0  & 0 & \cdots & 0 & \pm 1
    \end{pmatrix} \xrightarrow{O_{d-2}^T}\quad\cdots\quad \xrightarrow{O_1^T} \begin{pmatrix}
        1 & \cdots & 0 \\
        \vdots & \ddots & \vdots \\
        0  & \cdots & 1 
    \end{pmatrix}\,,
    \end{equation}
where in the last step, the $(1,1)$ entry of the matrix must be $1$ because the matrices $O_j$ belong to $\SO(d)$, and hence the matrix determinant must be equal to $1$.

We furthermore present the (normalized) right-and-left-invariant Haar measure on the standard representation of $\SO(d)$ under the previous decomposition in the following lemma.

\begin{lemma}\label{lem-ap:so-haar-measure}
    The normalized Haar measure $d\mu(O)$ on the standard representation of $\SO(d)$ with respect to the parametrization in Eqs.~\eqref{eq-ap:so-matrix-decomposition} and~\eqref{eq-ap:R-O-tilde-prod} is given by
    \begin{equation}
        d\mu(O) = \NC \prod_{1\leq j <k\leq d}\sin(\theta_{j,k})^{j-1} d\theta_{j,k}\,,
    \end{equation}
    where $\NC=\left((2\pi)^{d-1} \prod_{k=3}^{d}\frac{\sqrt{\pi}^{k-2}}{\Gamma(k/2)}\right)^{-1}$.
\end{lemma}

This result was first proven by Hurwitz in the context of invariant theory, and its  proof can be found in Ref.~\cite{diaconis2017hurwitz}. As we will see below, we will use Lemma~\ref{lem-ap:so-haar-measure}  to find quantum circuits implementing Haar random active FLO transformations.

\subsection{Isomorphism between the Lie algebras of $\mathbb{SO}(2n)$ and $\mathbb{SPIN}(2n)$}
\label{ap:lemma3}

We now provide an explicit isomorphism between the dynamical Lie algebra of FLO circuits and $\mathfrak{so}(2n)$.
\begin{lemma} \label{lem-ap:isomorphism}
    The isomorphism $\varphi$ connecting $\g$ in Eq.~\eqref{eq:dla} and $\mathfrak{so}(2n)$ is given by the linear map
   \begin{equation} \label{eq-ap:isomorphism}
       \varphi(L_{jk})= \frac{c_j c_k}{2}\,.
   \end{equation} 
\end{lemma}
\begin{proof}
    To show that $\varphi$ is indeed an isomorphism, we first need to show that it is an homomorphism, i.e., we need to prove that
    \begin{equation} \label{eq-ap:homomorphism}
        \varphi\left([L_{jk}, L_{j'k'}]\right)= [\varphi(L_{jk}), \varphi(L_{j'k'})]\,.
    \end{equation}
     We have
    \begin{equation} \begin{split}
        (L_{jk} L_{j'k'})_{mn} &= \sum_l (L_{jk})_{ml} (L_{j'k'})_{ln} = \sum_l (\delta_{jm}\delta_{kl}-\delta_{jl}\delta_{km}) (\delta_{j'l}\delta_{k'n}-\delta_{j'n}\delta_{k'l}) \\&= \delta_{kj'} \delta_{jm} \delta_{k'n} + \delta_{jk'}\delta_{km} \delta_{j'n}   -  \delta_{kk'}\delta_{jm} \delta_{j'n} - \delta_{jj'}\delta_{km} \delta_{k'n}\,,\end{split}
    \end{equation}
    from which it follows that
    \begin{equation}
        [L_{jk}, L_{j'k'}] = \delta_{kj'} L_{jk'} + \delta_{jk'} L_{kj'} -\delta_{kk'} L_{jj'} - \delta_{jj'} L_{kk'}\,.
    \end{equation}
    Hence,
    \begin{equation}\label{eq-ap:comm_Ls}\begin{split}
        \varphi\left([L_{jk}, L_{j'k'}]\right) = \delta_{k j'}\, \varphi(L_{jk'})+\delta_{jk'} \varphi(L_{kj'}) -\delta_{kk'} \varphi(L_{jj'}) - \delta_{jj'} \varphi(L_{kk'})\,.\end{split}
    \end{equation}
    On the other hand, using the anti-commutation relations of the Majorana operators, we find that
    \begin{equation} \begin{split}
        \left[\frac{c_j c_k}{2}, \frac{c_{j'} c_{k'}}{2}\right] = \delta_{k j'} \frac{ c_{j} c_{k'}}{2} + \delta_{jk'} \frac{ c_{k} c_{j'}}{2} -\delta_{k k'} \frac{ c_{j} c_{j'}}{2} - \delta_{j j'} \frac{ c_{k} c_{k'}}{2}  \,, \end{split}
    \end{equation}
     whenever $c_{j}c_{k}\neq\pm \,c_{j'}c_{k'}$ (and zero otherwise), and thus Eq.~\eqref{eq-ap:homomorphism} is satisfied.
     It remains to prove that $\varphi$ is bijective. For this, it suffices to notice that
     \begin{equation}
         \Tr\left[c_j c_k c_{j'} c_{k'}\right]= -2^n \,\delta_{jj'}\delta_{kk'} \,,
     \end{equation}
     for $j<k$ and $j'<k'.$
    That is, two ordered products of two Majorana operators are either orthogonal or equal. 
\end{proof}

\begin{figure*}[t]
    \centering
    \includegraphics[width=1.\linewidth]{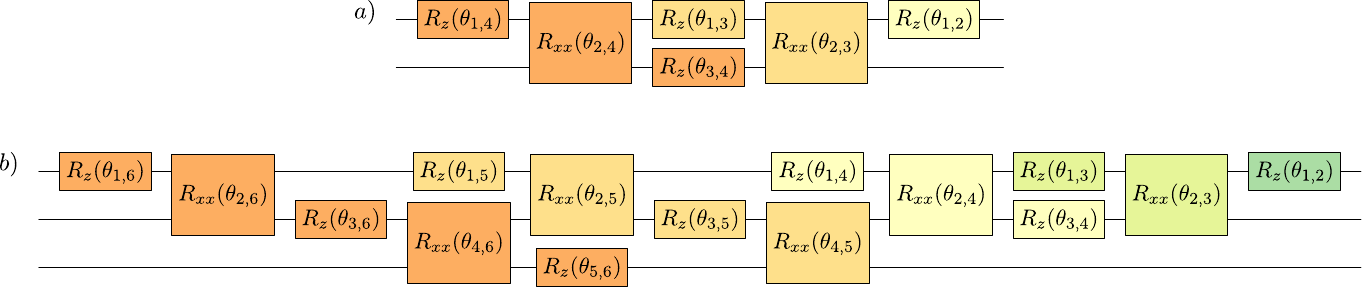}
    \caption{{\bf Suboptimal  Haar random active FLO circuits}. These quantum circuits produce Haar random active FLO on $n=2$ (a) and $n=3$ (b) qubits, respectively, when the parameters are sampled according to Eq.~\eqref{eq-ap:tfim_haar}. Gates are colored as per Eq.~\eqref{eq-ap:spin-matrix-decomposition}, explicitly showing the ladder-like structure of the circuits. }
    \refstepcounter{supfig}
    \label{fig-ap:so-circuits}
\end{figure*}

\subsection{From $\SO(2n)$ matrices to FLO circuits}

We here explain how to transform the decomposition of $\SO(2n)$ matrices in Eqs.~\eqref{eq-ap:so-matrix-decomposition} and~\eqref{eq-ap:R-O-tilde-prod} into FLO circuits, by simply using the isomorphism in Eq.~\eqref{eq-ap:isomorphism}. This results in Haar random active FLO circuits consisting of gates arranged in ``ladders'', as shown in Supp. Fig.~\ref{fig-ap:so-circuits}, and thus with a suboptimal depth. In particular, we prove the following useful lemma.
\begin{lemma} \label{lem-ap:fermion-haar}
Any fermionic linear optics circuit can be decomposed (up to a global minus sign) as 
\begin{equation} \label{eq-ap:spin-matrix-decomposition}
        U(\thv)= C_1 C_2\cdots C_{2n-1}\,,
\end{equation}
for 
\begin{equation}\label{eq-ap:R-prod}
    C_j= U_{R_j}(\theta_{j,j+1})U_{R_{j-1}}(\theta_{j-1,j+1})\cdots U_{R_1}(\theta_{1,j+1})\,,
\end{equation}
and
    \begin{equation} \label{eq-ap:U_Rj}
       U_{R_j}(\theta)=e^{\theta c_j c_{j+1}/2}\,,   
    \end{equation}
    with parameter values  $\theta_{1,j+1}\in [0,2\pi)$ and $\theta_{j',j+1}\in [0,\pi]$, for $j'>1$.
    
    Moreover, the normalized Haar measure for the adjoint representation with respect to the parametrization in Eqs.~\eqref{eq-ap:spin-matrix-decomposition} and~\eqref{eq-ap:R-prod} is given by
    \begin{equation}\label{eq-ap:tfim_haar}
        d\mu(U) = \NC\prod_{1\leq j <k\leq 2n}\sin(\theta_{j,k})^{j-1} d\theta_{j,k}\,,
    \end{equation}
    where $\NC=\left((2\pi)^{2n-1} \prod_{k=3}^{2n}\frac{\sqrt{\pi}^{k-2}}{\Gamma(k/2)}\right)^{-1}$.
\end{lemma}
\begin{proof}
We begin by recalling that under the adjoint action of the group of matchgate unitaries (i.e., the Lie group $e^{\mathfrak{g}}\cong\mathbb{SPIN}(2n)$, where $\mathfrak{g}$ is the matchgate dynamical Lie algebra defined in Eq.~\eqref{eq:dla}), the space of linear operators $\LC$ on the $n$-qubit Hilbert space decomposes onto irreducible representations (irreps) as~\cite{diaz2023showcasing}
\begin{equation} \label{eq-ap:irreps}
    \LC=\bigoplus_{\nu=0}^{2n} \LC_\nu\,,
\end{equation}
where $\LC_\nu={\rm span}_{\mathbb{C}}\{c_{j_1}c_{j_2}\cdots c_{j_\nu}\}_{1\leq j_1<j_2<\cdots< j_\nu\leq 2n}$ (with $\LC_0={\rm span}_{\mathbb{C}}\id$). That is, each irrep $\LC_\nu$ spans from complex valued linear combinations of products of $\nu$ strictly different Majorana operators. Here, the adjoint action of an element $U\in e^\g$ on an irrep $\LC_\nu$ is given by $\Phi_\nu^{\rm ad}\left(U\right)(\cdot)=U(\cdot)U^\dagger$. The adjoint representation of the matchgate group induces the adjoint representation of its Lie algebra $\g$.
\begin{definition}[Adjoint representation of $\g$]
    Given an element $G\in\g$, its adjoint representation on an irrep $\LC_\nu$ is obtained via the map $\phi_\nu^{\rm ad}:\g\rightarrow\mathbb{R}^{\dim(\LC_\nu)} \times \mathbb{R}^{\dim(\LC_\nu)}$, defined by
    \begin{equation}
        (\phi_\nu^{\rm ad}(G))_{lm} = \Tr\left[B_m^\nu \left[G,B_l^\nu\right]\right]\,,
    \end{equation}
    where the operators $\{B_l\}_l$ are a Hermitian basis for $\LC_\nu$.
\end{definition}

Moreover, it can be checked that the following definition holds~\cite{goh2023lie}.

\begin{definition}[Adjoint representation of $e^\g$] Given an element $U\in e^\g$, its adjoint representation on an irrep $\LC_\nu$ is a linear map $\Phi_\nu^{\rm ad}:e^\g\rightarrow \mathbb{GL}\left(\mathbb{R}^{\dim(\LC_\nu)}\right)\subset \mathbb{R}^{\dim(\LC_\nu)} \times \mathbb{R}^{\dim(\LC_\nu)}$, defined by
\begin{equation}
    \Phi_\nu^{\rm ad}\left(U=e^G\right) = e^{\phi_\nu^{\rm ad}(G)}\,,
\end{equation}    
where $\mathbb{GL}\left(\mathbb{R}^{\dim(\LC_\nu)}\right)$ is the general linear group of invertible matrices acting on $\mathbb{R}^{\dim(\LC_\nu)}$.
\end{definition}

Let us now use these definitions to compute the adjoint representation of the elements $U\in e^\g$ on $\LC_1$, the vector space spanned by the Majorana operators. We need to compute $\phi_\nu^{\rm ad}(\frac{c_j c_k}{2})$ for $\nu=1$. We choose as a Hermitian orthonormal basis for $\LC_1$ the set of orthonormal operators $\{c_j/\sqrt{2^n}\}_{j=1}^{2n}$. Hence, we find
\begin{align}
    \left(\phi_1^{\rm ad}\left(\frac{c_j c_k}{2}\right)\right)_{lm} = \Tr\left[\frac{c_m}{\sqrt{2^n}} \left[\frac{c_j c_k}{2},\frac{c_l}{\sqrt{2^n}} \right]\right] \nonumber =  \delta_{kl}\delta_{mj} - \delta_{jl}\delta_{km} = (L_{jk})_{lm}\,,
\end{align}
which implies
\begin{equation}
    e^{\theta c_j c_{j+1}/2} c_l\, e^{-\theta c_j c_{j+1}/2} = \sum_m \left(e^{\theta L_{jj+1}}\right)_{lm} c_m\,.
\end{equation}
We thus find that any fermionic linear optic circuit transformation (i.e., its adjoint action $\Phi_1^{\rm ad}$, up to a global minus sign) can be decomposed as in Eqs.~\eqref{eq-ap:spin-matrix-decomposition} and~\eqref{eq-ap:R-prod}, since any special orthogonal matrix $O$ can be written as in Eqs.~\eqref{eq-ap:so-matrix-decomposition} and~\eqref{eq-ap:R-O-tilde-prod}, and $\Phi_1^{\rm ad}$ and the standard representation of $\SO(2n)$ are isomorphic. Namely, we recover the well known result of Eq.~\eqref{eq:action-majorana}:
\begin{equation} \label{eq-ap:action-majorana}
    U c_{l\,} U^{\dagger} = \sum_
{m=1}^{2n} (O)_{lm}c_{m}\,.
\end{equation}

To prove that Eq.~\eqref{eq-ap:tfim_haar} gives the Haar measure on $\Phi_1^{\rm ad}$, we simply show that it is left-and-right invariant, as follows. We start by assigning 
    \begin{equation}
        d\mu(U) = d\mu(\Phi(U)) =d\mu(V)\,.
    \end{equation}
    Notice that this implies that since $d\mu(V)$ is not trivial, neither is $d\mu(U)$ (a trivial measure is always left-and-right invariant).
    Then, we have
    \begin{equation}\label{eq-ap:haar-invariance}\begin{split}
         d\mu(U'U) = d\mu\left(\Phi(U'U)\right) = d\mu\left(\Phi(U')\,\Phi(U)\right)= d\mu(V' V)= d\mu(V) = d\mu\left(\Phi(U)\right) = d\mu(U) \,,
    \end{split}
    \end{equation}
    where we used the invariance of the Haar measure on $\SO(2n)$ from Lemma~\ref{lem-ap:so-haar-measure}. Indeed, Eq.~\eqref{eq-ap:haar-invariance} indicates that since the representations are homomorphic, the invariance of the Haar measure on $\SO(2n)$ implies the invariance of the Haar measure on $\Phi_1^{\rm ad}$. The proof follows similarly for the right invariance.
    
    Finally, the measure is normalized since for $j>1$ we have
    \begin{equation}
        \int_0^\pi  \sin(\theta_{j,k})^{j-1} d\theta_{j,k} = \frac{\sqrt{\pi} \,\Gamma\left(\frac{j}{2}\right)}{\Gamma\left(\frac{j+1}{2}\right)}\,,
    \end{equation}
    and for $j=1$,
    \begin{equation}
        \int_0^{2\pi} d\theta_{1,k} = 2\pi\,.
    \end{equation}
Hence,
\begin{align}
    &\prod_{2\leq j <k\leq 2n}\int_0^{\pi} \sin(\theta_{j,k})^{j-1} d\theta_{j,k} =  \prod_{2\leq j <k\leq 2n} \frac{\sqrt{\pi} \,\Gamma\left(\frac{j}{2}\right)}{\Gamma\left(\frac{j+1}{2}\right)}=  \prod_{k=3}^{2n} \prod_{j=2}^{k-1} \frac{\sqrt{\pi} \,\Gamma\left(\frac{j}{2}\right)}{\Gamma\left(\frac{j+1}{2}\right)} = \prod_{k=3}^{2n} \frac{\sqrt{\pi}^{k-2}}{\Gamma(k/2)}\,,
\end{align}
and we find the normalization 
\begin{equation}
    \NC=\left((2\pi)^{2n-1} \prod_{k=3}^{2n}\frac{\sqrt{\pi}^{k-2}}{\Gamma(k/2)}\right)^{-1}\,.
\end{equation}
\end{proof}

\subsection{Optimal Haar random FLO circuits}
\label{ap:optimal_active}

\begin{figure}[t]
    \centering
    \includegraphics[width=1\linewidth]{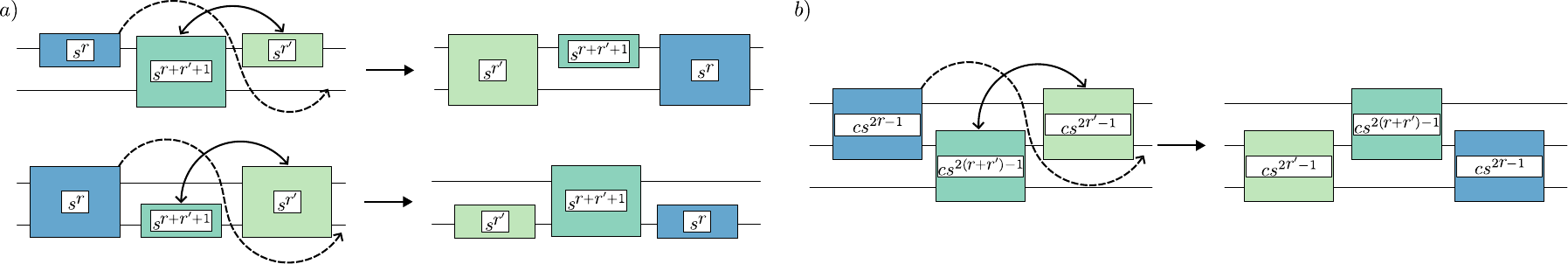}
    
    \caption{{\bf Transformation of the active and passive FLO Haar measure under ``turnovers''}. We graphically depict the transformation rules that we use to obtain optimal Haar random  active (a) and passive (b) FLO circuits. These rules specify how the local probability density functions providing the Haar measure change under certain ``turnovers'' of the gates, according to Lemmas~\ref{lem-ap:turnover-active} and~\ref{lem-ap:turnover-passive}, respectively. In the pictures, ``s'' stands for sine, while ``c'' is for cosine.}
    \refstepcounter{supfig}
    \label{fig-ap:turnover}
\end{figure}

We now make use of the ``turnover'' property of the gates $U_{R_j}$ and $U_{R_{j+1}}$ defined in Eq.~\eqref{eq-ap:U_Rj}, to bring the circuits in Lemma~\ref{lem-ap:fermion-haar} to optimal depth. This property was explained in Refs.~\cite{kokcu2022algebraic,camps2022algebraic}, and corresponds to the following transformation
\begin{equation}\label{eq-ap:turnover}
    U_{R_j}(\alpha)U_{R_{j+1}}(\beta)U_{R_j}(\gamma) = U_{R_{j+1}}(a)U_{R_j}(b)U_{R_{j+1}}(c) \,,
\end{equation}
under the appropriate relations between the set of angles $\{\alpha,\beta,\gamma\}$ and $\{a,b,c\}$, which can be shown to always exist~\cite{kokcu2022algebraic}. We here extend the turnover property to show how it affects the Haar measure in Eq.~\eqref{eq-ap:tfim_haar}. The result is stated in the following lemma and graphically depicted in Supp. Fig.~\ref{fig-ap:turnover}.

\begin{lemma} \label{lem-ap:turnover-active}
    Under the turnover property in Eq.~\eqref{eq-ap:turnover}, the unnormalized distribution
    \begin{equation}\label{eq-ap:turnover_f}
       \chi_{r,r'}(\alpha,\beta,\gamma) = \sin^r (\alpha) \sin^{r+r'+1}(\beta) \sin^{r'}(\gamma) \,,
    \end{equation}
    transforms into
    \begin{equation}\label{eq-ap:turnover_g}
        f_{r,r'}(a,b,c) = \sin^{r'} (a)\sin^{r+r'+1} (b) \sin^{r}(c) \,.
    \end{equation}
    Namely,  $\chi_{r,r'}(\alpha,\beta,\gamma)\,d\alpha\, d\beta\, d\gamma =f_{r,r'}(a, b, c)\, da\, db\, dc$. 
\end{lemma}

\begin{proof}
    We can derive relations between the angles $(\alpha,\beta,\gamma)$ and $(a,b,c)$ by inspecting the equalities stemming from Eq.~\eqref{eq-ap:turnover} when mapped to the standard representation of $\mathbb{SO}(3)$ (i.e., when replacing $U_{R_j}$ by $R_j$, see Eq.~\eqref{eq-ap:givens}). Particularly, the left-hand side reads
    \begin{equation} 
        \left(
        \begin{array}{ccc}
         \cos (\alpha ) \cos (\gamma )-\sin (\alpha ) \cos (\beta ) \sin (\gamma ) \quad & \sin (\alpha ) \cos (\beta ) \cos (\gamma )+\cos (\alpha ) \sin (\gamma )\quad & \sin (\alpha ) \sin (\beta ) \\
         -\sin (\alpha ) \cos (\gamma )-\cos (\alpha ) \cos (\beta ) \sin (\gamma ) \quad & \cos (\alpha ) \cos (\beta ) \cos (\gamma )-\sin (\alpha ) \sin (\gamma ) \quad & \cos (\alpha ) \sin (\beta ) \\
         \sin (\beta ) \sin (\gamma )\quad & -\sin (\beta ) \cos (\gamma ) \quad & \cos (\beta ) \\
        \end{array}
        \right)\,,
    \end{equation}
    while the right-hand side takes the form
    \begin{equation}
        \left(
        \begin{array}{ccc}
         \cos (b) \quad & \sin (b) \cos (c) \quad & \sin (b) \sin (c) \\
         -\cos (a) \sin (b) \quad & \cos (a) \cos (b) \cos (c)-\sin (a) \sin (c) \quad& \cos (a) \cos (b) \sin (c)+\sin (a) \cos (c) \\
         \sin (a) \sin (b) \quad & -\sin (a) \cos (b) \cos (c)-\cos (a) \sin (c) \quad & \cos (a) \cos (c)-\sin (a) \cos (b) \sin (c) \\
        \end{array}
        \right)\,.
    \end{equation}
    We will be interested in the relations obtained by equating the corner elements in the main anti-diagonal of these two matrices. Namely,
    \begin{align}\label{eq-ap:angle_relations_su2}
    \nonumber
        \sin(\alpha)\sin(\beta)&=\sin(b)\sin(c)\,, \\
        \sin(\beta)\sin(\gamma)&=\sin(a)\sin(b) \,.
    \end{align}
    Moreover, we make use of the fact that both the decompositions $U_{R_j}(\alpha)U_{R_{j+1}}(\beta)U_{R_j}(\gamma)$ and $U_{R_{j+1}}(a)U_{R_j}(b)U_{R_{j+1}}(c)$ in Eq.~\eqref{eq-ap:turnover} are parameterizations of the action of $\mathbb{SU}(2)$ over $\mathbb{C}^2$ \cite{kokcu2022algebraic}, corresponding to two different sets of Euler angles (one can readily verify that the sets $\{iZ_q,iX_qX_{q+1}\}$ and $\{iZ_{q+1},iX_qX_{q+1}\}$ both generate the $\mathfrak{su}(2)$ algebra under the Lie closure). Since we can write the unique Haar measure of $\mathbb{SU}(2)$ as $d\mu=\sin(\beta)\,d\alpha\, d\beta\, d\gamma$ or $d\mu=\sin(b)\,da\, db\, dc$, it follows that the Jacobian of the transformation is 
    \begin{equation}
        J = \left\vert \det\left[\frac{\partial(\alpha, \beta, \gamma)}{\partial(a,b,c)}\right]\right\vert = \frac{\sin (b)}{\sin (\beta)} \,.
    \end{equation} 
    Notice that the Jacobian is independent of the actual probability density function considered, i.e., of the powers $r, r'$ of the sine functions appearing in Eqs.~\eqref{eq-ap:turnover_f} and \eqref{eq-ap:turnover_g}.
    Lastly, we explicitly carry out the transformation between the two sets of angles using the Jacobian $J$ and the relations found in Eq.~\eqref{eq-ap:angle_relations_su2}:
    \begin{align}
        \nonumber
       \chi_{r,r'}(\alpha,\beta,\gamma) \,d\alpha\, d\beta\, d\gamma 
       \nonumber
       &= J\, \sin^r (\alpha) \sin^{r+r'+1}(\beta) \sin^{r'}(\gamma)\,da\, db\, dc
        \\& \nonumber = \frac{\sin (b)}{\sin (\beta)} \sin^r (\alpha) \sin^{r+r'+1}(\beta) \sin^{r'}(\gamma)\,da\, db\, dc \\
        \nonumber
        &= \sin(b) (\sin(\alpha)\sin(\beta))^r (\sin(\beta)\sin(\gamma))^{r'} da\, db\, dc \\& \nonumber= \sin(b) (\sin(b)\sin(c))^r (\sin(a)\sin(b))^{r'} da\, db\, dc \\
        \nonumber
        &= \sin^{r'} (a) \sin^{r+r'+1} (b) \sin^{r}(c) \, da\, db\, dc \nonumber \\&=  f_{r,r'}(a, b, c) \,da\, db\, dc\, ,
    \end{align}
    which completes the proof.
\end{proof}

We can interpret this lemma as enabling the passage of the rightmost gate $U_{R_j}$  through and eventually to the left of the subsequent gates $U_{R_{j+1}},U_{R_j}$ in Eq.~\eqref{eq-ap:turnover}, transforming it into a $U_{R_{j+1}}$ in the process, see Supp. Fig.~\ref{fig-ap:turnover}\footnote{Beware of the unfortunate standard convention whereby right in Eq.~\eqref{eq-ap:turnover} means left in the circuit diagram.}. During this transformation, the probability distribution function (PDF) for the angle associated with the gate passing through remains unchanged, that is, it retains the same power $p'$ in the sinusoidal distribution. The other two gates involved in the turnover preserve their gate type but exchange the powers of their respective sinusoidal PDFs. Here, we stress the important caveat that the aforementioned transformation of the Haar measure under a turnover of the gates is only valid when the original power of the sinusoidal function associated with the middle gate equals the sum of the powers of the two external gates plus one, i.e., $r+r'+1$.

By repeatedly applying the turnover property, one can show --as detailed in Refs.~\cite{kokcu2022algebraic,camps2022algebraic}-- that the triangle-shaped circuits in Lemma~\ref{lem-ap:fermion-haar}, obtained by mapping Hurwitz's decomposition of special orthogonal matrices into FLO, can be optimized into brick-wall-shaped circuits. This optimization yields the circuit architecture presented in Theorem~\ref{th-ap:fermion-haar}.
Yet, we need to show that the Haar measure in Eq.~\eqref{eq-ap:tfim_haar} transforms into the one reported in Eqs.~\eqref{eq-ap:active_haar-th},~\eqref{eq-ap:active_haar_layer} and~\eqref{eq-ap:f_function}. To illustrate the previous fact, we first observe that a ``ladder-like'' or triangular Haar random FLO circuit on $n$ qubits can be visualized as a stack of $2n-1$ ``anti-diagonals'' $\{d_k\}_{k=1}^{2n-1}$ of length $k$, arranged from left to right in the circuit, as depicted in Supp. Fig.~\ref{fig-ap:triangle_active_FLO}. Within each anti-diagonal $d_k$, the power of the sinusoidal PDF associated to each angle in Hurwitz's construction is given by $k-j$, where $j$ indicates the position of the gate within $d_k$ and takes values $j=1,\dots,k$ (see Supp. Fig.~\ref{fig-ap:triangle_active_FLO}).
The transformation from triangular to brick-wall circuits involves applying turnover operations to all anti-diagonals $d_{k}$  with odd $k$ (excluding the last one, $d_{2n-1}$), ``reflecting'' them with respect to $d_{2n-1}$. Throughout this process of successive turnovers, each gate moving through the circuit alternates between $U_{R_{\rm odd}}$ (Pauli-$Z$ rotations) and $U_{R_{\rm even}}$ (Pauli-$XX$ rotations) until it reaches its final position in the circuit, eventually yielding the optimal architecture described in Eq.~\eqref{eq-ap:spin-matrix-decomposition-th}.

Regarding the final Haar measure, we note that turning over the gate in the $j$-th position of anti-diagonal $d_k$ has the effect of swapping the PDFs of the $j$-th and $(j+1)$-th gates of all subsequent anti-diagonals $d_{i}$ with $i>k$ (again, see Supp. Fig.~\ref{fig-ap:triangle_active_FLO}). Denoting by $\vec{d_k}$ the vector containing the powers of the sinusoidal PDFs in $d_k$, i.e., whose entries  are $(\vec{d_k})_j=k-j$, and by $\sigma_{a,b}$ the swap operator exchanging positions $(a,b)$ of a vector, the combined effect of turning over the entire anti-diagonal $d_{k}$ through the subsequent diagonals is represented as
\begin{equation}
\Sigma_{k}=\sigma_{1,2}\,\sigma_{2,3}\cdots\sigma_{k,k+1}\,.    
\end{equation}
Crucially, during the turnover of an entire anti-diagonal $d_k$, we proceed starting from the last gate --the one associated with the last entry of the vector $\vec{d_k}$-- to the first, yielding precisely the action $\Sigma$ introduced above. This ordering, combined with the initial assignment of the entries of all the $\vec{d}_k$ vectors, ensures that the condition required for transforming the Haar measure --specified by Eqs.~\eqref{eq-ap:turnover_f} and~\eqref{eq-ap:turnover_g}-- is satisfied throughout the entire compression process.

\begin{figure}[t]
    \centering
    \includegraphics[width=\linewidth]{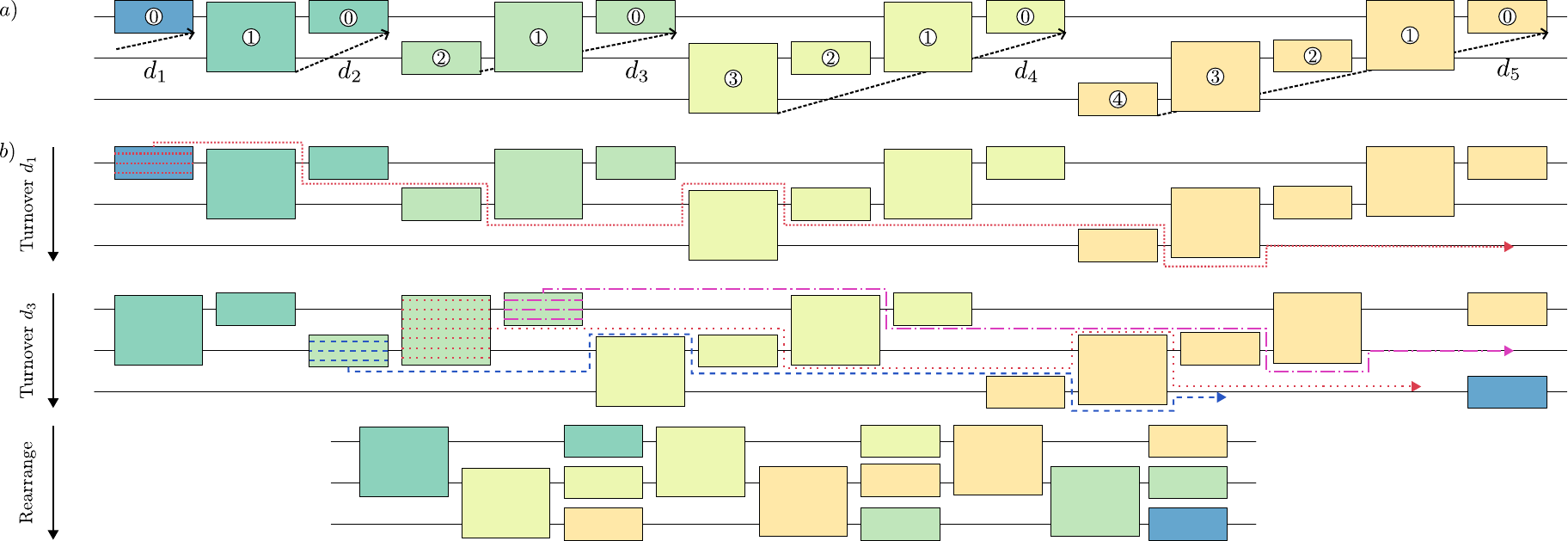}
    \caption{{\bf Compressing the circuits to optimal depth} (a) We show how the powers of the sine functions corresponding to the independent parameters' unnormalized PDFs are arranged in the ``ladder-like'' or triangular Haar random active FLO circuits in Lemma~\ref{lem-ap:fermion-haar}. We use different colors to highlight the diagonals $d_k$ that are used to keep track of the arrangement of said exponents while optimizing the circuit's architecture.
    (b) We illustrate how to transform a triangular-shaped circuit into an optimally compressed, brick-wall-structured one. We consider the case of an $n=3$ active FLO circuit, consisting of diagonals $d_1,\dots, d_5$.}
    \refstepcounter{supfig}
    \label{fig-ap:triangle_active_FLO}
\end{figure}

After applying all the turnovers, and renaming $k$ as $2k'\equiv k$ or $2k'+1\equiv k$, for $k'=1,2,\dots, n-1$ the resulting anti-diagonals become:
\begin{align}
    \nonumber
    \vec{\Bar{d}_{2k'}} &= \Sigma_{2k'-1}\Sigma_{2k'-3}\dots\Sigma_1 \vec{d_{2k'}}
    =(0,2,\dots,2k'-2,2k'-1,2k'-3,\dots,3,1)\,,\\[8pt]
    \vec{\Bar{d}_{2k'+1}} &= \Sigma_{2k'-1}\Sigma_{2k-3}\dots\Sigma_1 \vec{d_{2k'+1}}
    =(1,3,\dots,2k'-1,2k',2k'-2,\dots,2,0)\,,
\end{align}
alongside the unchanged anti-diagonal $\vec{\Bar{d}_1}=\vec{d_1}=(0)$. 
By construction, after all the turnovers have been carried out the ensuing circuit structure will consist of the stack $\Bar{d}_2,\Bar{d}_4,\dots,\Bar{d}_{2n-2},\Bar{d}_{2n-1},\Bar{d}_{2n-3},\dots,\Bar{d}_1$.
Finally, the function appearing in the active FLO Haar measure in Eq.~\eqref{eq-ap:active_haar_layer} is
\begin{equation}
    f_n(u,v)=
    \begin{cases}
        \min(2v-2,\, 4n-2u-1)\,, &{\rm if}\quad u>v, \\[6pt]
        \min(4n-2v,\, 2u-1)\,, &{\rm if}\quad u<v\,,
    \end{cases}
\end{equation}
which precisely corresponds to reading out the entries of the vectors $\vec{\Bar{d}_k}$ according to the sequential ordering of the gates in the optimized circuit.

\section{Proof of Theorem 2}
\label{ap:theo-2}
In this section we provide the proof of Theorem~\ref{th:passive-fermion-haar}, which we here recall for convenience.
\begin{theorem} \label{th-ap:passive-fermion-haar-th}
    Any passive fermionic linear optics circuit can be decomposed (up to a global minus sign) as 
    \begin{equation} \label{eq-ap:passive-flo-matrix-decomposition-th}
            U(\thv)= e^{i\frac{\lambda_n}{2}\sum_{q=1}^n Z_q} \tilde{L}_n \tilde{L}_{n-1} \cdots \tilde{L}_1\,,
    \end{equation}
    for 
    \begin{align}\label{eq-ap:R-tilde-prod-th}
        \nonumber
        \tilde{L}_{2k} &= \prod_{j=1}^{\lfloor \frac{n}{2}\rfloor} U_{\tilde{R}_{2j-1}}(\theta_{j,2k},\phi_{j,2k},\lambda_{2j-1}\delta_{2k, n})\,,\\
        \tilde{L}_{2k+1} &= \prod_{j=1}^{\lfloor \frac{n-1}{2}\rfloor} U_{\tilde{R}_{2j}}(\theta_{j,2k+1},\phi_{j,2k+1},\lambda_{2j-1}\delta_{2k+1, n}+\lambda_{2j}\delta_{2k+1,1})\,,
    \end{align}
    with  $k\in 0,1,\dots,\left\lfloor\frac{n}{2}\right\rfloor$, $\tilde{L}_0\equiv\id$, and
    \begin{equation} \label{eq-ap:U-R-tilde-th}
       U_{\tilde{R}_j}(\theta,\phi,\lambda)= e^{i\frac{\phi+\lambda}{4} Z_j} e^{i\frac{-\phi-\lambda}{4} Z_{j+1}} e^{i\frac{\theta}{2} (X_jY_{j+1}-Y_j X_{j+1})} e^{i\frac{-\phi+\lambda}{4} Z_j} e^{i\frac{\phi-\lambda}{4} Z_{j+1}} \,\!,
    \end{equation}
    with parameter values $\theta\in[0,\frac{\pi}{2}]$, $\phi\in[0,2\pi)$, and $\lambda\in[0,2\pi)$. 
    In addition, the normalized Haar measure for the adjoint representation with respect to the parametrization in Eqs.~\eqref{eq-ap:passive-flo-matrix-decomposition-th} and~\eqref{eq-ap:R-tilde-prod-th} is given by
    \begin{equation} \label{eq-ap:passive-FLO-haar-th}
        d\mu(U)= \NC \prod_{k=0}^{\lfloor \frac{n}{2}\rfloor} \mu(\tilde{L}_{2k}) \mu(\tilde{L}_{2k+1})\prod d\theta \,d\phi\, d\lambda\,, 
    \end{equation}
    where  $\NC=\left(\sqrt{2\pi}^{n(n+1)}\prod_{s=1}^{n-1}(2s)^{s-n}\right)^{-1}$,
    \begin{align}\label{eq-ap:passive-FLO-haar-layer-th}
        \mu(\tilde{L}_{2k}) &= \prod_{j} \cos( \theta_{j,2k})\sin(\theta_{j,2k})^{g_n(2j-1,2k)}   \nonumber\\ 
        \mu(\tilde{L}_{2k+1}) &= \prod_{j} \cos( \theta_{j,2k+1})\sin(\theta_{j,2k+1})^{g_n(2j,2k+1)} \,,
    \end{align}
    with $\mu(\tilde{L}_{0})\equiv1$, and
    \begin{equation} \label{eq-ap:g_function}
        g_n(u,v)=
    \begin{cases}
        \min(4v-3, 4n-4u-1) \quad &{\rm if} \quad u>v, \\
        \min(4n-4v+1, 4u-1) \quad &{\rm if} \quad u<v \\
    \end{cases}\,.
    \end{equation}
\end{theorem}
The outline of the proof is the same as for Theorem~\ref{th-ap:fermion-haar}. We first recall a decomposition of the elements in the standard representation of $\mathbb{U}(d)$, also reported in Ref.~\cite{diaconis2017hurwitz}, and the associated Haar measure. Next, we explicitly write down an isomorphism between the algebras of the groups $\mathbb{U}(n)$ and $\SO(2n)\cap\SPBB(2n,\mathbb{R})$, which we will use -together with the previous map between $\SO(2n)$ and $\mathbb{SPIN}(2n)$- to translate the $\mathbb{U}(n)$ decomposition to passive fermionic linear optics circuits. Lastly, using  circuit reshaping techniques analogous to those used for the active FLO case, we will prove how to obtain our optimal Haar random passive FLO circuits.

\subsection{Decomposition of $\U(d)$ matrices}
\label{ap:U-matrices}

The derivation for the unitary group $\U(d)$ is completely analogous to that for $\SO(d)$. We recall that the standard representation of the unitary group $\mathbb{U}(d)$ consists of the unitary matrices of size $d\times d$, acting irreducibly on $\mathbb{C}^{d}$. These matrices satisfy that $UU^\dagger=U^\dagger U=\id_d$ for every $U\in\U(d)$. Here, the statement is that any $U\in\U(d)$ can be factorized as
\begin{equation} \label{eq-ap:u-matrix-decomposition}
    U = e^{i\lambda_1}\, U_1U_2\cdots U_{d}\,,
\end{equation}
with
\begin{equation}    \label{eq-ap:R-U-tilde-prod}
   U_j = \tilde{R}_j(\theta_{j,j+1},\phi_{j,j+1},0) \cdots \tilde{R}_1(\theta_{1,j+1},\phi_{1,j+1},\lambda_{j+1})\,,
\end{equation}
and rotation matrices
\begin{align}
     \tilde{R}_j(\theta, \phi, \lambda) = e^{\frac{\phi+\lambda}{2} P_j} e^{\frac{-\phi-\lambda}{2} P_{j+1}} e^{\theta L_{jj+1}} e^{\frac{-\phi+\lambda}{2} P_j} e^{\frac{\phi-\lambda}{2} P_{j+1}} =\begin{pmatrix}
        \id_{j-1} \\ & e^{i\lambda} \cos\theta & e^{i\phi}\sin\theta \\ & -e^{-i\phi}\sin\theta & e^{-i\lambda}\cos\theta \\ & & & \id_{d-j-1}
    \end{pmatrix}\,, \label{eq-ap:R-tilde}
\end{align}
where we recall from Eq.~\eqref{eq-ap:L_def} the definition of the $\mathfrak{so}(d)$ basis elements, $\left(L_{jk}\right)_{m,l} = \delta_{jm}\delta_{kl}-\delta_{jl}\delta_{km}$, while the $\{P_j\}_{j=1}^d$ are a basis for the (real) vector space of purely-imaginary diagonal matrices, i.e.,
\begin{equation} \label{eq-ap:P_def}
    \left(P_{j}\right)_{m,l} = i\delta_{jm}\delta_{jl}\,.
\end{equation}
The angles can take values $\theta_{j',j+1}\in[0,\frac{\pi}{2}]$, $\phi_{j',j+1}\in[0,2\pi)$ (with $j'\leq j$), and $\lambda_j\in[0,2\pi)$. 
The main difference with the special orthogonal group is that since the entries of a unitary matrix $U$ can be complex, instead of Eq.~\eqref{eq-ap:cos-sin}, we now have to solve
\begin{equation}\label{eq-ap:u-cos-sin}
    u_{d 1}  \cos\theta_{1,d} + u_{d 2} e^{-i\phi}\sin \theta_{1,d} = 0\,,
\end{equation}
where we used the definition of the matrices $\tilde{R}_j$ in Eq.~\eqref{eq-ap:R-tilde}, and we set $\lambda=0$ as per Eq.~\eqref{eq-ap:R-U-tilde-prod}.
That is, we use the parameter $\phi$ to transform $-u_{d 1}/u_{d2}$ into a real number, and then solve the equation as in the orthogonal case. Since  $\phi$ can take values in $[0,2\pi)$, we can adjust the sign of $-e^{i\lambda}u_{d 1}/u_{d 2}$ at will, and hence we only need $\theta\in[0,\frac{\pi}{2}]$ to solve Eq.~\eqref{eq-ap:u-cos-sin}. Moreover, in the $\tilde{R}_1$ rotations the parameter $\lambda$ is not necessarily zero but instead is used to ensure that the final diagonal entry after making zeros in a row is $+1$ instead of a complex phase. Finally, in the last step, after multiplying by $U_1^\dagger$ from the left the $(1,1)$ matrix entry is not fixed as before but instead can also be an arbitrary complex phase. This is the reason why we need to add the phase $e^{i\lambda_1}$ in Eq.~\eqref{eq-ap:u-matrix-decomposition}.

\subsection{The isomorphism between the Lie algebras of $\U(n)$ and $\SO(2n)\cap\SPBB(2n,\mathbb{R})$}
\label{ap:un-lie-algebra}

We now provide an explicit isomorphism between the Lie algebras of the groups $\U(n)$ and $\SO(2n)\cap\SPBB(2n,\mathbb{R})$. Analogously to the case of active FLO, we will use this isomorphism (together with the one between the algebras of $\SO(2n)$ and $\mathbb{SPIN}(2n)$ presented above) to obtain a set of generators for passive FLO circuits and translate between parametrizations of $\U(n)$ and particle-preserving matchgate circuits.
The aforementioned isomorphism is realized for a matrix $O\in\SO(2n) \cap \mathbb{SP}(2n,\mathbb{R})$ and a matrix $U\in\U(n)$ by the map~\cite{mele2024efficient}
\begin{equation} \label{eq:U-SO-SP}
   O= {\rm Re}(U) \otimes \id + {\rm Im}(U) \otimes iY\,,
    \end{equation}
where $Y=\left(\begin{smallmatrix} 0 & -i \\ i & 0 \end{smallmatrix}\right)$ is the usual single-qubit Pauli matrix. At the Lie algebra level, the isomorphism reads
\begin{equation} \label{eq:u-so-sp}
    o = {\rm Re}(u) \otimes \id + {\rm Im}(u) \otimes iY\,,
\end{equation}
where $o\in \so(2n)\cap \mathfrak{sp}(2n,\mathbb{R})$ and $u\in\mathfrak{u}(n)$. 
Using Eq.~\eqref{eq:u-so-sp}, we can map a set of generators of $\mathfrak{u}(n)$ to a set of generators of $\so(2n)\cap \mathfrak{sp}(2n,\mathbb{R})$, as
\begin{equation}
    L_{2j-1,2j} = -iP_j\otimes iY\,, \qquad  L_{2j, 2j+2} + L_{2j-1, 2j+1} = L_{jj+1} \otimes \id\,.
\end{equation}
Indeed, since the $\mathfrak{u}(n)$ algebra can be entirely generated from the matrices $P_j$ and $L_{jj+1}$ via nested commutation,\footnote{Interestingly, the $P_j$ matrices are a basis for a Cartan subalgebra of $\mathfrak{u}(n)$.}  we thus find that the Lie algebra  $\so(2n)\cap \mathfrak{sp}(2n,\mathbb{R})$ can be generated from $ \{L_{2j-1, 2j}\}_{j\in[n]} \cup \{L_{2j, 2j+2} + L_{2j-1, 2j+1}\}_{j\in[n-1]}$.
We then transform these generators into Majorana operators using Lemma~\ref{lem-ap:isomorphism}, obtaining a set of generators $\{c_{2j-1}c_{2j}\}_{j\in[n]} \cup \{c_{2j}c_{2j+2} + c_{2j-1} c_{2j+1}\}_{j\in[n-1]}$ (notice that we omitted the factor $1/2$ as for now we are just concerned with finding a set of generators for passive FLO). Finally, in terms of Pauli matrices these generators are represented by
\begin{equation}
    \GC =  \{Z_q\}_{q=1}^n \cup \{X_qY_{q+1}-Y_q X_{q+1}\}_{q=1}^{n-1} \,.
\end{equation}

\subsection{From $\U(n)$ matrices to passive FLO circuits}

Using the previous isomorphisms between $\U(n)$, $\SO(2n)\cap\SPBB(2n,\mathbb{R})$, and the adjoint action of passive FLO, one can directly translate Hurwitz's construction in the SI~\ref{ap:U-matrices} to obtain passive Haar random FLO circuits with a triangular shape, as shown in Fig.~\ref{fig-ap:u-circuits}. In particular, the following lemma holds.

\begin{lemma} \label{lem-ap:passive-fermion-haar}
Any passive fermionic linear optics circuit can be decomposed (up to a global minus sign) as 
\begin{equation} \label{eq-ap:passive-flo-matrix-decomposition}
        U(\thv)= e^{i\frac{\lambda_1}{2}\sum_{q=1}^n Z_q} \tilde{C}_1 \tilde{C}_2\cdots \tilde{C}_{n-1}\,,
\end{equation}
for 
\begin{equation}\label{eq-ap:R-tilde-prod}
    \tilde{C}_j=  U_{\tilde{R}_j}(\theta_{j,j+1},\phi_{j,j+1},0) \cdots U_{\tilde{R}_1}(\theta_{1,j+1},\phi_{1,j+1},\lambda_{j+1})\,,
\end{equation} 
and
    \begin{align} \label{eq-ap:U-R-tilde}
       U_{\tilde{R}_j}(\theta,\phi,\lambda)= e^{i\frac{\phi+\lambda}{4} Z_j} e^{i\frac{-\phi-\lambda}{4} Z_{j+1}} e^{i\frac{\theta}{2} (X_jY_{j+1}-Y_j X_{j+1})} e^{i\frac{-\phi+\lambda}{4} Z_j} e^{i\frac{\phi-\lambda}{4} Z_{j+1}} \,\!,
    \end{align}
    with parameter values   $\theta\in[0,\frac{\pi}{2}]$, $\phi\in[0,2\pi)$, and $\lambda\in[0,2\pi)$. 

    In addition, the normalized Haar measure for the adjoint representation with respect to the parametrization in Eqs.~\eqref{eq:passive-flo-matrix-decomposition} and~\eqref{eq:R-tilde-prod} is given by
    \begin{equation}\label{eq-ap:passive-FLO-haar}
        d\mu(U) = \NC   \prod_{1\leq j <k\leq n} \cos( \theta_{j,k})\sin(\theta_{j,k})^{2j-1} d\theta_{j,k} d\phi_{j,k}\prod_{j=1}^n d\lambda_j\,,
    \end{equation}
    where $\NC=\left(\sqrt{2\pi}^{n(n+1)}\prod_{k=1}^{n-1}(2k)^{k-n}\right)^{-1}$.
\end{lemma}

\begin{proof}
    Completely analogous to that of Lemma~\ref{lem-ap:fermion-haar}.
\end{proof}

\begin{figure*}[t]
    \centering
    \includegraphics[width=\linewidth]{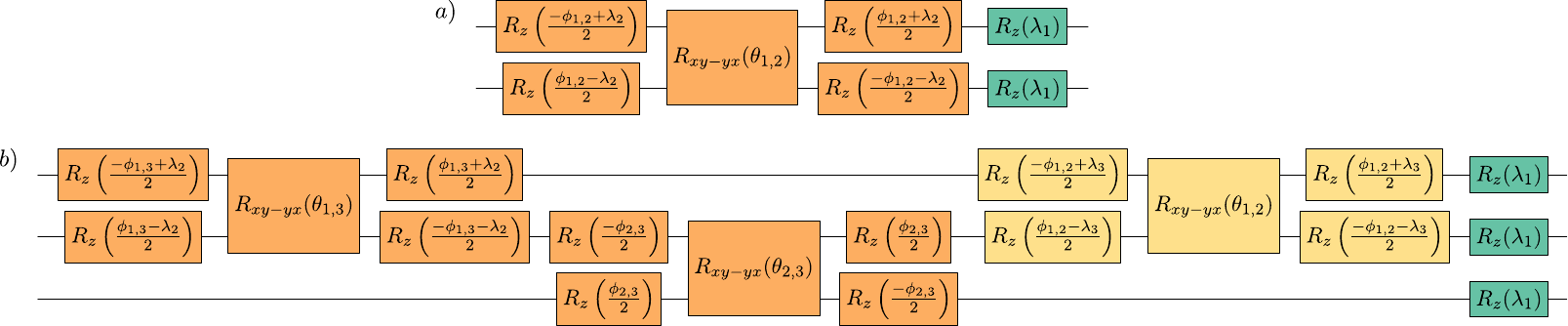}
    \caption{{\bf Suboptimal Haar random passive FLO circuits}. These quantum circuits produce Haar random passive FLO  on $n=2$ (a) and $n=3$ (b)  qubits, respectively, when the parameters are sampled according to Eq.~\eqref{eq-ap:passive-FLO-haar}. Gates are colored as per Eq.~\eqref{eq-ap:passive-flo-matrix-decomposition}, again showing the ladder-like structure of the circuits. }
    \label{fig-ap:u-circuits}
\end{figure*}

\subsection{Optimal Haar random passive FLO circuits}

We finally proceed to determine how the Haar measure on passive FLO transforms under the turnover property of the gates $U_{\tilde{R}_j}$ and $U_{\tilde{R}_{j+1}}$. Just as for the case of active FLO circuits, this property can be expressed by the following equality
\begin{equation}\label{eq-ap:turnover_passive}
    U_{\tilde{R}_j}(\theta_1,\phi_1,\lambda_1)U_{\tilde{R}_{j+1}}(\theta_2,\phi_2,\lambda_2)U_{\tilde{R}_j}(\theta_3,\phi_3,\lambda_3) = U_{\tilde{R}_{j+1}}(t_1,p_1,l_1)U_{\tilde{R}_j}(t_2,p_2,l_2)U_{\tilde{R}_{j+1}}(t_3,p_3,l_3) \,,
\end{equation}
under an appropriate, always existing, mapping of the original parameters $\{\theta_i,\phi_i,\lambda_i\}$ to the new ones $\{t_i,p_i,l_i\}$. The existence and validity of Eq.~\eqref{eq-ap:turnover_passive} is established within the proof of the following lemma.
\begin{lemma}\label{lem-ap:turnover-passive}
    Under the turnover property in Eq.~\eqref{eq-ap:turnover_passive}, the unnormalized distribution
    \begin{equation}\label{eq-ap:turnover_passive_f}
       \chi_{r,r'}(\theta_1,\theta_2,\theta_3) = \cos(\theta_1)\sin^{2r-1}(\theta_1) \cos(\theta_2)\sin^{2r+2r'-1}(\theta_2) \cos(\theta_3)\sin^{2r'-1}(\theta_3) \,,
    \end{equation}
    transforms into
    \begin{equation}\label{eq-ap:turnover_passive_g}
        f_{r,r'}(t_1,t_2,t_3) = \cos(t_1)\sin^{2r'-1}(t_1) \cos(t_2)\sin^{2r+2r'-1} (t_2) \cos(t_3)\sin^{2r}(t_3) \,.
    \end{equation}
    Namely,  $\chi_{r,r'}(\theta_1,\theta_2,\theta_3)\,d\theta_1\, d\theta_2\, d\theta_3 =f_{r,r'}(t_1, t_2, t_3)\, dt_1\, dt_2\, dt_3$. 
\end{lemma}
\begin{proof}
    We begin by noticing that given a $\U(3)$ matrix  $\mathcal{A}_1 \equiv e^{i\lambda_3} \tilde{R}_1(\theta_1,\phi_1,\lambda_1)\tilde{R}_{2}(\theta_2,\phi_2,0)\tilde{R}_1(\theta_3,\phi_3,\lambda_2)$, one can always re-cast it into the alternative parameterization $\mathcal{A}_2 \equiv e^{i\lambda_3}\tilde{R}_2(\theta_1,-\phi_1,-\lambda_1)\tilde{R}_{1}(\theta_2,-\phi_2,0)\tilde{R}_2(\theta_3,-\phi_3,-\lambda_2)$ by means of a similarity transformation, i.e., $\mathcal{A}_2 = M\mathcal{A}_1M^{-1}$, where 
    \begin{equation}
        M = \begin{pmatrix}
            0 & 0 & 1 \\
            0 & -1 & 0 \\
            1 & 0 & 0 \\
        \end{pmatrix}
    \end{equation}
    is itself an element of $\U(3)$. Indeed, one can readily prove that the following relations hold: $M\tilde{R}_1(\theta_1,\phi_1,\lambda_1)M^{-1}=\tilde{R}_2(\theta_1,-\phi_1,-\lambda_1)$, and $M\tilde{R}_2(\theta_2,\phi_2,\lambda_2)M^{-1}=\tilde{R}_1(\theta_2,-\phi_2,-\lambda_2)$. Using the latter, $\mathcal{A}_2 = M\mathcal{A}_1M^{-1}$ follows by direct inspection. To avoid confusion, we further rename the parameters in $\mathcal{A}_2$  from Greek to Roman letters, using $t_i,p_i,l_i$ instead of $\theta_i,\phi_i,\lambda_i$.

    Now we observe that the parameterization $\mathcal{A}_1$ corresponds to Hurwitz's decomposition of the $\U(3)$ matrices into Givens rotations explained in the SI~\ref{ap:U-matrices}, and hence it has an associated Haar measure with PDFs for the $\theta$ angles given by $\chi_{1,1}(\theta_1,\theta_2,\theta_3)$, and uniform for the rest. 
    Since the unique Haar measure is both left and right invariant, applying the similarity transform induced by $M$ must leave it unchanged. Hence, we find the Jacobian of the transformation to simply be given by the ratio of the two density functions in the corresponding parametrizations $\mathcal{A}_1$ and $\AC_2$, 
    \begin{equation} \label{eq-ap:jacobian-u}
        J = \frac{f_{1,1}}{\chi_{1,1}}=\frac{\cos(t_1)\sin(t_1)\cos(t_2)\sin^3(t_2)\cos(t_3)\sin(t_3)}{\cos(\theta_1)\sin(\theta_1)\cos(\theta_2)\sin^3(\theta_2)\cos(\theta_3)\sin(\theta_3)} \,.
    \end{equation}
    Here we used the fact that the transformations $p_i=-\phi_i$ and $l_i=-\lambda_1$ have a trivial effect on the Jacobian, since the latter is defined as an absolute value. Furthermore, sampling phases uniformly in $[0,2\pi)$ is clearly equivalent to sampling them in $[0,-2\pi)$.

    Furthermore, we can compare the matrix entries of $\AC_1$ and $\AC_2$, similarly to what we did in the active FLO case. The ensuing matrices, being now complex, are too cumbersome to be fully reported here, but one can check by comparing the module of the corner entries in the main anti-diagonal of these two matrices that  relations analogous to those in Eq.~\eqref{eq-ap:angle_relations_su2} hold in this case too, i.e.,
    \begin{align}\label{eq-ap:angle_relations_u3}
    \nonumber
        \sin(\theta_1)\sin(\theta_2)&=\sin(t_2)\sin(t_3)\,, \\
        \sin(\theta_2)\sin(\theta_3)&=\sin(t_1)\sin(t_2) \,.
    \end{align}

    Lastly, using the fact that the Jacobian is independent of the actual probability density function considered, i.e., of the powers $r, r'$ of the sine functions appearing in Eqs.~\eqref{eq-ap:turnover_passive_f} and \eqref{eq-ap:turnover_passive_g}, we can explicitly carry out the transformation between the two sets of angles using and Eqs.~\eqref{eq-ap:jacobian-u} and~\eqref{eq-ap:angle_relations_u3}:
    \footnotesize
    \begin{align}
        \nonumber
       \chi(\theta_1,\theta_2,\theta_3) \,d\theta_1\, d\theta_2\, d\theta_3 
       \nonumber
       &= J\, \cos(\theta_1)\sin^{2r-1}(\theta_1) \cos(\theta_2)\sin^{2r+2r'-1}(\theta_2) \cos(\theta_3)\sin^{2r'-1}(\theta_3)\,dt_1\, dt_2\, dt_3
        \\& \nonumber = \frac{\cos(t_1)\sin(t_1)\cos(t_2)\sin^3(t_2)\cos(t_3)\sin(t_3)}{\cos(\theta_1)\sin(\theta_1)\cos(\theta_2)\sin^3(\theta_2)\cos(\theta_3)\sin(\theta_3)} \cos(\theta_1)\sin^{2r-1}(\theta_1) \cos(\theta_2)\sin^{2r+2r'-1}(\theta_2) \cos(\theta_3)\sin^{2r'-1}(\theta_3)\,dt_1\, dt_2\, dt_3 \\
        \nonumber
        &= \cos(t_1)\sin(t_1)\cos(t_2)\sin^3(t_2)\cos(t_3)\sin(t_1)\left(\sin(\theta_1)\sin(\theta_2)\right)^{2r-2}\left(\sin(\theta_3)\sin(\theta_3)\right)^{2r'-2} dt_1\, dt_2\, dt_3 \\
        \nonumber
        &= \cos(t_1)\sin^{2r'-1} (t_1) \cos(t_2)\sin^{2r+2r'-1}(t_2) \cos(t_3)\sin^{2r-1}(t_3) \,dt_1\, dt_2\, dt_3 \nonumber \\
        &=  f(t_1, t_2, t_3) \,dt_1\, dt_2\, dt_3\, ,
    \end{align}
    \normalsize
    which completes the proof.
\end{proof}

With the turnover property established, we can proceed to optimize the Haar random passive FLO circuit architecture from the triangular shape in Lemma~\ref{lem-ap:passive-fermion-haar}, to a maximally compressed one.
This optimization is completely identical to the one carried out in the SI~\ref{ap:optimal_active} for active FLO, hence we will refrain from performing it explicitly again here.
Let us only notice that now the powers of the sinusoidal functions appearing in the parameters' PDFs are given by the odd integers $2r-1\geq1$, as opposed to the active FLO case where they were given by the integers $r\geq 0$. Since this is the only change to the definition of the $\vec{d_k}$ vectors (together with the fact that now there are $n$ such vectors instead of $2n-1)$, and since Lemma~\ref{lem-ap:turnover-passive} is completely analogous to Lemma~\ref{lem-ap:turnover-active}, one can follow the proof of Theorem~\ref{th-ap:fermion-haar} step by step to recover the statement in Theorem~\ref{th-ap:passive-fermion-haar-th}.

\section{Compilation costs}
\label{ap:compilation}

Let us discuss the cost of sampling a $2n\times 2n$ matrix $O$ from the Haar measure on the standard representation of $\SO(2n)$, and then compiling it into a matchgate circuit. We will assume that the matrix $O$ is sampled following the common method described in~\cite{mezzadri2006generate}. Namely, the first step is to independently draw the entries of the matrix from a Gaussian distribution with zero mean and unit variance. In other words, the matrix is drawn from the real Ginibre ensemble. Next, the columns are orthonormalized using the Gram-Schmidt algorithm. Finally, the determinant of the matrix is computed, and the first two columns are exchanged if the determinant is equal to $-1$.

The classical computational cost of the Gram-Schmidt orthogonalization is $\OC(n^3)$. This can be easily seen from the fact that there are $2n$ vectors to be orthonormalized, and that a single orthogonalization requires computing the inner product ($\OC(n)$ multiplications and additions) between a column vector and up to $2n-1$ other vectors. Hence the $\OC(n^3)$ complexity in terms of the number of floating-point operations required. 
Alternatively, instead of using the Gram-Schmidt algorithm, a QR decomposition can be performed with a fixed gauge to orthogonalize the matrix~\cite{mezzadri2006generate}.
 The QR decomposition may be computed using Gram-Schmidt itself. However, this algorithm is numerically unstable, and a method based on Householder reflections is usually preferred~\cite{mezzadri2006generate}. Householder reflections are reflections of a vector about a plane or a hyperplane, and $2n-1$ of them are required to compute the QR decomposition.  Each Householder reflection takes $\OC(n^2)$ arithmetic floating-point operations, and thus the computational cost of performing the QR decomposition in this preferred way is also $\OC(n^3)$. We will mention here that there exists a lesser-known algorithm that can compute the QR decomposition in $\OC(n^{2+\frac{1}{4-\alpha}})$ operations~\cite{knight1995fast}, using a matrix multiplication subroutine of complexity $\Theta(n^\alpha)$. Thus, using the Strassen algorithm for fast matrix multiplication~\cite{strassen1969gaussian}, $\alpha=2.807$ and the total asymptotic time complexity of computing the QR decomposition is $\OC(n^{2.838\dots})$ (while there exist other galactic algorithms with smaller $\alpha$, they are not practical due to the very large constant factors~\cite{le2014powers}).
 
 The last step to obtain a Haar random special orthogonal matrix requires computing the determinant of the matrix, which incurs the same cost than matrix multiplication (see Theorem 6.6 in~\cite{aho1974design}). Overall, we find that sampling a Haar random matrix from the standard representation of $\SO(2n)$ can be done classically with $\OC(n^3)$ floating-point operations.

 We remark that the algorithm for sampling a Haar random matrix $U$ from the standard representation of $\U(n)$ is very similar to the orthogonal case. The only differences being that one initially samples a $n\times n$ matrix from the complex Ginibre ensemble (i.e.,  the matrix entries follow a standard complex Gaussian distribution~\cite{ginibre1965statistical}), and that the determinant need not be computed at the end. Hence the classical computational complexity is again $\OC(n^3)$.  

 Next, we discuss the computational cost of compiling a matrix from the standard representations of $\SO(2n)$ and $\U(n)$ into a matchgate circuit. That is, given a matrix $O\in\SO(2n)$ or $U\in\U(n)$, we want to find the rotation angles in a matrix decomposition such as that in Eq.~\eqref{eq-ap:so-matrix-decomposition} or Eq.~\eqref{eq-ap:u-matrix-decomposition}. Once these angles have been found, the rotations appearing in the chosen matrix decomposition can be readily translated into quantum gates.
 For simplicity, we will consider the aforementioned decompositions, but we stress that other decompositions of special orthogonal or unitary matrices~\cite{oszmaniec2022fermion} lead to essentially the same compilation costs. 
 
 In order to find the rotation angles for a given matrix, we have to follow the steps explained in the SI~\ref{ap:decomposition}. This involves $\Theta(n^2)$ multiplications by $R_j$ or $\tilde{R}_j$ matrices. Since each of these multiplications acts non-trivially on only two columns (each column being of length $\Theta(n)$), we arrive at an $\OC(n^3)$ total compilation cost.

\section{Sampling from peaked distributions}
\label{ap:sampling}

The proposed circuit architectures require sampling parameters from the distributions given in Eq.~\eqref{eq-ap:tfim_haar}, for active FLO, and in Eq.~\eqref{eq-ap:passive-FLO-haar} for the case of passive FLO (or equivalently, from Eqs.~\eqref{eq-ap:active_haar_layer} and~\eqref{eq-ap:passive-FLO-haar-layer-th}).
As shown in Fig.~\ref{fig:pdfs} and explained in the main text, these distributions tend to peak around and towards the value $\theta_{j,k}=\pi/2$, as the number of qubits increases. 

To efficiently sample from these peaked distributions, a standard approach is to employ the rejection sampling technique. In rejection sampling, a candidate value is drawn from a simpler proposal distribution, usually uniform or Gaussian, and this value is accepted with a probability proportional to the target probability density function. Specifically, if the target density is $f(\theta)$ and the proposal density is $g(\theta)$ such that $f(\theta) \leq M g(\theta)$ for all $\theta$, with $M$ being a constant, a candidate $\theta$ sampled from $g(\theta)$ is accepted with a probability $f(\theta) / (M g(\theta))$. This ensures that the accepted samples follow the desired distribution $f(\theta)$.
Assuming that the proposal density is uniform $g(\theta)=c$ (with $c$ being the maximum of the --not necessarily normalized-- target distribution $f(\theta)$), one can uniformly draw a value of $\theta$ from its domain, then draw a value $y$ uniformly from $[0,c]$, and if $y<f(\theta)$, $\theta$ is accepted as a valid sample, and rejected otherwise.
The acceptance rate $R$, i.e., the probability of a sample being accepted, is equivalent to the ratio of the areas subtended by the unnormalized distributions.

As an example, let us consider the case of active FLO circuits. To achieve Haar measure sampling, the angles $\theta$ are drawn in the range $[0,\pi]$ from the, unnormalized, distributions $f(\theta)=\sin(\theta)^{j-1}$. The acceptance rate turns out to be $R=\frac{\Gamma(\frac{j}{2})}{\sqrt{\pi}\,\Gamma(\frac{j+1}{2})}$. This quantity goes to zero as $j$ goes to infinity, thus introducing a heavier and heavier sampling complexity as the active FLO circuit size $n$ increases.
To control the number of sampling attempts, one can reduce the sampled region by restricting the domain of $\theta$ to the range where $f(\theta)$ is significantly non-zero. For the case being considered, where the distribution $f(\theta)$ peaks around $\pi/2$, sampling can be restricted to a narrower interval by introducing a function $\epsilon(j)$, such that $\theta \in [\epsilon(j), \pi - \epsilon(j)]$. Particularly, define a threshold value $\eta$ such that we discard the values of $\theta$ for which $f(\theta)<\eta$. The acceptance rate now converges to $R\underset{j\rightarrow\infty}{\rightarrow}\frac{1}{2}\sqrt{\frac{-\pi}{\log(\eta)}}$, which is finite and equals to $R\simeq 0.292$ for a reasonable threshold value $\eta=10^{-4}$. This removes the $j$, and hence the $n$, sampling complexity scaling.
To check how much the reduced probability density function deviates from the true one, we can compute the total variation distance $d(f,g)=\int_I |f(\theta)-g(\theta)| d\theta$. For the case at hand, the total variation distance is given by twice the integral between $0$ and $\epsilon(j)$ of $f(\theta)$, or in other words, twice the cumulative of $f$ evaluated at $\epsilon(j)$. Numerical investigation reveals that for the same value of $\eta$ considered before, the total variation distance approaches a limit value of order $10^{-5}$ as $j$ goes to infinity.

The case of passive FLO circuits is simpler. Studying the independent angles unnormalized distributions $f(\theta)=\cos(\theta)\sin(\theta)^{2j-1}$ appearing in Eq.~\eqref{eq:passive-FLO-haar} one finds that their cumulatives simply turn out to be $F(\theta) = \sin(\theta)^{2j}$. This can be inverted, yielding $F^{-1}(y)=\arcsin(y^{\frac{1}{2j}})$. Having access to the inverse cumulative function enables the use of the inverse transform sampling, which allows to sample from the distribution $f$ by uniformly sampling a value $u\in [0,1]$ and generating a sample $\theta = F^{-1}(u)$. Given the exact nature of the inverses $F^{-1}$ for any value of $j$, the parameters of Haar random passive FLO circuits can be sampled at a fixed complexity regardless of the size of the circuit.

\section{Clifford FLO are not a FLO $4$-design}
\label{ap:4-design}

Let us prove Theorem~\ref{th:clifford}.

\begin{theorem}
    The Clifford FLO group does not form a FLO $4$-design.
\end{theorem}

\begin{proof}
We start by defining $\PC_n=\{\id,X,Y,Z\}^{\otimes n}$, the set of Pauli operators on $n$ qubits, and recalling that the operator
\begin{equation}
    Q=\sum_{P\in\PC_n} P^{\otimes 4}
\end{equation}
belongs to the $4$-fold commutant of the $n$-qubit Clifford group ${\rm Cl}_n$. That is, 
\begin{equation}\label{eq:com-clif-4}
    [W^{\otimes 4},Q]=0\,,\quad \forall W\in {\rm Cl}_n\,.
\end{equation}
In particular, any Clifford FLO will also satisfy Eq.~\eqref{eq:com-clif-4}. Then, we can show that Clifford FLO do not form a $4$-design over FLO by proving that $Q$ does not belong to the $4$-fold commutant of the spinor representation of $\SO(2n)$ (or adjoint representation of $\mathbb{SPIN}(2n)$). As such, we need to evaluate the quantity
\begin{align}
    &\int_{\mathbb{SPIN}(2n)} d\mu(U) U^{\otimes 4} Q(U\ad)^{\otimes 4}\nonumber\\
    &=\sum_{P\in\PC_n}\int_{\mathbb{SPIN}(2n)} d\mu(U) (UPU\ad)^{\otimes 4} \,.\label{eq:twirl-Q}
\end{align}

Here, we also find it convenient to recall that under the adjoint action of the Lie group $e^{\mathfrak{g}}\cong\mathbb{SPIN}(2n)$, where $\mathfrak{g}$ is the matchgate dynamical Lie algebra defined in Eq.~\eqref{eq:dla}, the space of operators $\LC$ on the $n$-qubit Hilbert space decomposes onto irreducible representations as in Eq.~\eqref{eq-ap:irreps}. Then, since matchgate unitaries preserve the number of Majoranas as per Eq.~\eqref{eq-ap:action-majorana}, and since every Pauli is expressed (up to a phase) as a product of Majoranas, we have that if $P\in\LC_\nu$, then $UPU\ad\in\LC_\nu$ for any matchgate unitary $U\in e^{\mathfrak{g}}$. 

From the previous, we can therefore analyze Eq.~\eqref{eq:twirl-Q} by changing the summation  $\sum_{P\in\PC_n}=\sum_{\nu =0}^{2n}\sum_{P\in\LC_\nu}$, and evaluating the integral in each irrep individually. Since $\LC_0$ only contains the identity operator, and therefore $U$ acts trivially thereon,  we will focus on the irrep $\LC_1={\rm span}_{\mathbb{C}}\{c_l\}_{l=1}^{2n}$. That is,
\begin{align}
&\sum_{l=1}^{2n}\int_{\mathbb{SPIN}(2n)} d\mu(U) (Uc_l U\ad)^{\otimes 4} \nonumber\\
&=\sum_{l=1}^{2n} \int_{\SO(2n)} d\mu(O) \left(\sum_
{m=1}^{2n}(O)_{lm}c_{m}\right)^{\otimes 4}\nonumber\\
&= \sum_{l=1}^{2n} \int_{\SO(2n)} d\mu(O) \left(O^T\cdot \vec{e}_{l}\right)^{\otimes 4}*\vec{c}^{\otimes 4}\nonumber\\
&= 2n\int_{\SO(2n)} d\mu(O) \left(O\cdot \vec{e}_{l}\right)^{\otimes 4}*\vec{c}^{\otimes 4}\,.
\end{align}
In the first equality we used Eq.~\eqref{eq-ap:action-majorana}, which leads to an  integral over the standard representation of $\SO(2n)$. Then, in the second equality  we defined the unit vectors $\vec{e}_l$ as the canonical basis of $\mathbb{R}^{2n}$ (i.e., vectors with a one in the $l$-th position and zeros  elsewhere); and the vector of Majorana operators $\vec{c}\equiv(c_1,\ldots,c_{2n})$. Here, we also leveraged the Hadamard product  $(a_1,a_2)*(b_1,b_2)=a_1b_1+a_2b_2$, and the fact that it commutes with the tensor product. In the last equality we have used that $\forall l,l'$ there exists a matrix in $\SO(2n)$ such that $O\vec{e}_{l'}=\vec{e}_{l}$, and hence from the right invariance of the Haar measure $\int_{\SO(2n)} d\mu(O) (O\cdot \vec{e}_{l'})^{\otimes 4}=\int_{\SO(2n)} d\mu(O) (O\cdot \vec{e}_{l})^{\otimes 4}$ for all  $l,l'$. Here, we have furthermore used the fact that the Haar measure on the special orthogonal group $\SO(2n)$ is invariant under transposition, i.e., $d\mu(O)=d\mu(O^T)$. 

Due to linearity of the Hadamard product, we can then only consider the integral. Noting that the vectorization map~\cite{mele2023introduction} takes the operator $O^{\otimes 2} |l\rangle\langle l|^{\otimes 2} (O^\dagger)^{\otimes 2}$ onto $(O\cdot \vec{e}_{l})^{\otimes 4}$ (recall that $O$ is real-valued, and thus $O^*=O$), we can focus on the integral 
\begin{align}
    &2n\int_{\SO(2n)} d\mu(O) (O\cdot \vec{e}_{l})^{\otimes 4}\rightarrow2n\int_{\SO(2n)} d\mu(O) O^{\otimes 2} |l\rangle\langle l|^{\otimes 2} (O^\dagger)^{\otimes 2}\,.
\end{align}
Using the Weingarten calculus (see 
Supplemental Information D of~\cite{garcia2023deep} for a detailed derivation), we find 
\begin{align}
    &2n\int_{\SO(2n)} d\mu(O) O^{\otimes 2} |l\rangle\langle l|^{\otimes 2} (O^\dagger)^{\otimes 2}=\frac{1}{(2n+2)}(\id\otimes \id+SWAP+\Pi)\,,
\end{align}
where $\id\otimes \id=\sum_{i,j=1}^{2n}|ij\rangle\langle ij|$, $SWAP=\sum_{i,j=1}^{2n}|ij\rangle\langle ji|$, and $\Pi=\sum_{i,j=1}^{2n}|ii\rangle\langle jj|$. 
Clearly, $|l\rangle\langle l|^{\otimes 2} \neq \frac{1}{(2n+2)}(\id\otimes \id+SWAP+\Pi)$ which means that averaging the component of $Q$ within $\LC_1$ over FLO changes it, returns an operator which is not itself, nor zero. Since the additional terms which appear cannot be canceled by the average of $Q$ on the other irreps (as these contain a different number of Majoranas), we find that 
\begin{equation}
    \int_{\mathbb{SPIN}(2n)} d\mu(U) U^{\otimes 4} Q(U\ad)^{\otimes 4}\neq Q,0\,,
\end{equation}
and therefore, Clifford FLO are not an FLO $4$-design. 

\end{proof}

\section{Random Clifford FLO sampling}
\label{ap:clifford}

In this section we prove the correctness of Algorithm~\ref{alg:clifford}, which we here report for convenience.

\begin{algorithm}
\caption{Random Clifford FLO sampling}\label{alg-ap:clifford}
\begin{algorithmic}[1]
\State \textbf{Input}: \texttt{Initial angles $\theta_{j,k}=0$ for all $j,k$}
 \For{$k$ in $2n,\dots,2$}
  \State \texttt{Uniformly sample a random integer $l$ in $[k]$}
  \For{$j\geq l$}
  \State \texttt{ $\theta_{j,k}\leftarrow\pi/2$}
  \EndFor
  \State \texttt{Set $\theta_{l,k}\leftarrow\theta_{l,k}$ or $\theta_{l,k}\leftarrow\theta_{l,k}+\pi$ at random with equal probability}
  \If{l=k}:
  \State \texttt{Set $\theta_{l-1,k}\leftarrow\theta_{l-1,k}$ or $\theta_{l-1,k}\leftarrow\theta_{l-1,k}+\pi$ at random with equal probability}
  \EndIf
  \EndFor
 \State \Return \texttt{angles}
\end{algorithmic}
\end{algorithm}

To understand why the group of Clifford FLO unitaries is uniformly sampled in this way, let us remind that their adjoint action is isomorphic to the group of signed permutations matrices of $2n$ items with unit determinant, $B_{2n}\cap \SO(2n)$. 
Hence, randomly sampling a Clifford FLO circuit must be equivalent to sampling from $B_{2n}\cap \SO(2n)$. 
Now, suppose a signed permutation matrix $O$ has been sampled. Its rows and columns have but one non-zero entry with value $\pm 1$. If we want to compile it into a quantum circuit, we can repeat the analysis in the SI~\ref{ap:decomposition} to find the angles $\theta_{j,k}$ in its decomposition. Then, each Givens rotation $R_j(\theta_{j,k})$ directly translates into a quantum gate as in Theorem~\ref{th:fermion-haar}. 
Let us also recall that the action of the Givens rotations $R_j^T(\theta_{j,k})$ on $O$ is rotating $O$'s $j$-th and $(j+1)$-th columns such that the new elements of $OR_j^T(\theta_{j,k})$ read
\begin{equation}
    \begin{aligned}
    o_{r,j}' &= o_{r,j}\cos \theta_{j,k} + o_{r,j+1}\sin \theta_{j,k}\\
    o_{r,j+1}' &= -o_{r,j}\sin \theta_{j,k} + o_{r,j+1}\cos \theta_{j,k}
    \end{aligned}\,,
\end{equation}
where $r$ indexes the row.
Since for signed permutations any matrix element $o_{r,c}\in \{0,1,-1\}$, we can see that for $j>1$ there are only three possible values for $\theta_{j,k}\in [0,\pi]$ that can be used to invert $O$, namely $\theta_{j,k}\in\{0,\pi/2,\pi\}$, while when $j=1$ we have also access to $3\pi/2$. While choosing $\theta_{j,k}=0$ clearly leaves the two columns unchanged, the choice $\theta_{j,k}=\pi$ flips their signs. Lastly, setting $\theta_{j,k}=\pi/2$ switches the $(j+1)$-th column with minus the $j$-th one, while $\theta_{j,k}=3\pi/2$ exchanges minus the $(j+1)$-th column with the $j$-th one.
Hence, remembering from  SI~\ref{ap:decomposition} that each layer $O_{k-1}$ is used to set the $k$-th row of $O$ to that of the identity matrix, 
it follows that to implement a Clifford FLO we need each layer to end with a sequence of $\pi/2$ rotations to transport the non-zero entry to its final destination. This sequence has to start from the column where said entry originally was, let us call it $l$ as in Algorithm~\ref{alg-ap:clifford}  (we will see that the case $\theta_{l,k}=3\pi/2$ will be used to account for the signs).  In a uniformly random permutation matrix, the non-zero entry of any row can be found on any column with equal probability. By first uniformly choosing a random integer $l\in[2n]$ as in Algorithm~\ref{alg-ap:clifford}, we guarantee that this is the case for the last row. Then, after moving it to the end of the row by means of the $O_{2n-1}$ layer, the non-zero element of the second-to-last row in the permutation matrix can again be found on any column but the last with equal probability. Thus, uniformly sampling a new random integer $l\in[2n-1]$ ensures that we obtain such distribution. 
 One then repeats this process until $O$ has been brought to identity form (up to signs).

As far as the signs are concerned, in a random signed permutation with unit determinant the non-zero entries are independently sampled with equal probability from $\{+1,-1\}$, except for one entry which fixes the determinant to one. The procedure described above moves those entries to the main diagonal, flipping their sign for each shift by one position. In doing so, it does not alter their distribution, yielding a final diagonal identity-like matrix where the signs on the diagonal are independent and equally distributed in $\{+1,-1\}$. To account for them, we tweak each layer $L_{k-1}$ by changing the angle of the Givens rotation at position $l$ within the layer to $\theta_{l,k}\gets \theta_{l,k}+\pi$. This flips
the sign of the non-zero entry in the $k$-th row before it gets transported to the diagonal, ensuring that we can correctly obtain the identity. Since each sign occurs with equal probability, randomly choosing the angle $\theta_{l,k}$ between $\pi/2$ and $3\pi/2$, reproduces the desired statistics.  
This ends the proof for Algorithm~\ref{alg-ap:clifford}.

To show that Algorithm~\ref{alg-ap:clifford} is indeed optimal, it suffices to notice that the number of inversions (i.e., nearest-neighbor transpositions) needed on average to decompose a random permutation of $2n$ items is 
\begin{equation}\label{eq-ap:optimal_inversions}
 \frac{n(2n-1)}{2}\,.
\end{equation}
Half of the inversions correspond to one-qubit gates and the other half correspond to two-qubit gates. Hence, the optimal number of two-qubit gates when $n$ becomes large is $\sim n^2/2$, which is precisely what Algorithm~\ref{alg-ap:clifford} delivers. 
Indeed, recalling that the layer $C_{k-1}$ (for $k=2,\dots,2n$) contains $k-1$ gates, and that we pick an integer $l$ uniformly at random from $\{1,\dots,k\}$, setting the first $l-1$ gates to the identity, we find that the expected number of non-trivial gates in $C_{k-1}$ is
\begin{equation}
\mathbb{E}_l[k - l] = \frac{k-1}{2}\,.
\end{equation}
Since the layers are sampled independently, summing over all $k$ yields
\begin{equation}
    \sum_{k=2}^{2n}\mathbb{E}_l[k - l] = \sum_{k=2}^{2n}\frac{k-1}{2} = \frac{n(2n-1)}{2}\,,
\end{equation}
in exact agreement with Eq.~\eqref{eq-ap:optimal_inversions}.

 \medskip

We now provide an algorithm to uniformly sample the group of passive Clifford FLO, that we used to compute the corresponding frame potentials shown in the main text.

 \begin{algorithm}
\caption{Random passive Clifford FLO sampling}\label{alg-ap:clifford-passive}
\begin{algorithmic}[1]
\State \textbf{Input}: \texttt{Initial angles $\theta_{j,k}=0$ for all $j,k$}
 \For{$k$ in $2n,2n-2,\dots,4$}
  \State \texttt{Uniformly sample a random integer $l$ in $[k]$}
  \If{$l\,\%\,2=0$}:
  \For{$j> l$}
  \State \texttt{ $\theta_{j,k}\,,\, \theta_{j-1,k-1}\leftarrow\pi/2$}
  \EndFor
  \State \texttt{Set $\theta_{l,k}\,,\,\theta_{l-1,k-1}\leftarrow\frac{\pi}{2}$ or $\theta_{l,k}\,,\,\theta_{l-1,k-1}\leftarrow\frac{3\pi}{2}$ at random with equal probability}
  \Else:
  \For{$j> l$}
  \State \texttt{ $\theta_{j,k}\,,\,\theta_{j,k-1}\leftarrow\pi/2$}
  \EndFor
  \State \texttt{Set $\theta_{l,k}\leftarrow\frac{\pi}{2}$, or $\theta_{l,k}\leftarrow\frac{3\pi}{2}$, at random with equal probability}
  \EndIf
  \EndFor
  \State \texttt{Set $\theta_{1,2}\leftarrow0 \,,\, \theta_{1,2}\leftarrow\frac{\pi}{2}\,,\, \theta_{1,2}\leftarrow\pi$ or $\theta_{1,2}\leftarrow\frac{3\pi}{2}$ at random with equal probability}
 \State \Return \texttt{angles}
\end{algorithmic}
\end{algorithm}

To see why Algorithm~\ref{alg-ap:clifford-passive} works, we need to notice that a random passive Clifford circuit is mapped to a random symplectic signed permutation, with the symplectic form given by $\Omega=\id_n \otimes iY$. This means that the permutation must permute the Majorana operators in pairs. That is, the Majorana operators $c_{2j-1}c_{2j}$, with $j=1,\dots,n$, must be permuted together, up to signs, by a permutation $\sigma$, to $c_{\sigma(2j-1)}, c_{\sigma(2j-1)\pm1}$ (where the $\pm 1$ depends on whether $\sigma(2j-1)$ is odd ($+$), or even ($-$)). Furthermore, within each pair the Majoranas can be permuted in four different ways,
\begin{equation} \label{eq-ap:symplectic-permutation}
    \begin{pmatrix}
        1 & 0 \\ 0 & 1
    \end{pmatrix}\,, \qquad 
    \begin{pmatrix}
        -1 & 0 \\ 0 & -1
    \end{pmatrix}\,, \qquad
    \begin{pmatrix}
        0 & -1 \\ 1 & 0
    \end{pmatrix}\,, \qquad
    \begin{pmatrix}
        0 & 1 \\ -1 & 0
    \end{pmatrix}\,.
\end{equation}
These permutations of Majorana operators clearly preserve the particle number, or equivalently, $\sum_{q=1}^n Z_q = -i\sum_{j=1}^n  c_{2j-1}c_{2j}$ (recall that Majoranas anti-commute, so that if we permute their order within a pair, a sign appears that must be compensated for). 

Now, in order to randomly sample a signed symplectic permutation of $2n$ items, we must first sample a random permutation of $n$ pairs, and then randomly choose between the four possibilities in Eq.~\eqref{eq-ap:symplectic-permutation} to shuffle the Majoranas within each pair. The previous explains the loop over $k$ even in  Algorithm~\ref{alg-ap:clifford-passive}, targeting two consecutive layers at a time, $O_{k-1}$ and $O_{k-2}$ (see Algorithm~\ref{alg-ap:clifford} for comparison). As before, we will use the $\pi/2$ angles to move the non-zero $\pm1$ entries of the permutation matrix to its final diagonal destination. But now we have to distinguish between two different cases, depending on whether the random integer $l$ is even or odd. If $l$ is even, this means that the $\pm1$ on the second column of the pair is moved first. This corresponds to the cases $\begin{pmatrix}
        1 & 0 \\ 0 & 1
    \end{pmatrix}$ or $\begin{pmatrix}
        -1 & 0 \\ 0 & -1
    \end{pmatrix}$. Here, the layer $O_{k-1}$ does not affect the $\pm1$ in the $(k-2)$-th row. Hence, both non-zero entries reach their final positions with the same sign, and we simply have to choose between adding $+\pi$ or not to both $\theta_{l,k}$ and $\theta_{l-1,k-1}$ in order to equally weight the probability of the two cases. Then, if $l$ is odd, the $\pm1$ on the first column of the pair is moved first. This corresponds to the cases $\begin{pmatrix}
        0 & -1 \\ 1 & 0
    \end{pmatrix}$ or $\begin{pmatrix}
        0 & 1 \\ -1 & 0
    \end{pmatrix}$. Here, the layer $O_{k-1}$ has a non-trivial effect on the $(k-2)$-th row,  moving its non-zero entry one step to the left and flipping its sign. Thus, once both $\pm1$ reach their final position, they do it with the same sign. As a consequence, one can equally weight the probability of the two cases by either adding $+\pi$ to the angle $\theta_{l,k}$ or leaving it untouched (as here adding $+\pi$ to the angle $\theta_{l,k}$ will flip the sign of the $(k-2)$-th row as well, in contrast to the previous situation, when $l$ is even). Finally, randomly setting $\theta_{1,2}\in\{0,\frac{\pi}{2},\pi,\frac{3\pi}{2}\}$ is equivalent to choosing at random between the four cases in Eq.~\eqref{eq-ap:symplectic-permutation}, whenever the pair of Majoranas is not permuted with any other pair.

\section{A method for exact commutant computations}
\label{ap:Howe}

We here provide a general method to compute the dimension of the $t$-th fold commutant of arbitrary compact reductive Lie algebra representations.

Let $\mf{g}$ be a compact reductive Lie algebra, with corresponding Lie group $G=e^\mf{g}$, and consider a $G$-module $\HC$. Suppose $\HC$ decomposes into $G$-irreps as
\begin{equation}
\HC\cong \bigoplus_\lm \bigoplus_{i=1}^{m_\lm^\HC} \HC_{\lm,i}\,,
\end{equation}
where $\lm$ labels the irreducible type (i.e., $\lm$ is the highest-weight) and $i$ is a multiplcity index: for all $i\in[m_\lm^\HC]$, $\HC_{\lm,i}\cong V^\lm_G$.  The $t$-th fold commutant -- the subspace of $G$-invariant operators in $\LC(\HC\ot{t})$-- can be constructed by 
\begin{enumerate}
  \item Decomposing \(\HC^{\otimes t}\) into irreps,
  \item Taking the tensor product \(\HC^{\otimes t}\otimes (\HC^*)^{\otimes t} \cong \mathcal{L}(\HC^{\otimes t})\), 
  \item Applying Schur's Lemma.
\end{enumerate}
Indeed, if
\begin{equation}
\HC^{\otimes t} = \bigoplus_\lambda \bigoplus_{i=1}^{m_\lambda^{\HC^{\otimes t}}} \TC_{\lambda,i},
\end{equation}
then, as we will see in Sec.~\ref{subs_schur}, by Schur's Lemma 
\begin{equation}
\dim\big(\mathcal{L}(\HC^{\otimes t})^G\big) = \sum_{\lambda} \Bigl(m_\lambda^{\HC^{\otimes t}}\Bigr)^2.
\end{equation}
We now detail these steps. 

\subsection{Irreps in the tensor powers}\label{subs_tensor}

Let us begin with the decomposition of the $G$-module
\begin{align}
\TC &\equiv \HC\ot{t} \nonumber\\
&= \Big( \bigoplus_\lm \bigoplus_{i=1}^{m_\lm^\HC} \HC_{\lm,i} \Big)\ot{t}\nonumber \\
&=  \bigoplus_{\lm_1,\cdots,\lm_t} \bigoplus_{i_1,\cdots,i_t=1}^{m_\lm} \HC_{\lm_1,i_1} \otimes \cdots \otimes \HC_{\lm_t,i_t}\,.
\end{align}
Given two $G$-irreps of type $\lm$ and $\mu$, their tensor decomposes as
\begin{equation}\label{eq_tensor_irreps}
V^\lm_G \otimes V^\mu_G \cong \bigoplus_{\nu} \bigoplus_{j=1}^{m_j^{\lm\otimes \mu}} V^\nu_G\,,
\end{equation}
where $m_j^{\lm\otimes \mu}$ are arising multiplicity coefficients, e.g., Littlewood-Richardson coefficients when $G=\mbb{U}(n)$. These multiplicity coefficients can be computed by writing the characters of $V^\lm$ and $V^\mu$ as elements of the weight lattice (more on this in Subsection \ref{subs_computation}). Iterating such computations we can decompose each of the terms in $\TC$,
\begin{align}
\TC^{(\lm_i,i_1),\cdots,(\lm_t,i_t)} &\equiv\HC_{\lm_1,i_1} \otimes \cdots \otimes \HC_{\lm_t,i_t}\nonumber \\
&\cong \bigoplus_\lm \bigoplus_{i=1}^{m_\lm^{\lm_1\otimes \cdots\otimes\lm_t}} \TC_{\lm,i}\,,
\end{align}
arriving at
\begin{equation}
\TC = \bigoplus_\lm \bigoplus_{i=1}^{m_\lm^{\HC\ot{t}}} \TC_{\lm,i}\,.
\end{equation}

\subsection{Schur's Lemma and the decomposition of the space of $G$-invariant $t$-fold tensor operators}\label{subs_schur}

Let $\MC$ and $\NC$ be irreducible $G$-modules. Schur's Lemma (see, for example, page 7 in Fulton and Harris~\cite{fulton1991representation}) states that 
\begin{equation}
{\rm Hom}(\MC,\NC)^G = \begin{cases}
    {\rm span} \{A\}\quad \text{ if $\MC\cong\NC$}\\
    0\quad \text{ else}
\end{cases}\,,
\end{equation}
where $A:\MC\arr\NC$ is the identity isomorphism between the modules.
We can use Schur's lemma together with the decomposition of $\HC\ot{t}$ to find the $G$-invariant subspace of $\LC(\HC\ot{t})$. That is,
\begin{align}
\LC(\HC\ot{t})^G &= {\rm Hom}(\TC,\TC)^G\\
&= \Big( {\rm Hom} \Big( \bigoplus_{\lm,i} \TC_{\lm,i}, \bigoplus_{\lm',j} \TC_{\lm',j} \Big) \Big)^G \\
&= \bigoplus_{\lm} \bigoplus_{i,j=1}^{m_\lm^{\HC\ot{t}}} {\rm Hom}(\TC_{\lm,i},\TC_{\lm,j})^G\,,
\end{align}
where we used the fact that the only $G$-invariant Homs appear iff the modules are isomorphic, in which case $(\TC_{\lm,i},\TC_{\lm,j})^G$ is the identity homomorphism between the modules. Finally, the dimension of the commutant is
\begin{equation}
\dim(\LC(\HC\ot{t})^G) = \sum_{\lm} \Big(m_\lm^{\HC\ot{t}} \Big)^2\,.
\end{equation}

We have used this formula to compute the
$t$-fold commutant dimensions for active and passive FLO presented in Fig.~\ref{fig:brick-haar}.

\subsection{Multiplicity computation via the Weyl Character Ring}\label{subs_computation}

In this section we show how to compute the multiplicity coefficients appearing the the tensor product of irreducible representations (see e.g., Eq.~\eqref{eq_tensor_irreps}).
Let $P$ be the \textit{weight lattice} of some semisimple Lie algebra $\mf{g}$~\footnote{For simplicity we assume $\mf{g}$ to be semisimple; the case of nontrivial center follows straightforwardly.}
\begin{equation}
P = \Big\{ \sum_{i=1}^{r} n_i \w_i\quad|\quad \text{$n_i\in \mbb{Z}$} \Big\}\,,
\end{equation}
where $r={\rm rank}[\mf{g}]$ and $\w_i$ are the fundamental weights of $\mf{g}$. The elements of $P$ are the weights that could occur in some $\mf{g}$-representation --hence the name weight lattice-- although a given representation typically uses only a subset of $P$. If we restrict to elements of $P$ of the form $ \sum_{i=1}^{r} n_i \w_i$ with all $n_i$ non-negative, those weights are called \textit{dominant} weights and label the different irreducible representations of $\mf{g}$. 

The weight lattice $P$ is an abelian group under addition: $a*b=a+b \in P$ for all $a,b\in P$. Given a group $G$, its group algebra is the vector space (under some field of choice $\mbb{F}$) with basis given by elements of $G$ and algebra product given by the group product: $\mbb{F}[G]= {\rm span}_{\mbb{F}} \{\ket{g}\}_{g\in G}$\footnote{Many references use $e^g$ instead of $\ket{g}$ for the formal basis elements of the group algebra.}. The \textit{Weyl Character Ring} (WCR) is the group algebra (over $\mbb{Z}$) of the weight lattice
\begin{equation}
    \mbb{Z}[P] = \Big\{ \sum_a c_a \ket{a} \quad c_a\in\mbb{Z}, a \in P\Big\}\,.
\end{equation}
Thus, $\mbb{Z}[P]$ is a ring (in fact, a commutative ring) whose structure encodes the additive structure of $P$ in an algebraic way. Crucially, and maybe not surprisingly given its name, characters are elements of the WCR: Formal sums of the weights appearing in the given representation, with coefficients being the weight multiplicity.

Because characters are represented as formal sums in the group algebra $\mbb{Z}[P]$, the multiplication of two characters corresponds to the convolution of their weight multiplicity functions. Specifically, if~\footnote{This decomposition can be obtained using the \emph{Brauer--Klimyk method}~\cite{humphreys2012introduction}, also known as the \textit{Racah-Speiser algorithm}~\cite{racah2006group} in physics literature. The method uses the Weyl character formula and a recursive procedure to extract weight multiplicities. For instance, it starts with the highest weight (known to have multiplicity 1) and systematically subtracts contributions from lower weights until the full expansion is obtained.}
\begin{equation}
\ket{\chi_\lm} = \sum_a c_a^\lm \ket{a} \quad  \ket{\chi_\mu} = \sum_b c_b^\mu \ket{b}\,,
\end{equation}
then their product is
\begin{equation}
\chi_\lm * \chi_\mu = \sum_\nu \Big( \sum_{a+b=\nu} c_a^\lm c_b^\mu \Big) \ket{\nu}\,.
\end{equation}
This convolution over the weight lattice computes the multiplicity of each weight $\nu$ in the tensor product $V^\lm\otimes V^\mu$. Subsequently, by collecting terms corresponding to dominant weights, one obtains the decomposition into irreducible components, with the coefficients 
\begin{equation}
    m_\nu^{\lm\otimes \nu} = \sum_{a+b=\nu} c^\lm_a c^\mu_b
\end{equation}
representing the multiplicities. A package implementing WCR methods is available in SageMath \cite{sage}. For further details on the WCR, we recommend Chapter VI in Knapp's book~\cite{knapp2013lie}.

\subsection{Active and Passive FLO commutants}
We note the method presented here works for essentially \textit{any} rep of any reductive Lie algebra. We now specialize to the cases of active and passive FLO.

Let $\HC\cong (\mbb{C}^2)\tn$
be the state space of $n$ fermionic modes, and let $R_{\rm act}:\mbb{SO}(2n)\arr \mbb{U}(\HC)$ and $R_{\rm pass}:\mbb{U}(n)\arr \mbb{U}(\HC)$ be the group homomorphisms whose image correspond to active and passive FLO, respectively. As a $G$-module, the state space decomposes as
\begin{align}
&\HC_{\rm act} \cong V^{\w_{n-1}}_{\mbb{SO}(2n)}\oplus V^{\w_{n}}_{\mbb{SO}(2n)}\,,\\
&\HC_{\rm pass} \cong \bigoplus_{k=0}^n V^{[1^k]}_{\mbb{U}(n)}\,,
\end{align}
where $\w_{n-1}=(1/2,\cdots,-1/2)$ and $\w_n =(1/2,\cdots,1/2)$ are the so-called spinor reps of $\SO(2n)$, while $\lm=[1^k]$ corresponds to $\Lambda^{k}(\mbb{C}^d) \subset (\mbb{C}^d)^{\otimes k}$ the completely antisymmetric subspace of the $k$-fold tensor of the standard $\mbb{U}(n)$-rep.

Having  the decomposition of $\HC$ into $G$-irreps, we are ready to apply the method described at the beginning of this section: First decompose $\HC\ot{t}$ by iterating the multiplicity computations described in Sec.\ref{subs_computation}, then apply Schur's lemma. We note that in this particular cases of $\HC_{\rm act}$ and $\HC_{\rm pass}$, the decomposition of $\HC\ot{t}$ into irreps could alternatively (and much more efficiently) be obtained using Howe duality~\cite{nazarov2024skew}.

\clearpage
\end{document}